\documentclass[reqno]{amsart}
\usepackage[utf8]{inputenc}
\usepackage{xcolor}
\usepackage{url}
\usepackage{enumitem}
\usepackage{amstext}
\usepackage{amsthm}
\usepackage{amssymb}
\usepackage{graphicx}
\usepackage{esint}
\PassOptionsToPackage{normalem}{ulem}
\usepackage{ulem}
\usepackage[unicode=true,pdfusetitle,
 bookmarks=true,bookmarksnumbered=false,bookmarksopen=false,
 breaklinks=false,pdfborder={0 0 0},pdfborderstyle={},backref=false,colorlinks=true]
 {hyperref}

\makeatletter

\DeclareTextSymbolDefault{\textquotedbl}{T1}
\providecommand{\tabularnewline}{\\}
\providecolor{lyxadded}{rgb}{0,0,1}
\providecolor{lyxdeleted}{rgb}{1,0,0}
\DeclareRobustCommand{\lyxsout}[1]{\ifx\\#1\else\sout{#1}\fi}

\numberwithin{equation}{section}
\numberwithin{figure}{section}
\theoremstyle{plain}
\newtheorem{thm}{\protect\theoremname}
\theoremstyle{definition}
\newtheorem{defn}[thm]{\protect\definitionname}
\theoremstyle{plain}
\newtheorem{lem}[thm]{\protect\lemmaname}
\theoremstyle{remark}
\newtheorem{rem}[thm]{\protect\remarkname}
\theoremstyle{plain}
\newtheorem{cor}[thm]{\protect\corollaryname}
\theoremstyle{definition}
\newtheorem{example}[thm]{\protect\examplename}

\@ifundefined{date}{}{\date{}}
\usepackage{pdfsync}
\usepackage{accents}
\usepackage{graphics}
\usepackage{graphicx}
\usepackage{epstopdf}
\usepackage[all]{xy}
\usepackage{amscd}
\usepackage{enumitem}
\usepackage{mathtools}
\usepackage[utf8]{inputenc}
\usepackage{bbold}
\setlist[enumerate]{leftmargin=*,label=(\roman*),align=left}

\newcommand{\xyR}[1]{ \makeatletter
\xydef@\xymatrixrowsep@{#1} \makeatother} % end of \xyR
\newcommand{\xyC}[1]{ \makeatletter
\xydef@\xymatrixcolsep@{#1} \makeatother} % end of \xyC
\entrymodifiers={++[ ][F]} % to have more space around objects in diagrams

\newcommand{\ra}{\longrightarrow}
\newcommand{\field}[1]{\mathbb{#1}}
\newcommand{\R}{\field{R}} % reals
\newcommand{\N}{\field{N}} % naturals
\newcommand{\Z}{\ensuremath{\mathbb{Z}}} % integers
\newcommand{\D}{\mathcal{D}}

\newcommand{\eps}{\varepsilon} % for the sake of brevity only
\renewcommand{\phi}{\varphi}
\newcommand{\diff}[1]{\ifmmode\mathchoice{\hbox{\rm d}#1}  % displaystyle
 {\hbox{\rm d}#1}  % normal 
 {\scalebox{0.75}{$\hbox{\rm d}#1$}}  % scriptstyle 
 {\scalebox{0.35}{$\hbox{\rm d}#1$}}  % scriptscriptstyle
 \fi} % dt,dx,... for integrals
\newcommand{\abs}[2][\empty]{\ifx#1\empty\left|#2\right|%
\else#1\vert #2 #1\vert\fi}% optional arg=size
\newcommand{\Coo}{\mbox{\ensuremath{\mathcal{C}}}^{\infty}} % C infinity
\DeclareMathOperator{\Set}{{\bf Set}} % category of Sets
\newcommand{\st}[1]{{#1^\circ}} % standard part map
\newcommand{\Rtil}{\widetilde \R} % real Colombeau generalized number
\newcommand{\gs}{\mathcal{G}^s} % special Colombeau algebra
\newcommand{\ns}{\mathcal{N}^s} % negligible nets for the special CA
\newcommand{\sint}[1]{\langle#1\rangle} % strong internal set
\newcommand{\Eball}{B^{{\scriptscriptstyle \text{\rm E}}}} % ordinary Euclidean ball
\newcommand{\csp}[1]{{\text{\rm c}}({#1})}
\newcommand{\fcmp}{\Subset_{\text{f}}}
\newcommand{\frontRise}[2]{\ifmmode\mathchoice{{\vphantom{#1}}^{\scalebox{0.6}{$#2$}}}  % displaystyle
 {{\vphantom{#1}}^{\scalebox{0.56}{$#2$}}}  % normal 
 {{\vphantom{#1}}^{\scalebox{0.47}{$#2$}}}  % scriptstyle 
 {{\vphantom{#1}}^{\scalebox{0.35}{$#2$}}}\fi} % scriptscriptstyle 
\newcommand{\RC}[1]{\frontRise{\R}{#1}\Rtil}
\newcommand{\frontRiseDown}[3]{\ifmmode\mathchoice{{\vphantom{#1}}^{\scalebox{0.6}{$#2$}}_{\scalebox{0.6}{$#3$}}}  % displaystyle
 {{\vphantom{#1}}^{\scalebox{0.56}{$#2$}}_{\scalebox{0.56}{$#3$}}}  % normal 
 {{\vphantom{#1}}^{\scalebox{0.47}{$#2$}}_{\scalebox{0.47}{$#3$}}}  % scriptstyle 
 {{\vphantom{#1}}^{\scalebox{0.35}{$#2$}}_{\scalebox{0.35}{$#3$}}}\fi} % scriptscriptstyle 

\newcommand{\rcrho}{\RC{\rho}}
\newcommand{\rti}{\RC{\rho}}
\newcommand{\rccrho}{\frontRise{\mathbb{C}}{\rho}\widetilde{\mathbb{C}}}
\newcommand{\gsf}{\frontRise{\mathcal{G}}{\rho}\mathcal{GC}^{\infty}}
\newcommand{\gckf}{\frontRise{\mathcal{G}}{\rho}\mathcal{GC}^{k}}
\newcommand{\gcf}[1]{\frontRise{\mathcal{G}}{\rho}\mathcal{GC}^{#1}}
\newcommand{\Dgsf}{\frontRise{\mathcal{D}}{\rho}\mathcal{GD}}

\newcommand{\subzero}{\subseteq_{0}}
\newcommand{\dom}[1]{\text{dom}(#1)}

\newcommand{\sbpt}[1]{#1_{\text{\rm s}}}

\newcommand{\norm}[1]{\Vert{#1}\Vert}

\newcommand{\ptind}{\displaystyle \mathop {\ldots\ldots\,}} % marks with smth over. Usage: \ptind^{...}
\newcommand{\supp}{\mbox{supp}}
\newcommand{\cinfty}{{\mathcal C}^\infty}

\newcommand{\comp}{\Subset}
\newcommand{\esm}{{\mathcal E}_M}

\newcommand{\Om}{\Omega}

\newcommand{\sse}{\subseteq}

\makeatother

\providecommand{\corollaryname}{Corollary}
\providecommand{\definitionname}{Definition}
\providecommand{\examplename}{Example}
\providecommand{\lemmaname}{Lemma}
\providecommand{\remarkname}{Remark}
\providecommand{\theoremname}{Theorem}

 \hypersetup{
	breaklinks=true,   % splits links across lines
	colorlinks=true,   % displays links as colored text instead of blocks
	pdfusetitle=true,  % \title and \author values into pdf metadata
}

\begin{document}
\title{Infinitesimal and infinite numbers in applied mathematics}
\author{Aleksandr Bryzgalov \and Kevin Islami \and Paolo Giordano}
\thanks{A.~Bryzgalov had been supported by grant P34113 of the Austrian Science
Fund FWF.}
\address{\textsc{The Martin Luther University Halle-Wittenberg (MLU), Faculty
of Medicine, Institute of Medical Epidemiology, Biometry and Informatics
(IMEBI), Magdeburger Straße 8, 06112 Halle (Saale), Germany}}
\email{aleksandr.bryzgalov@uk-halle.de}
\thanks{K.~Islami has been supported by grants P30407, P34113 and P33538
of the Austrian Science Fund FWF.}
\address{\textsc{Faculty of Mathematics, University of Vienna, Austria, Oskar-Morgenstern-Platz
1, 1090 Wien, Austria}}
\email{kevin.islami@univie.ac.at}
\thanks{P.~Giordano has been supported by grants P30407, P33538, P33945 and
P34113 of the Austrian Science Fund FWF.}
\address{\textsc{Faculty of Mathematics, University of Vienna, Austria, Oskar-Morgenstern-Platz
1, 1090 Wien, Austria}}
\email{paolo.giordano@univie.ac.at}
\begin{abstract}
The need to describe abrupt changes or response of nonlinear systems
to impulsive stimuli is ubiquitous in applications. Also the informal
use of infinitesimal and infinite quantities is still a method used
to construct idealized but tractable models within the famous J.~von
Neumann \emph{reasonably wide area} of applicability. We review the
theory of \emph{generalized smooth functions} as a candidate to address
both these needs: a rigorous but simple language of infinitesimal
and infinite quantities, and the possibility to deal with continuous
and generalized function as if they were smooth maps: with pointwise
values, free composition and hence nonlinear operations, all the classical
theorems of calculus, a good integration theory, and new existence
results for differential equations. We exemplify the applications
of this theory through several models of singular dynamical systems:
deduction of the heat and wave equations extended to generalized functions,
a singular variable length pendulum wrapping on a parallelepiped,
the oscillation of a pendulum damped by different media, a nonlinear
stress-strain model of steel, singular Lagrangians as used in optics,
and some examples from quantum mechanics.
\end{abstract}

\maketitle
\tableofcontents{}

\section{\label{sec:Intro}Introduction: informal uses of infinitesimals and infinities in applied mathematics}

Even if infinitesimal numbers have been banished by modern mathematics,
several physicists, engineers and mathematicians still profitably
continue to use them. Usually, this is in order to simplify calculations,
to construct idealized but notwithstanding interesting models of physical
systems, or to relate different parts of physics, such as in passing
from quantum to classical mechanics if $\hbar$ is infinitesimal.
An authoritative example in this direction is given by A.~Einstein
when he writes
\begin{equation}
\frac{1}{\sqrt{1-{\displaystyle \frac{v^{2}}{c^{2}}}}}=1+\frac{v^{2}}{2c^{2}}\qquad\qquad\sqrt{1-h_{44}(x)}=1-\frac{1}{2}h_{44}(x)\label{eq:EinsteinInfinitesimal}
\end{equation}
with explicit use of infinitesimals $v/c\ll1$ or $h_{44}(x)\ll1$
such that e.g.~$h_{44}(x)^{2}=0$. More generally, in \cite{Ein}
Einstein writes the formula (using the equality sign and not the approximate
equality sign $\simeq$) 
\begin{equation}
f(x,t+\tau)=f(x,t)+\tau\cdot\frac{\partial f}{\partial t}(x,t)\label{eq:EinsteinDerivationFormula}
\end{equation}
justifying it with the words ``\emph{since $\tau$ is very small}'';
note that \eqref{eq:EinsteinInfinitesimal} are a particular case
of the general \eqref{eq:EinsteinDerivationFormula}. Also P.A.M.~Dirac
in \cite{Dir} writes an analogous equality studying the Newtonian
approximation in general relativity.

A certain degree of inconsistency appears also at the level of elementary
topics, e.g.~in the deduction of the wave and heat equations, see
e.g.~\cite{Vla71}. For example, if $u(x,t)$ is the string displacement,
then formula \eqref{eq:EinsteinDerivationFormula} is once again used
e.g.~``to ignore magnitudes of order greater than $\frac{\partial u}{\partial x}$''.
This means that we need to have $\left(\frac{\partial u}{\partial x}\right)^{2}=0$
to arrive at the final equation with an equality sign and not with
some kind of approximation $\simeq$. But then the length of the string
becomes
\[
L=\int_{a}^{b}\sqrt{1+\left[\frac{\partial u}{\partial x}(x,t)\right]^{2}}\,\diff{x}=b-a,
\]
and its clear that this necessarily yields that the function $u$
is constant. It clearly does not really help to use $\simeq$ when
we have a contradiction, but then to change it into $=$ when we need
the final equation. It is for this type of motivations that V.I.~Arnol’d
in \cite{Arn90} wrote: \emph{Nowadays, when teaching analysis, it
is not very popular to talk about infinitesimal quantities. Consequently
present-day students are not fully in command of this language. Nevertheless,
it is still necessary to have command of it}.

A similar, but sometimes more troublesome, practice concerns the use
of infinite numbers. A typical example is given by Heisenberg's uncertainty
principle
\begin{equation}
\Delta p_{x}\Delta x\geq\frac{\hbar}{2}.\label{eq:Heisenberg}
\end{equation}
It is frequently informally argued that if the position $x$ is measured
by a Dirac delta, then $\Delta x\approx0$ is infinitesimal; thereby,
\eqref{eq:Heisenberg} necessarily implies that $\Delta p_{x}$ must
be an infinite number.

Another classical example of informal use of infinite numbers concerns
Schwartz distributions and their point values. Many relevant physical
systems are in fact described by singular Hamiltonians. Among them,
we can e.g.~list:
\begin{enumerate}
\item Non smooth classical mechanics. For example, a classical particle
(or a high-frequency wave) moving through discontinuous media containing
barriers or interfaces where the Hamiltonian is discontinuous: see
e.g.~\cite{Bro1,Bro2,daSMa,TeKeTo,JiWuHu08,SaPaMa,SiJe13} and references
therein.
\item Discontinuous Lagrangian in geometrical optics: see e.g.~\cite{LaGhTh11}.
\item Nonlinear deformation in continuum mechanics, which includes non differentiable
stress-strain relations: see e.g.~\cite{UgFe11,ChCh11}.
\item The use of infinite quantities in quantum mechanics. An elementary
example is given by the solution of the stationary Schrödinger equation
for an infinite rectangular potential well (a case that cannot be
formalized using Schwartz distributions, see e.g.~\cite{GaPa90}).
\end{enumerate}
This type of problems is hence widely studied from the mathematical
point of view (see e.g.~\cite{CrLi83,Lim89,LeMoSj91,Tan93,Kunzle96,KuObStVi,LiHuZh14}),
even if the presented solutions are not general and hold only in special
conditions. In this sense, the fact that J.D.~Marsden's works \cite{Mar68,Mar69}
did not start a consolidated research thread to study singular Hamiltonian
mechanics using Schwartz distributions, can be considered as a clue
that the classical distributional framework is not well suited to
face this problem in general terms.

A related problem concerns nonlinear operations on Schwartz distributions,
which can be simply presented as follows. Let $A$ be an associative
and commutative algebra endowed with a derivation (satisfying the
Leibniz rule). Then any element $H$ of $A$ such that $H\cdot H=H$
is necessarily a constant, that is, $H'=0$. Indeed, $(H^{2})'=2HH'$
and $(H^{3})'=3H^{2}H'$ . Now $H=H^{2}=H^{3}$ , so this implies
$2HH'=H'=3HH'$ . Therefore, $HH'=0$ and hence also $H'=2HH'=0$.
Even worse, we also recall that Dirac in \cite{Dir2} uses terms of
the form $\sqrt{\delta}$ in his proposal for the foundation of quantum
mechanics.

There are obviously many possibilities to formalize this kind of intuitive
reasonings, obtaining a more or less good dialectic between informal
and formal thinking, and nowadays there are indeed several rigorous
theories of infinitesimals, infinities and generalized functions.
Concerning the notion of infinitesimal, we can distinguish two definitions:
in the first one we have at least a ring $R$ containing the real
field $\R$ and infinitesimals are elements $h\in R$ such that $-r<h<r$
for every positive standard real $r\in\R_{>0}$. The second type of
infinitesimal is defined using some algebraic property of nilpotency,
i.e. $h^{n}=0$ for some natural number $n\in\N$ (therefore, in this
case $h$ cannot be trivially invertible and we cannot form infinities
as reciprocal of infinitesimals). For some ring $R$ these definitions
can coincide, but anyway they lead, of course, only to the trivial
infinitesimal $h=0$ if $R=\R$ is the real field.

Mathematical theories of infinitesimals can also be classified as
belonging to two different classes. In the first one, we have theories
needing a certain amount of non trivial results of mathematical logic,
whereas in the second one we have attempts to define sufficiently
strong theories of infinitesimals without the use of non trivial results
of mathematical logic. In the first class, we can list \emph{nonstandard
analysis} and \emph{synthetic differential geometry} (also called
\emph{smooth infinitesimal analysis}, see e.g. \cite{Bel,Koc,Lav,MoRe}),
in the second one we have, e.g., \emph{Weil functors} (see \cite{KrMi2}),
\emph{Levi-Civita field} (see \cite{ShBe}), \emph{surreal numbers}
(see \cite{Con}), \emph{Fermat reals} (see \cite{Gio10}), \emph{Colombeau's
generalized numbers} (see \cite{C1} and \cite{Ehr,Kap} for a general
survey). More precisely, we can say that to work both in nonstandard
analysis and in synthetic differential geometry, one needs a formal
control stronger than the one used in ``standard mathematics''.
Indeed, to use nonstandard analysis one has to be able to formally
write sentences in order to apply the transfer theorem. Whereas synthetic
differential geometry does not admit models in classical logic, but
in intuitionistic logic only, and hence we have to be sure that in
our proofs there is no use of the law of the excluded middle, or e.g.~of
the classical part of De Morgan's law or of some form of the axiom
of choice or of the implication of double negation toward affirmation
and any other logical principle which do not hold in intuitionistic
logic. Physicists, engineers, but also the greatest part of mathematicians
are not used to have this strong formal control in their work, and
it is for this reason that there are attempts to present both nonstandard
analysis and synthetic differential geometry reducing as much as possible
the necessary formal control, even if at some level this is technically
impossible (see e.g. \cite{Kei,Hen}, and \cite{BeDN1,BeDN2} for
nonstandard analysis; \cite{Bel} and \cite{Lav} for synthetic differential
geometry).

On the other hand, nonstandard analysis is surely the best known theory
of invertible infinitesimals with results in several areas of mathematics
and its applications, see e.g. \cite{AFHL}. Synthetic differential
geometry is a theory of nilpotent infinitesimals with non trivial
results of differential geometry in infinite dimensional spaces.

Concerning mathematical theories of generalized functions, the difficulties
stem\-med from dealing with the lacking of well-posedness in PDE
initial value problems led to a zoo of spaces of generalized functions.
In an incomplete list we can mention: Schwartz distributions, Colombeau
generalized functions, ultradistributions, hyperfunctions, nonstandard
theory of Colombeau generalized functions, ultrafunctions, etc. See
e.g.~\cite{HoPi05} for a survey, and the International Conference
on Generalized Functions series, e.g.~\url{https://ps-mathematik.univie.ac.at/e/index.php?event=GF2022}.

Unfortunately, there is a certain lacking of dialog between the most
used theory of generalized functions, i.e.~Schwartz distributions,
and the actual use of generalized functions in physics and engineering,
where e.g.~point values and nonlinear operations are frequently needed,
see e.g.~\cite{C1}.

In the present paper, we describe the main results of \emph{generalized
smooth functions} (GSF) theory and some of its applications in applied
mathematics. GSF are an extension of classical distribution theory
and of Colombeau theory, which makes it possible to model nonlinear
singular problems, while at the same time sharing a number of fundamental
properties with ordinary smooth functions, such as the closure with
respect to composition and several non trivial classical theorems
of the calculus, see \cite{Gio-Kun-Ver15,Gio-Kun-Ver19,Gio-Kun18a,Gio-Kun16,LuGi20b,GiLu20a}.
One could describe GSF as a methodological restoration of Cauchy-Dirac's
original conception of generalized function (GF), see \cite{Laug89}.
In essence, the idea of Cauchy and Dirac (but also of Poisson, Kirchhoff,
Helmholtz, Kelvin and Heaviside) was to view generalized functions
as suitable types of smooth set-theoretical maps obtained from ordinary
smooth maps depending on suitable infinitesimal or infinite parameters.

The calculus of GSF is closely related to classical analysis, in fact:
\begin{enumerate}
\item GSF are set-theoretical maps defined on, and attaining values in the
non-Ar\-chi\-me\-de\-an ring $\rcrho$ of Robinson-Colombeau.
Therefore, in $\rcrho$ we have infinitesimals, infinities and also
a suitable language of nilpotent infinitesimals, see Sec.~\ref{sec:Numbers},
Sec.~\ref{sec:Functions}.
\item GSF include all Colombeau generalized functions and hence also all
Schwartz distributions, see Sec.~\ref{sec:Functions}.
\item They allow nonlinear operations on GF and to compose them unrestrictedly,
so that terms such as $\delta^{2}(x)$ or even $\delta(\delta(x))$
are possible, see Sec.~\ref{sec:Functions}.
\item GSF allow us to prove a number of analogues of theorems of classical
analysis: e.g., mean value theorem, intermediate value theorem, extreme
value theorem, Taylor's theorem, local and global inverse function
theorem, integrals via primitives, multidimensional integrals, theory
of compactly supported GF. Therefore, this approach to GF results
in a flexible and rich framework which allows both the formalization
of calculations appearing in physics and engineering and the development
of new applications in mathematics and mathematical physics. Some
of these results are presented in Sec.~\ref{sec:Functions}.
\item Several results of the classical theory of calculus of variations
and optimal control can be developed for GSF: the fundamental lemma,
second variation and minimizers, necessary Legendre condition, Jacobi
fields, conjugate points and Jacobi’s theorem, Noether’s theorem,
see \cite{LeLuGi17,FGBL}.
\item The closure with respect to composition leads to a solution concept
of differential equations close to the classical one. In GSF theory,
we have a non-Archimedean version of the Banach fixed point theorem,
a Picard-Lindelöf theorem for both ODE and PDE, results about the
maximal set of existence, Gronwall theorem, flux properties, continuous
dependence on initial conditions, full compatibility with classical
smooth solutions, etc., see Sec.~\ref{sec:Differential-equations-1}.
\end{enumerate}
As we will see in Sec.~\ref{sec:Formal-deductions} and Sec.~\ref{sec:Examples-of-applications},
using GSF theory, we have a rigorous theory of infinitesimal and infinite
numbers that can be used to develop mathematical models of physical
problems. On the other hand, it is also a flexible theory of GF that
can be used to model situations with singular (non smooth) physical
quantities. One of the main aim of this paper is to show that within
GSF theory several informal calculations of physics or engineering
now become perfectly rigorous without detaching too much from the
original deduction.

The structure of the paper is as follows: in Sec.~\ref{sec:Numbers},
we introduce our new ring of scalars $\rcrho$, the ring of Robinson-Colombeau.
In Sec.~\ref{sec:Functions}, we define GSF as suitable set-theoretical
functions defined and valued in the new ring of scalars; we will also
see the relationships with Colombeau GF and hence with Schwartz distributions.
In Sec.~\ref{sec:Differential-equations-1}, we see the Picard-Lindelöf
theorem for singular nonlinear ODE involving GSF. In Sec\@.~\ref{sec:Formal-deductions},
we see how to transform the classical deductions of the wave and heat
equations into formal mathematical theorems whose scope includes GSF.
Finally, in Sec.~\ref{sec:Examples-of-applications} we see applications
to non-smooth mechanics, an empirical non-linear stress-strain model
for steel, some applications in optics with discontinuous Lagrangians,
how to see the classical finite and infinite potential wells of QM
within GSF theory.

The paper is a review of GSF theory, so it is self-contained in the
sense that it contains all the statements required for the proofs
of Sec.~\ref{sec:Formal-deductions} and Sec.~\ref{sec:Examples-of-applications}.
We also introduce clear intuitions about the new mathematical objects
of this theory and references for the complete proofs. Therefore,
to understand this paper, only a basic knowledge of distribution theory
is needed.

\section{\label{sec:Numbers}Numbers: The ring of Robinson-Colombeau}

In this first section, we introduce our non-Archimedean ring of scalars.
More in depth details about the basic notions and the omitted proofs
as well can be found in \cite{Gio-Kun-Ver19,Gio-Kun-Ver15}.
\begin{defn}
\label{def:RCGN}Let $I:=(0,1]$ and $\rho=(\rho_{\eps})\in(0,1]^{I}$
be a net such that $(\rho_{\eps})\to0$ as $\eps\to0^{+}$ (in the
following, such a net will be called a \emph{gauge}), then
\begin{enumerate}
\item $\mathcal{I}(\rho):=\left\{ (\rho_{\eps}^{-a})\mid a\in\R_{>0}\right\} $
is called the \emph{asymptotic gauge} generated by $\rho$.
\item If $\mathcal{P}(\eps)$ is a property of $\eps\in I$, we use the
notation $\forall^{0}\eps:\,\mathcal{P}(\eps)$ to denote $\exists\eps_{0}\in I\,\forall\eps\in(0,\eps_{0}]:\,\mathcal{P}(\eps)$.
We can read $\forall^{0}\eps$ as: \emph{``for $\eps$ small}''.
\item We say that a net $(x_{\eps})\in\R^{I}$ \emph{is $\rho$-moderate},
and we write $(x_{\eps})\in\R_{\rho}$, if 
\begin{equation}
\exists(J_{\eps})\in\mathcal{I}(\rho):\ x_{\eps}=O(J_{\eps})\text{ as }\eps\to0^{+},\label{eq:yourlabel}
\end{equation}
i.e., if 
\[
\exists N\in\N\,\forall^{0}\eps:\ |x_{\eps}|\le\rho_{\eps}^{-N}.
\]
\item Let $(x_{\eps})$, $(y_{\eps})\in\R^{I}$, then we say that $(x_{\eps})\sim_{\rho}(y_{\eps})$
if 
\[
\forall(J_{\eps})\in\mathcal{I}(\rho):\ x_{\eps}=y_{\eps}+O(J_{\eps}^{-1})\text{ as }\eps\to0^{+},
\]
that is if 
\begin{equation}
\forall n\in\N\,\forall^{0}\eps:\ |x_{\eps}-y_{\eps}|\le\rho_{\eps}^{n}.\label{eq:rhoNegligible}
\end{equation}
This is a congruence relation on the ring $\R_{\rho}$ of moderate
nets with respect to pointwise operations, and we can hence define
\begin{equation}
\RC{\rho}:=\R_{\rho}/\sim_{\rho},\label{eq:RCN}
\end{equation}
which we call \emph{Robinson-Colombeau ring of generalized numbers}.
\end{enumerate}
\end{defn}

\noindent This name is justified by \cite{Rob73,C1}: Indeed, in \cite{Rob73}
A.~Robinson introduced the notion of moderate and negligible nets
depending on an arbitrary fixed infinitesimal $\rho$ (in the framework
of nonstandard analysis); independently, J.F.~Colombeau, cf.~e.g.~\cite{C1}
and references therein, studied the same concepts without using nonstandard
analysis, but considering only the particular infinitesimal $\rho_{\eps}=\eps$.
Equivalence classes of the quotient ring \eqref{eq:RCN} are simply
denoted with $[x_{\eps}]:=[(x_{\eps})_{\eps}]_{\sim_{\rho}}\in\rti$.

Considering constant net $x_{\eps}=r\in\R$ we have the embedding
$\R\subset\rcrho$. We define $[x_{\eps}]\le[y_{\eps}]$ if there
exists $(z_{\eps})\in\R^{I}$ such that $(z_{\eps})\sim_{\rho}0$
(we then say that $(z_{\eps})$ is \emph{$\rho$-negligible}) and
$x_{\eps}\le y_{\eps}+z_{\eps}$ for $\eps$ small. Equivalently,
we have that $x\le y$ if and only if there exist representatives
$[x_{\eps}]=x$ and $[y_{\eps}]=y$ such that $x_{\eps}\le y_{\eps}$
for all $\eps$. A proficient intuitive point of view on these generalized
numbers is to think at $[x_{\eps}]\in\rcrho$ as a \emph{dynamic point
in the time $\eps\to0^{+}$}; classical real numbers are hence \emph{static
points}. We say that $x=[x_{\eps}]\in\rti$ is \emph{near-standard}
if $\exists\lim_{\eps\to0^{+}}x_{\eps}=:\st{x}\in\R$.

Even though the order $\le$ is not total, we still have the possibility
to define the infimum $[x_{\eps}]\wedge[y_{\eps}]:=[\min(x_{\eps},y_{\eps})]$,
the supremum $[x_{\eps}]\vee[y_{\eps}]:=\left[\max(x_{\eps},y_{\eps})\right]$
of a finite amount of generalized numbers. Henceforth, we will also
use the customary notation $\RC{\rho}^{*}$ for the set of invertible
generalized numbers, and we write $x<y$ to say that $x\le y$ and
$x-y\in\rcrho^{*}$, i.e.~if $x$ is less of equal to $y$ and $x-y$
is invertible. The intervals are denoted by: $[a,b]:=\{x\in\RC{\rho}\mid a\le x\le b\}$,
$[a,b]_{\R}:=[a,b]\cap\R$. Finally, we set $\diff{\rho}:=[\rho_{\eps}]\in\RC{\rho}$,
which is a positive invertible infinitesimal, whose reciprocal is
$\diff{\rho}^{-1}=[\rho_{\eps}^{-1}]$, which is necessarily a strictly
positive infinite number. It is remarkable to note that $x=[x_{\eps}]\in\rti$
is an infinitesimal number, i.e.~$|x|\le r$ for all $r\in\R_{>0}$,
denoted by $x\approx0$, if and only if $\lim_{\eps\to0^{+}}x_{\eps}=0$;
similarly, $x$ is an infinite number, i.e.~$|x|\ge r$ for all $r\in\R_{>0}$,
if and only if $\lim_{\eps\to0^{+}}|x_{\eps}|=+\infty$. This intuitively
clear result is not possible neither in nonstandard analysis nor in
synthetic differential geometry, see \cite{Gio10,Hen,Koc}.

The following result proves to be useful in dealing with positive
and invertible generalized numbers. For its proof, see e.g.~\cite{GKOS}.
\begin{lem}
\label{lem:mayer} Let $x\in\RC{\rho}$. Then the following are equivalent:
\begin{enumerate}
\item \label{enu:positiveInvertible}$x$ is invertible and $x\ge0$, i.e.~$x>0$.
\item \label{enu:strictlyPositive}For each representative $(x_{\eps})\in\R_{\rho}$
of $x$ we have $\forall^{0}\eps:\ x_{\eps}>0$.
\item \label{enu:greater-i_epsTom}For each representative $(x_{\eps})\in\R_{\rho}$
of $x$ we have $\exists m\in\N\,\forall^{0}\eps:\ x_{\eps}>\rho_{\eps}^{m}$.
\item \label{enu:There-exists-a}There exists a representative $(x_{\eps})\in\R_{\rho}$
of $x$ such that $\exists m\in\N\,\forall^{0}\eps:\ x_{\eps}>\rho_{\eps}^{m}$.
\end{enumerate}
\end{lem}

One can clearly feel insecure in working with a ring of scalar which
is not a totally ordered field. On the one hand, we can reread the
list of results presented in Sec.~\ref{sec:Intro} to get a reassurance
that these properties are actually not indispensable to obtain all
these well-known classical results. On the other hand, using the notion
of \emph{subpoint} (e.g.~a meaningful case is given by a subpoint
$[x_{\eps_{n}}]$ of $[x_{\eps}]$ which is considered only on a sequence
$(\eps_{n})_{n\in\N}\to0^{+}$), see \cite{MTAG}, we developed very
practical substitutes of both the field and the total order property.

\subsection{\label{subsec:subpt}The language of subpoints}

The following simple language allows us to simplify some proofs using
steps recalling the classical real field $\R$. We first introduce
the notion of \emph{subpoint}:
\begin{defn}
For subsets $J$, $K\subseteq I$ we write $K\subzero J$ if $0$
is an accumulation point of $K$ and $K\sse J$ (we read it as: $K$
\emph{is co-final in $J$}). Note that for any $J\subzero I$, the
constructions introduced so far in Def.~\ref{def:RCGN} can be repeated
using nets $(x_{\eps})_{\eps\in J}$. We indicate the resulting ring
with the symbol $\rcrho^{n}|_{J}$. More generally, no peculiar property
of $I=(0,1]$ will ever be used in the following, and hence all the
presented results can be easily generalized considering any other
directed set. If $K\subzero J$, $x\in\rcrho^{n}|_{J}$ and $x'\in\rcrho^{n}|_{K}$,
then $x'$ is called a \emph{subpoint} of $x$, denoted as $x'\subseteq x$,
if there exist representatives $(x_{\eps})_{\eps\in J}$, $(x'_{\eps})_{\eps\in K}$
of $x$, $x'$ such that $x'_{\eps}=x_{\eps}$ for all $\eps\in K$.
Intuitively, $x'$ is a subpoint of $x$ if the net of dynamical values
$(x'_{\eps})$ is extracted from $(x_{\eps})$ only for all $\eps\in K$,
and in this $K\subseteq I$ we can take $\eps\to0^{+}$. In other
words, it is the same dynamical point but only for some $\eps$-times
that accumulate around $0^{+}$. In this case we write $x'=x|_{K}$,
$\dom{x'}:=K$, and the restriction $(-)|_{K}:\rcrho^{n}\ra\rcrho^{n}|_{K}$
is a well defined operation. In general, for $X\sse\rcrho^{n}$ we
set $X|_{J}:=\{x|_{J}\in\rcrho^{n}|_{J}\mid x\in X\}$.
\end{defn}

In the next definition, we introduce binary relations that hold only
\emph{on subpoints}. This idea is inherited from nonstandard analysis,
where cofinal subsets are always taken in a fixed ultrafilter.
\begin{defn}
Let $x$, $y\in\rcrho$, $L\subzero I$, then we say
\begin{enumerate}
\item $x<_{L}y\ :\iff\ x|_{L}<y|_{L}$ (the latter inequality has to be
meant in the ordered ring $\rcrho|_{L}$). We read $x<_{L}y$ as ``\emph{$x$
is less than $y$ on $L$}''.
\item $x\sbpt{<}y\ :\iff\ \exists L\subzero I:\ x<_{L}y$. We read $x\sbpt{<}y$
as ``\emph{$x$ is less than $y$ on subpoints''.}
\end{enumerate}
Analogously, we can define other relations holding only on subpoints
such as e.g.: $\sbpt{\in}$, $\sbpt{\le}$, $\sbpt{=}$, $\sbpt{\subseteq}$,
etc.
\end{defn}

\noindent For example, we have
\begin{align*}
x\le y\  & \iff\ \forall L\subzero I:\ x\le_{L}y\\
x<y\  & \iff\ \forall L\subzero I:\ x<_{L}y,
\end{align*}
the former following from the definition of $\le$, whereas the latter
following from Lem.~\ref{lem:mayer}. Moreover, if $\mathcal{P}\left\{ x_{\eps}\right\} $
is an arbitrary property of $x_{\eps}$, then
\begin{equation}
\neg\left(\forall^{0}\eps:\ \mathcal{P}\left\{ x_{\eps}\right\} \right)\ \iff\ \exists L\subzero I\,\forall\eps\in L:\ \neg\mathcal{P}\left\{ x_{\eps}\right\} .\label{eq:negation}
\end{equation}

Note explicitly that, generally speaking, relations on subpoints,
such as $\sbpt{\le}$ or $\sbpt{=}$, do not inherit the same properties
of the corresponding relations for points. So, e.g., both $\sbpt{=}$
and $\sbpt{\le}$ are not transitive relations.

The next result clarifies how to equivalently write a negation of
an inequality or of an equality using the language of subpoints.
\begin{lem}
\label{lem:negationsSubpoints}Let $x$, $y\in\rcrho$, then
\begin{enumerate}
\item \label{enu:neg-le}$x\nleq y\quad\Longleftrightarrow\quad x\sbpt{>}y$
\item \label{enu:negStrictlyLess}$x\not<y\quad\Longleftrightarrow\quad x\sbpt{\ge}y$
\item \label{enu:negEqual}$x\ne y\quad\Longleftrightarrow\quad x\sbpt{>}y$
or $x\sbpt{<}y$.
\end{enumerate}
\end{lem}

Using the language of subpoints, we can write different forms of dichotomy
or trichotomy laws for inequality. The first form is the following
\begin{lem}
\label{lem:trich1st}Let $x$, $y\in\rcrho$, then
\begin{enumerate}
\item \label{enu:dichotomy}$x\le y$ or $x\sbpt{>}y$
\item \label{enu:strictDichotomy}$\neg(x\sbpt{>}y$ and $x\le y)$
\item \label{enu:trichotomy}$x=y$ or $x\sbpt{<}y$ or $x\sbpt{>}y$
\item \label{enu:leThen}$x\le y\ \Rightarrow\ x\sbpt{<}y$ or $x=y$
\item \label{enu:leSubpointsIff}$x\sbpt{\le}y\ \Longleftrightarrow\ x\sbpt{<}y$
or $x\sbpt{=}y$.
\end{enumerate}
\end{lem}

\noindent As usual, we note that these results can also be trivially
repeated for the ring $\rcrho|_{L}$. So, e.g., we have $x\not\le_{L}y$
if and only if $\exists J\subzero L:\ x>_{J}y$, which is the analog
of Lem.~\ref{lem:negationsSubpoints}.\ref{enu:neg-le} for the ring
$\rcrho|_{L}$.

The second form of trichotomy (which for $\rcrho$ can be more correctly
named as \emph{quadrichotomy})\emph{ is} stated as follows:
\begin{lem}
\label{lem:trich2nd}Let $x=[x_{\eps}]$, $y=[y_{\eps}]\in\rcrho$,
then
\begin{enumerate}
\item \label{enu:quadrichotomyLessEq}$x\le y$ or $x\ge y$ or $\exists L\subzero I:\ L^{c}\subzero I,\ x\ge_{L}y\text{ and }x\le_{L^{c}}y$
\item \label{enu:fromL2noL}If for all $L\subzero I$ the following implication
holds
\begin{equation}
x\le_{L}y,\text{ or }x\ge_{L}y\ \Rightarrow\ \forall^{0}\eps\in L:\ \mathcal{P}\left\{ x_{\eps},y_{\eps}\right\} ,\label{eq:trichHp}
\end{equation}
then $\forall^{0}\eps:\ \mathcal{P}\left\{ x_{\eps},y_{\eps}\right\} $.
\item \label{enu:fromL2noLStrict}If for all $L\subzero I$ the following
implication holds
\begin{equation}
x<_{L}y,\text{ or }x>_{L}y\text{ or }x=_{L}y\ \Rightarrow\ \forall^{0}\eps\in L:\ \mathcal{P}\left\{ x_{\eps},y_{\eps}\right\} ,\label{eq:trichHpStrict}
\end{equation}
then $\forall^{0}\eps:\ \mathcal{P}\left\{ x_{\eps},y_{\eps}\right\} $.
\end{enumerate}
\end{lem}

\noindent Property Lem.~\ref{lem:trich2nd}.\ref{enu:fromL2noL}
represents a typical replacement of the usual dichotomy law in $\R$:
for arbitrary $L\subzero I$, we can assume to have two cases: either
$x\le_{L}y$ or $x\ge_{L}y$. If in both cases we are able to prove
$\mathcal{P}\{x_{\eps},y_{\eps}\}$ for $\eps\in L$ small, then we
always get that this property holds for all $\eps$ small. Similarly,
we can use the strict trichotomy law stated in \ref{enu:fromL2noLStrict}.

\subsection{\label{subsec:Topologies}Topologies on $\RC{\rho}^{n}$}

A first non-trivial conceptual step is to consider $\rcrho$ as our
new ring of scalar. The natural extension of the Euclidean norm on
the $\RC{\rho}$-module $\RC{\rho}^{n}$, i.e.~$|[x_{\eps}]|:=[|x_{\eps}|]\in\RC{\rho}$,
where $[x_{\eps}]\in\RC{\rho}^{n}$, goes exactly in this direction.
In fact, even if this generalized norm takes values in $\RC{\rho}$,
and not in the old $\R$, it shares some essential properties with
classical norms: 
\begin{align*}
 & |x|=x\vee(-x)\\
 & |x|\ge0\\
 & |x|=0\Rightarrow x=0\\
 & |y\cdot x|=|y|\cdot|x|\\
 & |x+y|\le|x|+|y|\\
 & ||x|-|y||\le|x-y|.
\end{align*}

\noindent It is therefore natural to consider on $\RC{\rho}^{n}$
topologies generated by balls defined by this generalized norm and
a set of radii. A second non-trivial step is to understand that the
meaningful set of radii we need to have continuity of our class of
generalized function is the set $\RC{\rho}_{\ge0}^{*}=\rcrho_{>0}$
of positive and invertible generalized numbers:
\begin{defn}
\label{def:setOfRadii}We define
\begin{enumerate}
\item $B_{r}(x):=\left\{ y\in\RC{\rho}^{n}\mid\left|y-x\right|<r\right\} $
for any $r\in\rcrho_{>0}$.
\item $\Eball_{r}(x):=\{y\in\R^{n}\mid|y-x|<r\}$, for any $r\in\R_{>0}$,
denotes an ordinary Euclidean ball in $\R^{n}$.
\end{enumerate}
\end{defn}

\noindent The relation $<$ has more beneficial topological properties
as compared to the usual strict order relation $x\le y$ and $x\ne y$
(a relation that we will therefore \emph{never} use) due to the following
result:
\begin{thm}
\label{thm:intersectionBalls}The set of balls $\left\{ B_{r}(x)\mid r\in\rcrho_{>0},\ x\in\RC{\rho}^{n}\right\} $
is a base for a topology on $\RC{\rho}^{n}$ called \emph{sharp topology},
and we call \emph{sharply open set} any open set in this topology.
\end{thm}

\noindent We also recall that the sharp topology on $\rcrho^{n}$
is Hausdorff and Cauchy complete, see e.g.~\cite{Gio-Kun-Ver19,Gio-Kun-Ver15}.
A peculiar property of the sharp topology is that it is also generated
by all the infinitesimal balls of the form $B_{\diff{\rho}^{q}}(x)$,
where $q\in\N_{>0}$. The necessity to consider infinitesimal neighborhoods
occurs in any theory containing continuous GF which have infinite
derivatives. Indeed, from the mean value theorem Thm.~\ref{thm:classicalThms}.\ref{enu:meanValue}
below, we have $f(x)-f(x_{0})=f'(c)\cdot(x-x_{0})$ for some $c\in[x,x_{0}]$.
Therefore, we have $f(x)\in B_{r}(f(x_{0}))$, for a given $r\in\RC{\rho}_{>0}$,
if and only if $|x-x_{0}|\cdot|f'(c)|<r$, which yields an infinitesimal
neighborhood of $x_{0}$ in case $f'(c)$ is infinite; see \cite{GK13b,Gio-Kun-Ver15}
for precise statements and proofs corresponding to this intuition.
On the other hand, the existence of infinitesimal neighborhoods implies
that the sharp topology induces the discrete topology on $\R$; once
again, this is a general result that occurs in all the theories of
infinitesimals, see \cite{GK13b}.

A natural way to obtain sharply open, closed and bounded sets in $\RC{\rho}^{n}$
is by using a net $(A_{\eps})$ of subsets $A_{\eps}\subseteq\R^{n}$.
Once again, thinking at $[x_{\eps}]$ and $(A_{\eps})$ as a dynamic
point and set respectively, we have two ways of extending the membership
relation $x_{\eps}\in A_{\eps}$ to generalized points $[x_{\eps}]\in\RC{\rho}^{n}$:
\begin{defn}
\label{def:internalStronglyInternal}Let $(A_{\eps})$ be a net of
subsets of $\R^{n}$, then
\begin{enumerate}
\item $[A_{\eps}]:=\left\{ [x_{\eps}]\in\RC{\rho}^{n}\mid\forall^{0}\eps:\,x_{\eps}\in A_{\eps}\right\} $
is called the \emph{internal set} generated by the net $(A_{\eps})$.
\item Let $(x_{\eps})$ be a net of points of $\R^{n}$, then we say that
$x_{\eps}\in_{\eps}A_{\eps}$, and we read it as $(x_{\eps})$ \emph{strongly
belongs to $(A_{\eps})$}, if
\begin{enumerate}
\item $\forall^{0}\eps:\ x_{\eps}\in A_{\eps}$.
\item If $(x'_{\eps})\sim_{\rho}(x_{\eps})$, then also $x'_{\eps}\in A_{\eps}$
for $\eps$ small.
\end{enumerate}
\noindent Moreover, we set $\sint{A_{\eps}}:=\left\{ [x_{\eps}]\in\RC{\rho}^{n}\mid x_{\eps}\in_{\eps}A_{\eps}\right\} $,
and we call it the \emph{strongly internal set} generated by the net
$(A_{\eps})$.
\item We say that the internal set $K=[A_{\eps}]$ is \emph{sharply bounded}
if there exists $M\in\RC{\rho}_{>0}$ such that $K\subseteq B_{M}(0)$.
\item Finally, we say that the net $(A_{\eps})$ is \emph{sharply bounded
}if there exists $N\in\R_{>0}$ such that $\forall^{0}\eps\,\forall x\in A_{\eps}:\ |x|\le\rho_{\eps}^{-N}$.
\end{enumerate}
\end{defn}

\noindent Therefore, $x\in[A_{\eps}]$ if there exists a representative
$[x_{\eps}]=x$ such that $x_{\eps}\in A_{\eps}$ for $\eps$ small,
whereas this membership is independent from the chosen representative
in case of strongly internal sets. An internal set generated by a
constant net $A_{\eps}=A\subseteq\R^{n}$ will simply be denoted by
$[A]$.

The following theorem shows that internal and strongly internal sets
have dual topological properties:
\begin{thm}
\noindent \label{thm:strongMembershipAndDistanceComplement}For $\eps\in I$,
let $A_{\eps}\subseteq\R^{n}$ and let $x_{\eps}\in\R^{n}$. Then
we have
\begin{enumerate}
\item \label{enu:internalSetsDistance}$[x_{\eps}]\in[A_{\eps}]$ if and
only if $\forall q\in\R_{>0}\,\forall^{0}\eps:\ d(x_{\eps},A_{\eps})\le\rho_{\eps}^{q}$.
Therefore $[x_{\eps}]\in[A_{\eps}]$ if and only if $[d(x_{\eps},A_{\eps})]=0\in\RC{\rho}$.
\item \label{enu:stronglyIntSetsDistance}$[x_{\eps}]\in\sint{A_{\eps}}$
if and only if $\exists q\in\R_{>0}\,\forall^{0}\eps:\ d(x_{\eps},A_{\eps}^{\text{c}})>\rho_{\eps}^{q}$,
where $A_{\eps}^{\text{c}}:=\R^{n}\setminus A_{\eps}$. Therefore,
if $(d(x_{\eps},A_{\eps}^{\text{c}}))\in\R_{\rho}$, then $[x_{\eps}]\in\sint{A_{\eps}}$
if and only if $[d(x_{\eps},A_{\eps}^{\text{c}})]>0$.
\item \label{enu:internalAreClosed}$[A_{\eps}]$ is sharply closed.
\item \label{enu:stronglyIntAreOpen}$\sint{A_{\eps}}$ is sharply open.
\item \label{enu:internalGeneratedByClosed}$[A_{\eps}]=\left[\text{\emph{cl}}\left(A_{\eps}\right)\right]$,
where $\text{\emph{cl}}\left(S\right)$ is the closure of $S\subseteq\R^{n}$.
\item \label{enu:stronglyIntGeneratedByOpen}$\sint{A_{\eps}}=\sint{\text{\emph{int}\ensuremath{\left(A_{\eps}\right)}}}$,
where $\emph{int}\left(S\right)$ is the interior of $S\subseteq\R^{n}$.
\end{enumerate}
\end{thm}

\noindent For example, it is not hard to show that the closure in
the sharp topology of a ball of center $c=[c_{\eps}]$ and radius
$r=[r_{\eps}]>0$ is 
\begin{equation}
\overline{B_{r}(c)}=\left\{ x\in\rti^{d}\mid\left|x-c\right|\le r\right\} =\left[\overline{\Eball_{r_{\eps}}(c_{\eps})}\right],\label{eq:closureBall}
\end{equation}
whereas
\[
B_{r}(c)=\left\{ x\in\rti^{d}\mid\left|x-c\right|<r\right\} =\sint{\Eball_{r_{\eps}}(c_{\eps})}.
\]

\noindent The reader can be concerned with the fact that the ring
of scalar $\rti$ is not a totally ordered field. Besides the language
of subpoints (see Sec.~\ref{subsec:subpt}) that allows one to proceed
alternatively when total order or invertibility properties are in
play, the following result is also useful:
\begin{lem}
\label{lem:invDense}Invertible elements of $\rti$ are dense in the
sharp topology, i.e.
\[
\forall h\in\rti\,\forall\delta\in\rti_{>0}\,\exists k\in(h-\delta,h+\delta):\ k\text{ is invertible}.
\]
\end{lem}

\noindent This is even more important since our GSF are continuous
in the sharp topology, as we will see in the next section.

\section{\label{sec:Functions}Functions and distributions: generalized smooth
functions}

After the introduction of numbers, their sets and topologies, we introduce
the notion of function.

\subsection{\label{subsec:DefSharpCont}Definition of GSF and sharp continuity}

Using the ring $\rti$, it is easy to consider a Gaussian with an
infinitesimal standard deviation. If we denote this probability density
by $f(x,\sigma)$, and if we set $\sigma=[\sigma_{\eps}]\in\RC{\rho}_{>0}$,
where $\sigma\approx0$, we obtain the net of smooth functions $(f(-,\sigma_{\eps}))_{\eps\in I}$.
This is the basic idea we are going to develop in the following definitions.
We will first introduce the notion of a net of functions $(f_{\eps})$
defining a generalized smooth function of the type $X\longrightarrow Y$,
where $X\subseteq\RC{\rho}^{n}$ and $Y\subseteq\RC{\rho}^{d}$. This
is a net of smooth functions $f_{\eps}\in\Coo(\Omega_{\eps},\R^{d})$
that induces well-defined maps of the form $[\partial^{\alpha}f_{\eps}(-)]:\sint{\Omega_{\eps}}\ra\RC{\rho}^{d}$,
for every multi-index $\alpha\in\N^{n}$.
\begin{defn}
\label{def:netDefMap}Let $X\subseteq\RC{\rho}^{n}$ and $Y\subseteq\RC{\rho}^{d}$
be arbitrary subsets of generalized points. Let $(\Omega_{\eps})$
be a net of open subsets of $\R^{n}$, and $(f_{\eps})$ be a net
of smooth functions, with $f_{\eps}\in\Coo(\Omega_{\eps},\R^{d})$.
Then, we say that 
\[
(f_{\eps})\textit{ defines a generalized smooth function}:X\longrightarrow Y
\]
if:
\begin{enumerate}
\item \label{enu:dom-cod}$X\subseteq\sint{\Omega_{\eps}}$ and $[f_{\eps}(x_{\eps})]\in Y$
for all $[x_{\eps}]\in X$.
\item \label{enu:partial-u-moderate}$\forall[x_{\eps}]\in X\,\forall\alpha\in\N^{n}:\ (\partial^{\alpha}f_{\eps}(x_{\eps}))\in\R_{\rho}^{d}$.
\end{enumerate}
\noindent Where the notation 
\[
\forall[x_{\eps}]\in X:\ \mathcal{P}\{(x_{\eps})\}
\]
means 
\[
\forall(x_{\eps})\in\R_{\rho}^{n}:\ [x_{\eps}]\in X\ \Rightarrow\ \mathcal{P}\{(x_{\eps})\},
\]
i.e.~for all representatives $(x_{\eps})$ generating a point $[x_{\eps}]\in X$,
the property $\mathcal{P}\{(x_{\eps})\}$ holds.
\end{defn}

\noindent A generalized smooth function (or map, in this paper these
terms are used as synonymous) is simply a function of the form $f=[f_{\eps}(-)]|_{X}$:
\begin{defn}
\label{def:generalizedSmoothMap} Let $X\subseteq\RC{\rho}^{n}$ and
$Y\subseteq\RC{\rho}^{d}$ be arbitrary subsets of generalized points,
then we say that 
\[
f:X\longrightarrow Y\text{ is a \emph{generalized smooth function}}
\]
if $f\in\Set(X,Y)$ and there exists a net $f_{\eps}\in\Coo(\Omega_{\eps},\R^{d})$
defining a generalized smooth map of type $X\longrightarrow Y$, in
the sense of Def.\ \ref{def:netDefMap}, such that 
\begin{equation}
\forall[x_{\eps}]\in X:\ f\left([x_{\eps}]\right)=\left[f_{\eps}(x_{\eps})\right].\label{eq:f-u-relations}
\end{equation}

\noindent We will also say that $f$ \emph{is defined by the net of
smooth functions} $(f_{\eps})$ or that the net $(f_{\eps})$ \emph{represents}
$f$. The set of all these GSF will be denoted by $\gsf(X,Y)$.
\end{defn}

Let us note explicitly that definitions \ref{def:netDefMap} and \ref{def:generalizedSmoothMap}
state minimal logical conditions to obtain a set-theoretical map from
$X$ into $Y$ which is defined by a net of smooth functions such
that all the derivatives still lie in our ring of scalars for condition
Def.~\ref{def:netDefMap}.\ref{enu:partial-u-moderate}. In particular,
the following Thm.~\ref{thm:indepRepr} states that in equality \eqref{eq:f-u-relations}
we have independence from the representatives for all derivatives
$[x_{\eps}]\in X\mapsto[\partial^{\alpha}f_{\eps}(x_{\eps})]\in\RC{\rho}^{d}$,
$\alpha\in\N^{n}$.
\begin{thm}
\label{thm:indepRepr}Let $X\subseteq\RC{\rho}^{n}$ and $Y\subseteq\RC{\rho}^{d}$
be arbitrary subsets of generalized points. Let $(\Omega_{\eps})$
be a net of open subsets of $\R^{n}$, and $(f_{\eps})$ be a net
of smooth functions, with $f_{\eps}\in\Coo(\Omega_{\eps},\R^{d})$.
Assume that $(f_{\eps})$ defines a generalized smooth map of the
type $X\longrightarrow Y$, then 
\[
\forall\alpha\in\N^{n}\,\forall(x_{\eps}),(x'_{\eps})\in\R_{\rho}^{n}:\ [x_{\eps}]=[x'_{\eps}]\in X\ \Rightarrow\ (\partial^{\alpha}f_{\eps}(x_{\eps}))\sim_{\rho}(\partial^{\alpha}f_{\eps}(x'_{\eps})).
\]
\end{thm}

\noindent Note that taking arbitrary subsets $X\subseteq\RC{\rho}^{n}$
in Def.\ \ref{def:netDefMap}, we can also consider GSF defined on
closed sets, like the set of all infinitesimals (which is actually
clopen, like in all non trivial theories of infinitesimals), or like
a closed interval $[a,b]\subseteq\RC{\rho}$. We can also consider
GSF defined at infinite generalized points. A simple case is the exponential
map 
\begin{equation}
e^{(-)}:[x_{\eps}]\in\left\{ x\in\RC{\rho}\mid\exists z\in\RC{\rho}_{>0}:\ x\le\log z\right\} \mapsto\left[e^{x_{\eps}}\right]\in\RC{\rho}.\label{eq:exp}
\end{equation}
The domain of this map depends on the infinitesimal net $\rho$. For
instance, if $\rho=(\eps)$ then all its points are bounded by generalized
numbers of the form $[-N\log\eps]$, $N\in\N$; whereas if $\rho=\left(e^{-\frac{1}{\eps}}\right)$,
all points are bounded by $[N\eps^{-1}]$, $N\in\N$.

A first regularity property of GSF is the above cited continuity with
respect to the sharp topology, as proved in the following
\begin{thm}
\label{thm:GSF-continuity}Let $X\subseteq\RC{\rho}^{n}$, $Y\subseteq\RC{\rho}^{d}$
and $f_{\eps}\in\Coo(\Omega_{\eps},\R^{d})$ be a net of smooth functions
that defines a GSF of the type $X\longrightarrow Y$. Then
\begin{enumerate}
\item \label{enu:locLipSharp}For all $\alpha\in\N^{n}$, the GSF $g:[x_{\eps}]\in X\mapsto[\partial^{\alpha}f_{\eps}(x_{\eps})]\in\Rtil^{d}$
is locally Lipschitz in the sharp topology, i.e.~each $x\in X$ possesses
a sharp neighborhood $U$ such that $|g(x)-g(y)|\le L|x-y|$ for all
$x$, $y\in U$ and some $L\in\RC{\rho}$.
\item \label{enu:GSF-cont}Each $f\in\gsf(X,Y)$ is continuous with respect
to the sharp topologies induced on $X$, $Y$.
\item \label{enu:globallyDefNet}$f:X\longrightarrow Y$ is a GSF if and
only if there exists a net $v_{\eps}\in\cinfty(\R^{n},\R^{d})$ defining
a generalized smooth map of type $X\longrightarrow Y$ such that $f=[v_{\eps}(-)]|_{X}$.
\end{enumerate}
\end{thm}

\subsection{\label{sec:Embeddings}Embedding of Schwartz distributions and Colombeau generalized functions}

Among the re-occurring themes of this work are the choices which the
solution of a given problem within our framework may depend upon.
For instance, \eqref{eq:exp} shows that the domain of a GSF depends
on the infinitesimal net $\rho$. It is also easy to show that the
trivial Cauchy problem 
\[
\begin{cases}
x'(t)=[\eps^{-1}]\cdot x(t)\\
x(0)=1
\end{cases}
\]
has no solution in $\gsf(\R,\R)$ if $\rho_{\eps}=\eps$ because the
solution is not moderate e.g.~at $t=1$. Nevertheless, it has the
unique solution $x(t)=\left[e^{\frac{1}{\eps}t}\right]\in\gsf(\R,\R)$
if $\rho_{\eps}=e^{-\frac{1}{\eps}}$. Therefore, the choice of the
infinitesimal net $\rho$ is closely tied to the possibility of solving
a given class of differential equations in \emph{non infinitesimal}
intervals (a solution in a suitable \emph{infinitesimal} interval
always exists, see Sec.~\ref{sec:Differential-equations-1}). This
illustrates the dependence of the theory on the infinitesimal net
$\rho$.

Further choices concern the embedding of Schwartz distributions: Since
we need to associate a net of smooth functions $(f_{\eps})$ to a
given distribution $T\in\D'(\Omega)$, this embedding is naturally
built upon a regularization process. In our approach, this regularization
will depend on an infinite number $b\in\rcrho$, and the choice of
$b$ depends on what properties we need from the embedding. For example,
if $\delta$ is the (embedding of the) one-dimensional Dirac delta,
then we have the property 
\begin{equation}
\delta(0)=b,\label{eq:deltaAt0}
\end{equation}
We can also choose the embedding so as to get the property 
\begin{equation}
H(0)=\frac{1}{2},\label{eq:H-at-0}
\end{equation}
where $H$ is the (embedding of the) Heaviside step function. Equalities
like these are used in diverse applications (see, e.g., \cite{C1}
and references therein). In fact, we are going to construct a family
of structures of the type $(\mathcal{G},\partial,\iota)$, where $(\mathcal{G},\partial)$
is a sheaf of differential algebra and $\iota:\D'\ra\mathcal{G}$
is an embedding. The particular structure we need to consider depends
on the problem we have to solve. Of course, one may be more interested
in having an intrinsic embedding of distributions. This can be done
by following the ideas of the full Colombeau algebra (see e.g.~\cite{GKOS}).
Nevertheless, this choice decreases the simplicity of the present
approach and is incompatible with properties like \eqref{eq:deltaAt0}
and \eqref{eq:H-at-0}.

If $\phi\in\mathcal{D}(\R^{n})$, $r\in\R_{>0}$ and $x\in\R^{n}$,
we use the notation $r\odot\phi$ for the function $x\in\R^{n}\mapsto\frac{1}{r^{n}}\cdot\phi\left(\frac{x}{r}\right)\in\R$.
Our embedding procedure will ultimately rely on convolution with suitable
mollifiers. To construct these, we need some technical preparations.
\begin{lem}
\label{lem:ColombeauMollifier} For any $n\in\N_{>0}$ there exists
some $\mu_{n}\in\mathcal{S}(\R)$ with the following properties:
\begin{enumerate}
\item \label{enu:iColMoll}$\int\mu_{n}(x)\,\diff{x}=1$.
\item \label{enu:nullMoments}$\int_{0}^{\infty}x^{\frac{j}{n}}\mu_{n}(x)\,\diff{x}=0$
for all $j\in\N_{>0}$.
\item $\mu_{n}(0)=1$.
\item $\mu_{n}$ is even.
\item \label{enu:vColMoll}$\mu_{n}(k)=0$ for all $k\in\mathbb{Z}\setminus\{0\}$.
\end{enumerate}
\end{lem}

\noindent We call \emph{Colombeau mollifier} (for a fixed dimension
$n$) any function $\mu$ that satisfies the properties of the previous
lemma. Concerning embeddings of Schwartz distributions, the idea is
classically to regularize distributions using a mollifier. The use
of a Colombeau mollifier allows us, on the one hand, to identify the
distribution $\phi\in\D(\Omega)\mapsto\int f\phi$ with the GSF $f\in\Coo(\Omega)\subseteq\gsf(\Omega,\R)$
(thanks to property \ref{enu:nullMoments}); on the other hand, it
allows us to explicitly calculate compositions such as $\delta\circ\delta$,
$H\circ\delta$, $\delta\circ H$ (see below).
\noindent \begin{center}
\begin{figure}
\label{fig: Col_mol}
\begin{centering}
\includegraphics[scale=0.2]{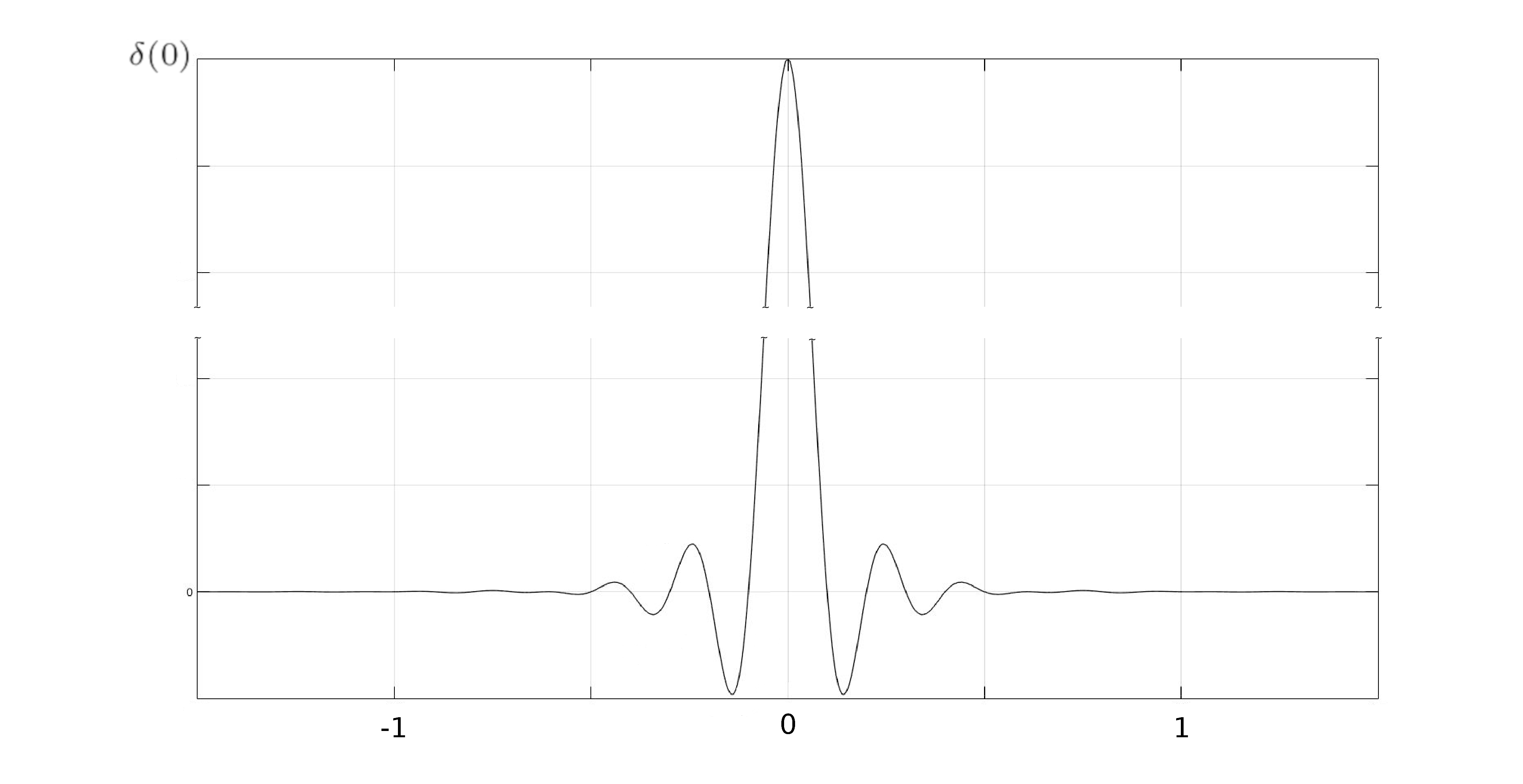}
\par\end{centering}
\begin{centering}
\includegraphics[scale=0.2]{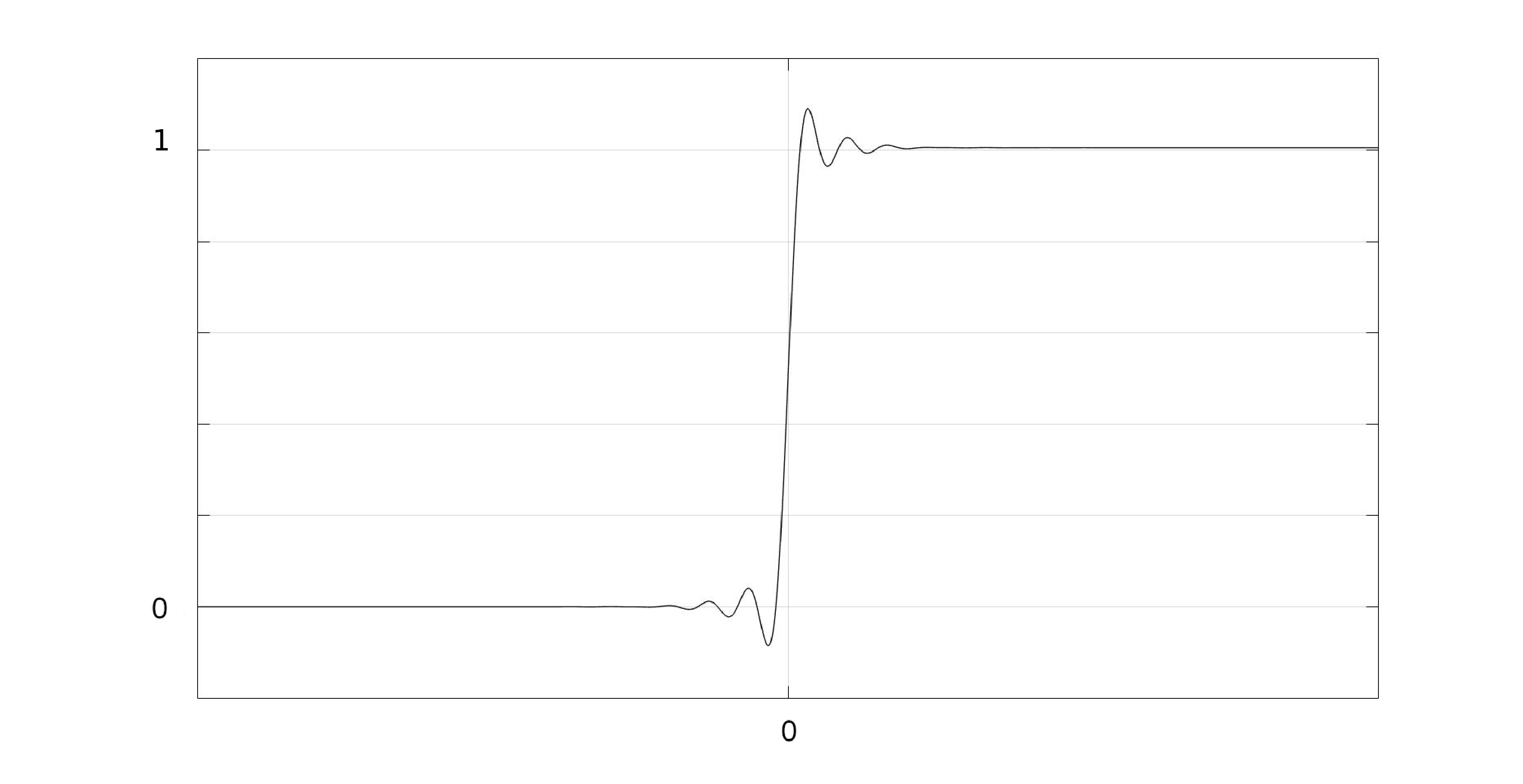}
\par\end{centering}
\caption{\label{fig:MollifierHeaviside}A representation of Dirac delta and
Heaviside function. A Colombeau mollifier has a representation similar
to Dirac delta (but with finite values).}
\end{figure}
\par\end{center}

As a final preparation for the embedding of $\D'(\Omega)$ into $\gsf(\csp{\Omega},\RC{\rho})$
we need to construct suitable $n$-dimensional mollifiers from a Colombeau
mollifier $\mu$ as given by Lemma \ref{lem:ColombeauMollifier}.
To this end, let $\omega_{n}$ denote the surface area of $S^{n-1}$
and set 
\[
c_{n}:=\left\{ \begin{array}{lr}
\frac{2n}{\omega_{n}} & \text{for }n>1\\
1 & \text{for }n=1.
\end{array}\right.
\]
Then let $\tilde{\mu}:\R^{n}\to\R$, $\tilde{\mu}(x):=c_{n}\mu(|x|^{n})$.
Since $\mu$ is even, $\tilde{\mu}$ is smooth. Moreover, by Lemma
\ref{lem:ColombeauMollifier}, it has unit integral and all its higher
moments $\int x^{\alpha}\tilde{\mu}(x)\,dx$ vanish ($|\alpha|\ge1$).

Both Schwartz distributions and Colombeau generalized functions are
naturally defined only on \emph{finite }points of $\langle\Omega\rangle$
(also called \emph{compactly supported points}), i.e.~on the set
\[
\csp{\Omega}:=\{x\in\langle\Omega\rangle\mid\exists R\in\R_{>0}:\ |x|\le R,\ d(x,\partial\Omega)\in\R_{>0}\}
\]
of finite points that remain sufficiently far from the boundary. This
underscores an important difference between this type of GF and GSF,
since the latter can also be defined on purely infinitesimal domains
(note that $\Omega\subseteq\csp{\Omega}$) or on infinite points.
\begin{thm}
\label{thm:embeddingD'} Let $(\emptyset\not=)\Omega\subseteq\R^{n}$
be an open set. Set
\[
\Omega{}_{\eps}:=\left\{ x\in\Omega\mid d(x,\Omega^{c})\ge\eps,\ |x|\le\frac{1}{\eps}\right\} 
\]
and fix some $\chi\in\D(\R^{n})$, $\chi=1$ on $\overline{\Eball_{1}(0)}$,
$0\le\chi\le1$ and $\chi=0$ on $\R^{n}\setminus\Eball_{2}(0)$.
Take $\kappa_{\eps}\in\D(\Omega)$ such that $\kappa_{\eps}=1$ on
a neighborhood $L_{\eps}$ of $\Omega_{\eps}$. Also, let $b=[b_{\eps}]\in\rcrho$
be an infinite positive number, i.e.~$\lim_{\eps\to0^{+}}b_{\eps}=+\infty$.
Set 
\begin{equation}
\mu_{\eps}^{b}(x):=(b_{\eps}^{-1}\odot\tilde{\mu})(x)\chi(x|\log(b_{\eps})|)=b_{\eps}^{n}\tilde{\mu}(b_{\eps}x)\chi(x|\log(b_{\eps})|).\label{col_mol-1}
\end{equation}
Then the map 
\begin{equation}
\iota_{\Omega}^{b}:T\in\D'(\Omega)\mapsto\left[\left(\left(\kappa_{\eps}\cdot T\right)*\mu_{\eps}^{b}\right)(-)\right]\in\gsf(\csp{\Omega},\RC{\rho}).\label{eq:ColEmb}
\end{equation}
satisfies:
\begin{enumerate}
\item \label{enu:restriction}$\iota^{b}:\mathcal{D}'\ra\gsf(\csp{-},\rti)$
is a sheaf-morphism of real vector spaces: If $\Omega'\subseteq\Omega$
is another open set and $T\in\mathcal{D}'(\Omega)$, then $\iota_{\Omega}^{b}(T)|_{\csp{\Omega'}}=\iota_{\Omega'}^{b}(T|_{\Omega'})$.
\item \label{enu:embInjective} $\iota^{b}$ preserves supports, where $\supp(f):=\left(\bigcup\left\{ \Omega'\subseteq\Omega\mid\Omega'\text{ open},\ f|_{\Omega'}=0\right\} \right)^{\text{c}}$,
hence is in fact a sheaf-monomorphism.
\item \label{enu:embeddingSmoothUpToInfinitesimals} Any $f\in\cinfty(\Omega)$
can naturally be considered an element of $\gsf(\csp{\Omega},\RC{\rho})$
via $[x_{\eps}]\mapsto[f(x_{\eps})]$. Moreover, $\forall q\in\N_{>0}\ \forall x\in\csp{\Omega}:\ \left|\iota_{\Omega}^{b}(f)(x)-f(x)\right|\le b^{-q}$.
\item \label{enu:smoothEmb} If $f\in\cinfty(\Omega)$ and if $b\ge\diff{\rho}^{-a}$
for some $a\in\R_{>0}$, then $\iota_{\Omega}^{b}(f)=f$. In particular,
$\iota^{b}$ then provides a multiplicative sheaf-monomorphism $\Coo(-)$
$\hookrightarrow\gsf(c(-),\R)$.
\item \label{enu:commute_der} For any $T\in\D'(\Omega)$ and any $\alpha\in\N^{n}$,
$\iota_{\Omega}^{b}(\partial^{\alpha}T)=\partial^{\alpha}\iota_{\Omega}^{b}(T)$.
\item \label{enu:valuesDistr}Let $b\ge\diff{\rho}^{-a}$ for some $a\in\R_{>0}$.
Then for any $\phi\in\mathcal{D}(\Omega)$ and any $T\in\mathcal{D}'(\Omega)$,
\[
\big[\int_{\Omega}\iota_{\Omega}^{b}(T)_{\eps}(x)\cdot\phi(x)\,\diff{x}\big]=\langle T,\phi\rangle\quad\text{in }\rcrho.
\]
\item \label{enu:deltaH}$\iota_{\R^{n}}^{b}(\delta)(0)=c_{n}b^{n}$ and
if $b\ge\diff{\rho}^{-a}$ for some $a\in\R_{>0}$, then $\iota_{\R}^{b}(H)(0)=\frac{1}{2}$.
\item \label{enu:depXi-e}The embedding $\iota^{b}$ does not depend on
the particular choice of $(\kappa_{\eps})$ and (if $b\ge\diff{\rho}^{-a}$
for some $a\in\R_{>0}$) $\chi$ as above.
\item \label{enu:iota_indep_of_repr} $\iota^{b}$ does not depend on the
representative $(b_{\eps})$ of $b$.
\end{enumerate}
\end{thm}

Whenever we use the notation $\iota^{b}$ for an embedding, we assume
that $b\in\rcrho$ satisfies the overall assumptions of Thm.~\ref{thm:embeddingD'}
and of \ref{enu:smoothEmb} in that Theorem, and that $\iota^{b}$
has been defined as in \eqref{eq:ColEmb} using a Colombeau mollifier
$\mu$ for the given dimension.
\begin{rem}
\label{rem:embedding_properties}~
\begin{enumerate}
\item \label{enu:embDelta}Let $\delta$, $H\in\gsf(\rcrho,\rcrho)$ be
the corresponding $\iota^{b}$-embeddings of the Dirac delta and of
the Heaviside function. Then $\delta(x)=b\cdot\mu(b\cdot x)$ and
$\delta(x)=0$ if $x$ is near-standard and $\st{x}\ne0$ or if $x$
is infinite because $\mu\in\mathcal{S}(\R)$. Also, by construction
of $\mu_{\eps}^{b}$, $\delta$ can be represented like in the first
diagram of Fig.~\ref{fig:MollifierHeaviside}. E.g., $\delta(k/b)=0$
for each $k\in\Z\setminus\{0\}$, and each $\frac{k}{b}$ is a nonzero
infinitesimal. Similar properties can be stated e.g.~for $\delta^{2}(x)=b^{2}\cdot\mu(b\cdot x)^{2}$.
\item \label{enu:embH}Analogously, we have $H(x)=1$ if $x$ is near-standard
and $\st{x}>0$ or if $x>0$ is infinite; $H(x)=0$ if $x$ is near-standard
and $\st{x}<0$ or if $x<0$ is infinite.
\item Let $\text{vp}(\frac{1}{x})\in\mathcal{D}'(\R)$ be the Cauchy principal
value. If $x=[x_{\eps}]$ is far from the origin, in the sense that
$|x|\ge r$ for some $r\in\R_{>0}$. Then $\iota_{\R}^{b}(\text{vp}(\frac{1}{x}))(x)=\frac{1}{x}$.
The behavior of the GSF $\iota_{\R}^{b}(\text{vp}(\frac{1}{x}))(-)$
in an infinitesimal neighborhood of the origin depends on the Colombeau
mollifier $\mu$. For example, if in Lem.~\ref{lem:ColombeauMollifier}
we add the linear condition $\int\frac{\mu_{n}(x)}{x}\,\diff x=0$,
then also $\iota_{\R}^{b}(\text{vp}(\frac{1}{x}))(0)=0$.
\end{enumerate}
Colombeau's special (or simplified) algebra $\mathcal{G}(\Omega)$
(see \cite{C1,GKOS}) is defined, for $\Om\sse\R^{n}$ open, as the
quotient $\gs(\Om):=\esm(\Om)/\ns(\Om)$ of \emph{moderate nets} modulo
\emph{negligible nets}, where 
\[
\mathcal{E}_{M}(\Omega):=\{(u_{\eps})\in\cinfty(\Omega)^{I}\mid\forall K\Subset\Omega\,\forall\alpha\in\N^{n}\,\exists N\in\N:\sup_{x\in K}|\partial^{\alpha}u_{\eps}(x)|=O(\eps^{-N})\}
\]
and 
\[
\ns(\Omega):=\{(u_{\eps})\in\cinfty(\Omega)^{I}\mid\forall K\Subset\Omega\,\forall\alpha\in\N^{n}\,\forall m\in\N:\sup_{x\in K}|\partial^{\alpha}u_{\eps}(x)|=O(\eps^{m})\}.
\]
It follows from \cite[Th. 37]{GKOS} that $\gs(\Om)$ can be identified,
in the special case of $\rho(\eps)=\eps$, with the algebra $\gsf(\csp{\Omega},\RC{\rho})$
of GSF defined on finite points of $\Omega$. In this setting, Thm.~\ref{thm:embeddingD'}
gives an alternative point of view of the well known facts that the
Colombeau algebra contains $\cinfty(\Om)$ as a faithful subalgebra,
$\D'(\Om)$ as a linear subspace and that the embedding is a sheaf
morphism that commutes with partial derivatives. More general domains
are both useful in applications, for solutions of differential equations,
see Sec.~\ref{sec:Differential-equations-1}, for Fourier transform,
see \cite{MT-Gio22}, and are indeed a necessary requirement for obtaining
a construction that is closed with respect to composition of GF.
\end{rem}

\subsection{\label{subsec:ClosureComposition}Closure with respect to composition}

In contrast to the case of distributions, there is no problem in considering
the composition of two GSF. This property opens new interesting possibilities,
e.g.~in considering differential equations $y'=f(y,t)$, where $y$
and $f$ are GSF. For instance, there is no problem in studying $y'=\delta(y)$
(see Sec.~\ref{sec:Differential-equations-1}).
\begin{thm}
\label{thm:GSFcategory} Subsets $S\subseteq\RC{\rho}^{s}$ with the
trace of the sharp topology, and generalized smooth maps as arrows
form a subcategory of the category of topological spaces. We will
call this category $\gsf$, the \emph{category of GSF}.
\end{thm}

\noindent For instance, we can think of the Dirac delta as a map of
the form $\delta:\RC{\rho}\longrightarrow\RC{\rho}$, and therefore
the composition $e^{\delta}$ is defined in $\{x\in\RC{\rho}\mid\exists z\in\RC{\rho}_{>0}:\ \delta(x)\le\log z\}$,
which of course does not contain $x=0$ but only suitable non zero
infinitesimals. On the other hand, $\delta\circ\delta:\RC{\rho}\ra\RC{\rho}$.
Moreover, from the inclusion of ordinary smooth functions (Thm.~\ref{thm:embeddingD'})
and the closure with respect to composition, it directly follows that
every $\gsf(U,\RC{\rho})$ is an algebra with pointwise operations
for every subset $U\subseteq\RC{\rho}^{n}$.

A natural way to define a GSF is to follow the original idea of classical
authors (see \cite{Laug89,Dir2}) to fix an infinitesimal or infinite
parameter in a suitable ordinary smooth function. We will call this
type of GSF of \emph{Cauchy-Dirac type}; the next theorem specifies
this notion and states that GSF are of Cauchy-Dirac type whenever
the generating net $(f_{\eps})$ is smooth in $\eps$.
\begin{cor}
\label{cor:quasi-stdAreGSF}Let $X\subseteq\R^{n}$, $Y\subseteq\R^{d}$,
$P\subseteq\R^{m}$ be open sets and $\phi\in\cinfty(P\times X,Y)$
be an ordinary smooth function. Let $p\in[P]$, and define $f_{\eps}:=\phi(p_{\eps},-)\in\cinfty(X,Y)$,
then $[f_{\eps}(-)]:[X]\longrightarrow[Y]$ is a GSF. In particular,
if $f:[X]\longrightarrow[Y]$ is a GSF defined by $(f_{\eps})$ and
the net $(f_{\eps})$ is smooth in $\eps$, i.e.\ if 
\[
\exists\phi\in\cinfty((0,1)\times X,Y):\ f_{\eps}=\phi(\eps,-)\quad\forall\eps\in(0,1),
\]
and if $[\eps]\in\rcrho$, then the GSF $f$ is of Cauchy-Dirac type
because $f(x)=\phi([\eps],x)$ for all $x\in[X]$. Finally, Cauchy-Dirac
GSF are closed with respect to composition.
\end{cor}

\begin{example}
\label{enu:deltaCompDelta}The composition $\delta\circ\delta\in\gsf(\rcrho,\rcrho)$
is given by $(\delta\circ\delta)(x)=b\mu\left(b^{2}\mu(bx)\right)$
and is an even function. If $x$ is near-standard and $\st{x}\ne0$,
or $x$ is infinite, then $(\delta\circ\delta)(x)=b$. Since $(\delta\circ\delta)(0)=0$,
by the intermediate value theorem (see Cor\@.~\ref{cor:intermValue}
below), we have that $\delta\circ\delta$ attains any value in the
interval $[0,b]\subseteq\rcrho$. If $0\le x\le\frac{1}{2b}$, then
(for a $\mu$ as in Fig.\ \ref{fig: Col_mol}) $x$ is infinitesimal
and $(\delta\circ\delta)(x)=0$ because $\delta(x)\ge b\mu\left(\frac{1}{k}\right)$
is an infinite number. If $x=\frac{k}{b}$ for some $k\in\N_{>0}$,
then $x$ is still infinitesimal but $(\delta\circ\delta)(x)=b$ because
$\mu(bx)=0$. A representation of $\delta\circ\delta$ is given in
Fig.~\ref{fig:deltaCompDelta}. Analogously, one can deal with $H\circ\delta$
and $\delta\circ H$.
\end{example}

\begin{figure}
\begin{centering}
\includegraphics[scale=0.2]{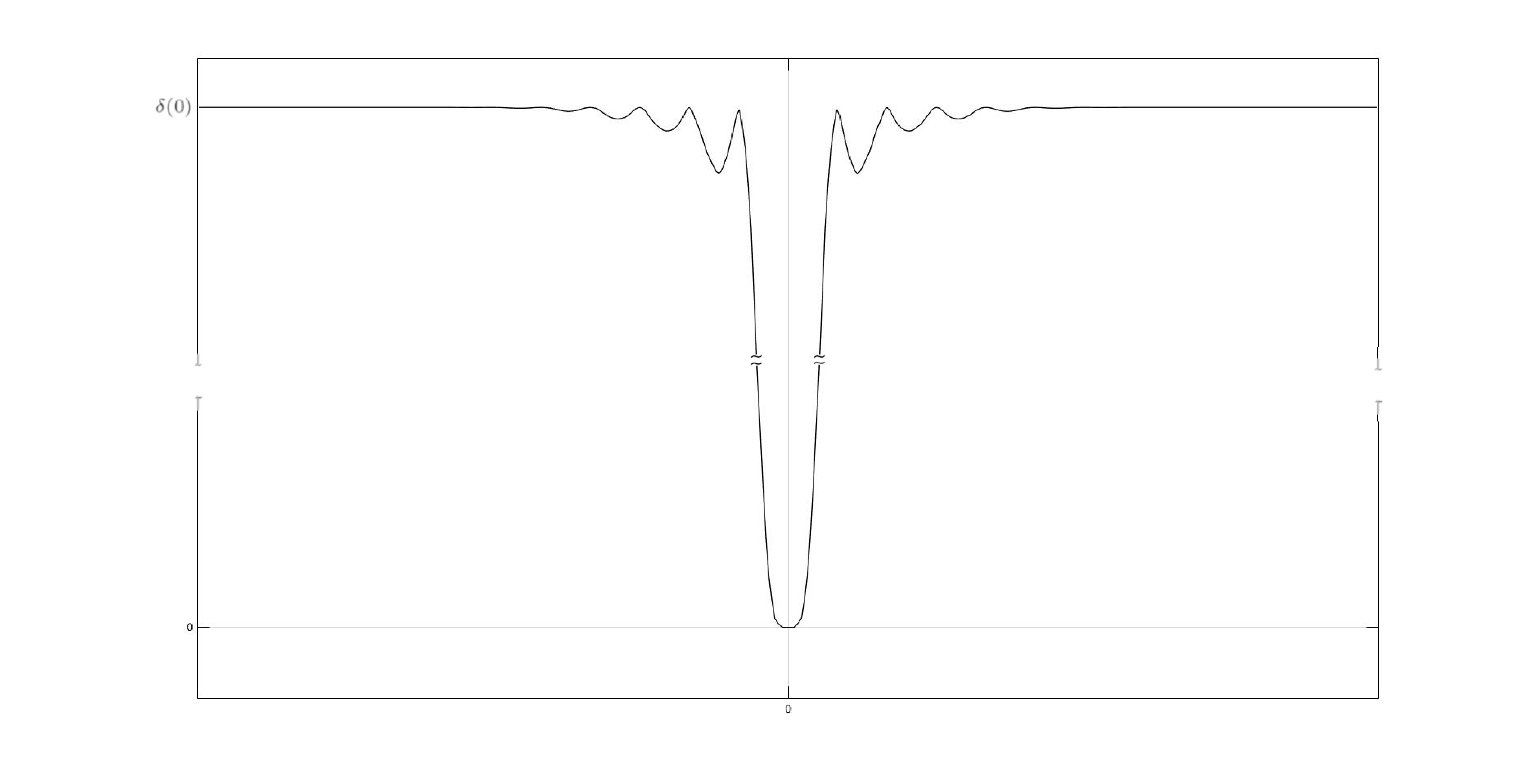}
\par\end{centering}
\caption{\label{fig:deltaCompDelta}A representation of $\delta\circ\delta$}
\end{figure}

The theory of GSF originates from the theory of Colombeau quotient
algebras. In this well-developed approach, strong analytic tools,
including microlocal analysis, and an elaborate theory of pseudodifferential
and Fourier integral operators have been developed over the past few
years (cf.~\cite{C1,GKOS} and references therein). As we already
mentioned above, in these quotient algebras, each generalized function
generates a unique GSF defined on a subset of $\RC{(\eps)}$. On the
other hand, Colombeau GF are in general not closed with respect to
composition, see e.g.~\cite{GKOS}.

Similarly, we can define generalized functions of class $\gckf$,
with $k\le+\infty$:
\begin{defn}
Let $X\subseteq\RC{\rho}^{n}$ and $Y\subseteq\RC{\rho}^{d}$ be arbitrary
subsets of generalized points and $k\in\N\cup\{+\infty\}$. Then we
say that 
\[
f:X\longrightarrow Y\text{ is a \emph{generalized }}\mathcal{C}^{k}\text{ \emph{function}}
\]
if there exists a net $f_{\eps}\in\mathcal{C}^{k}(\Omega_{\eps},\R^{d})$
defining $f$ in the sense that
\begin{enumerate}
\item $X\subseteq\langle\Omega_{\eps}\rangle$,
\item $f([x_{\eps}])=[f_{\eps}(x_{\eps})]\in Y$ for all $x=[x_{\eps}]\in X$,
\item \label{enu:moderCk}$(\partial^{\alpha}f_{\eps}(x_{\eps}))\in\R_{{\scriptscriptstyle \rho}}^{d}$
for all $x=[x_{\eps}]\in X$ and all $\alpha\in\N^{n}$ such that
$|\alpha|\le k$.
\item \label{enu:k1}$\forall\alpha\in\N^{n}\,\forall[x_{\eps}],[x'_{\eps}]\in X:\ |\alpha|=k,\ [x_{\eps}]=[x'_{\eps}]\ \Rightarrow\ [\partial^{\alpha}f_{\eps}(x_{\eps})]=[\partial^{\alpha}f_{\eps}(x'_{\eps})]$.
\item \label{enu:k2}For all $\alpha\in\N^{n}$, with $|\alpha|=k$, the
map $[x_{\eps}]\in X\mapsto[\partial^{\alpha}f_{\eps}(x_{\eps})]\in\rti^{d}$
is continuous in the sharp topology.
\end{enumerate}
The space of generalized $\mathcal{C}^{k}$ functions from $X$ to
$Y$ is denoted by $\gckf(X,Y)$.
\end{defn}

\noindent Note that properties \ref{enu:k1}, \ref{enu:k2} are required
only for $|\alpha|=k$ because for lower length they can be proved
using property \ref{enu:moderCk} and the classical mean value theorem
for $f_{\eps}$ (see e.g.~\cite{Gio-Kun-Ver19}). From Thm.~\ref{thm:indepRepr}
and Thm.~\ref{thm:GSF-continuity}.\ref{enu:GSF-cont} it follows
that this definition of $\gckf$ is equivalent to Def.~\ref{def:netDefMap}
if $k=+\infty$. Moreover, properties similar to \ref{enu:globallyDefNet}
and Thm.~\ref{thm:GSFcategory} can also be proved for $\gckf$.

Note that the absolute value function $|-|:\rti\ra\rti$ is not a
GSF because its derivative is not sharply continuous at the origin;
clearly, it is a $\gcf{0}$ function.

\subsection{\label{sec:Differential-calculus}Differential calculus of GSF}

In this section we show how the derivatives of a GSF can be calculated
using a form of incremental ratio. The idea is to prove the Fermat-Reyes
theorem for GSF (see \cite{Gio-Kun-Ver19,GK13b,Koc}). Essentially,
this theorem shows the existence and uniqueness of another GSF serving
as incremental ratio. This is the first of a long list of results
demonstrating the close similarities between ordinary smooth functions
and GSF.

\noindent In the present setting, the Fermat-Reyes theorem (also called
Carathéodory definition of derivative) is the following.
\begin{thm}
\noindent \label{thm:FR-forGSF} Let $U\subseteq\RC{\rho}^{n}$ be
a sharply open set, let $v=[v_{\eps}]\in\RC{\rho}^{n}$, $k\in\N\cup\{+\infty\}$,
and let $f\in\gcf{k+1}(U,\RC{\rho})$ be a $\gcf{k+1}$map generated
by the net of functions $f_{\eps}\in\mathcal{C}^{k+1}(\Omega_{\eps},\R)$.
Then
\begin{enumerate}
\item \label{enu:existenceRatio}There exists a sharp neighborhood $T$
of $U\times\{0\}$ and a map $r\in\gckf(T,\RC{\rho})$, called the
\emph{generalized incremental ratio} of $f$ \emph{along} $v$, such
that 
\[
\forall(x,h)\in T:\ f(x+hv)=f(x)+h\cdot r(x,h).
\]
\item \label{enu:uniquenessRatio}Any two generalized incremental ratios
coincide on a sharp neighborhood of $U\times\{0\}$, so that we can
use the notation $\frac{\partial f}{\partial v}[x;h]:=r(x,h)$ if
$(x,h)$ are sufficiently small.
\item \label{enu:defDer}We have $\frac{\partial f}{\partial v}[x;0]=\left[\frac{\partial f_{\eps}}{\partial v_{\eps}}(x_{\eps})\right]$
for every $x\in U$ and we can thus define $df(x)\cdot v:=\frac{\partial f}{\partial v}(x):=\frac{\partial f}{\partial v}[x;0]$,
so that $\frac{\partial f}{\partial v}\in\gckf(U,\RC{\rho})$.
\end{enumerate}
\end{thm}

Note that this result allows us to consider the partial derivative
of $f$ with respect to an arbitrary generalized vector $v\in\rcrho^{n}$
which can be, e.g., near-standard or infinite. Since any partial derivative
of a GSF is still a GSF, higher order derivatives $\frac{\partial^{\alpha}f}{\partial v^{\alpha}}\in\gsf(U,\rcrho)$
are simply defined recursively.

As follows from Thm.\ \ref{thm:FR-forGSF}.\ref{enu:existenceRatio}
and Thm\@.~\ref{thm:embeddingD'}.\ref{enu:commute_der}, the concept
of derivative defined using the Fermat-Reyes theorem is compatible
with the classical derivative of Schwartz distributions via the embeddings
$\iota^{b}$ from Thm.\ \ref{thm:embeddingD'}. The following result
follows from the analogous properties for the nets of smooth functions
defining $f$ and $g$ or directly from the Fermat-Reyes Thm.~\ref{thm:FR-forGSF}.
\begin{thm}
\label{thm:linearityLeibnizDer} Let $U\subseteq\rcrho^{n}$ be an
open subset in the sharp topology, let $v\in\rcrho^{n}$ and $f$,
$g:U\longrightarrow\rcrho$ be generalized smooth maps. Then
\begin{enumerate}
\item $\frac{\partial(f+g)}{\partial v}=\frac{\partial f}{\partial v}+\frac{\partial g}{\partial v}$
\item $\frac{\partial(r\cdot f)}{\partial v}=r\cdot\frac{\partial f}{\partial v}\quad\forall r\in\rcrho$
\item $\frac{\partial(f\cdot g)}{\partial v}=\frac{\partial f}{\partial v}\cdot g+f\cdot\frac{\partial g}{\partial v}$
\item For each $x\in U$, the map $\diff{f}(x).v:=\frac{\partial f}{\partial v}(x)\in\rcrho$
is $\rcrho$-linear in $v\in\rcrho^{n}$.
\item Let $V\subseteq\rcrho^{d}$ be open subsets in the sharp topology
and $h\in\gsf(V,U)$ be a generalized smooth maps. Then for all $x\in V$
and all $v\in\rcrho^{d}$ 
\begin{align*}
\frac{\partial\left(f\circ h\right)}{\partial v}(x) & =\diff{f}\left(h(x)\right).\frac{\partial h}{\partial v}(x)\\
\,\diff{\!\left(f\circ h\right)}(x) & =\diff{f}\left(h(x)\right)\circ\diff{h}(x).
\end{align*}
\end{enumerate}
\end{thm}

\subsection{\label{sec:Integral-calculus}Integral calculus using primitives}

In this section, we inquire existence and uniqueness of primitives
$F$ of a GSF $f\in\gsf([a,b],\rcrho)$. To this end, we shall have
to introduce the derivative $F'(x)$ at boundary points $x\in[a,b]$,
i.e.~such that $x-a$ or $b-x$ is not invertible. Let us note explicitly,
in fact, that the Fermat-Reyes Theorem \ref{thm:FR-forGSF} is stated
only for sharply open domains.

The following result shows that every GSF can have at most one primitive
GSF up to an additive constant.
\begin{thm}
\label{thm:uniquenessOfPrimitives}Let $X\subseteq\RC{\rho}$ and
let $f\in\gsf(X,\rcrho)$ be a generalized smooth function. Let $a$,
$b\in\rcrho$, with $a<b$, such that $(a,b)\subseteq X$. If $f'(x)=0$
for all $x\in\text{\emph{int}}(a,b)$, then $f$ is constant on $(a,b)$.
An analogous statement holds if we take any other type of interval
(closed or half closed) instead of $(a,b)$.
\end{thm}

\begin{rem}
\label{rem:I-function}From the Fermat-Reyes Thm.~\ref{thm:FR-forGSF}
and from Thm.~\ref{thm:uniquenessOfPrimitives}, it follows that
the function $i(x):=1$ if $x\approx0$ and $i(x):=0$ otherwise cannot
be a GSF on any large neighborhood of $x=0$. This example stems from
the property that different standard real numbers can always be separated
by infinitesimal balls.
\end{rem}

At interior points $x\in[a,b]$ in the sharp topology, the definition
of derivative $f^{(k)}(x)$ follows from the Fermat-Reyes Theorem
\ref{thm:FR-forGSF}. At boundary points, we have the following
\begin{thm}
\label{thm:existenceOfDerivativesAtBorderPoints}Let $a$, $b\in\rcrho$
with $a<b$, and $f\in\gsf([a,b],\rcrho)$ be a generalized smooth
function. Then for all $x\in[a,b]$, the following limit exists in
the sharp topology 
\[
\lim_{\substack{y\to x\\
y\in\text{\emph{int}}\left([a,b]\right)
}
}f^{(k)}(y)=:f^{(k)}(x).
\]
Moreover, if the net $f_{\eps}\in\cinfty(\Omega_{\eps},\R)$ defines
$f$ and $x=[x_{\eps}]$, then $f^{(k)}(x)=[f_{\eps}^{(k)}(x_{\eps})]$
and hence $f^{(k)}\in\gsf([a,b],\rcrho)$.
\end{thm}

We can now state existence and uniqueness of primitives of GSF:
\begin{thm}
\label{thm:existenceUniquenessPrimitives}Let $k\in\N\cup\{+\infty\}$
and $f\in\gckf([a,b],\rcrho)$ be defined in the interval $[a,b]\subseteq\RC{\rho}$,
where $a<b$. Let $c\in[a,b]$. Then, there exists one and only one
generalized $\mathcal{C}^{k+1}$ map $F\in\gcf{k+1}([a,b],\rcrho)$
such that $F(c)=0$ and $F'(x)=f(x)$ for all $x\in[a,b]$. Moreover,
if $f$ is defined by the net $f_{\eps}\in\mathcal{C}^{k}(\R,\R)$
and $c=[c_{\eps}]$, then $F(x)=\left[\int_{c_{\eps}}^{x_{\eps}}f_{\eps}(s)\diff{s}\right]$
for all $x=[x_{\eps}]\in[a,b]$.
\end{thm}

\begin{defn}
\label{def:integral}Under the assumptions of Theorem \ref{thm:existenceUniquenessPrimitives},
we denote by $\int_{c}^{(-)}f:=\int_{c}^{(-)}f(s)\,\diff{s}\in\gsf([a,b],\rcrho)$
the unique generalized smooth function such that:
\begin{enumerate}
\item $\int_{c}^{c}f=0$
\item $\left(\int_{c}^{(-)}f\right)'(x)=\frac{\diff{}}{\diff{x}}\int_{c}^{x}f(s)\,\diff{s}=f(x)$
for all $x\in[a,b]$.
\end{enumerate}
\end{defn}

To consider a generalization of this concept of integration to GSF
in several variables and to more general domains of integration $M\subseteq\rcrho^{d}$,
see Sec.~\ref{subsec:Multidimensional-integration} and \cite{Gio-Kun-Ver19}.
\begin{example}
~
\begin{enumerate}
\item Since $\rcrho$ contains both infinitesimal and infinite numbers,
our notion of definite integral also includes ``improper integrals''.
Let e.g.~$f(x)=\frac{1}{x}$ for $x\in\rcrho_{>0}$ and $a=1$, $b=\diff{\rho}^{-q}$,
$q>0$. Then 
\begin{equation}
\int_{a}^{b}f(s)\,\diff{s}=\left[\int_{1}^{\rho_{\eps}^{-q}}\frac{1}{s}\,\diff{s}\right]=[\log\rho_{\eps}^{-q}]-\log1=-q\log\diff{\rho},\label{eq:imprInt}
\end{equation}
which is, of course, a positive infinite generalized number. This
apparently trivial result is closely tied to the possibility to define
GSF on arbitrary domains, like $F\in\gsf([a,b],\rcrho)$ in Thm.~\ref{thm:existenceUniquenessPrimitives}
where $b$ is an infinite number as in \eqref{eq:imprInt}, which
is one of the key properties allowing one to get the closure with
respect to composition.
\item If $p$, $q\in\rti$, $p<0<q$ and both $p$ and $q$ are not infinitesimal,
then $\int_{p}^{q}\delta(t)\,\diff{t}\approx1$. If $p\le-r$ and
$q\ge s$ where $r$, $s\in\R_{>0}$, then $\int_{p}^{q}\delta(t)\,\diff t=1$.
\end{enumerate}
\end{example}

\begin{thm}
\label{thm:intRules}Let $f\in\gsf(X,\rcrho)$ and $g\in\gsf(Y,\rcrho)$
be generalized smooth functions defined on arbitrary domains in $\rcrho$.
Let $a$, $b\in\rcrho$ with $a<b$ and $[a,b]\subseteq X\cap Y$,
then
\begin{enumerate}
\item \label{enu:additivityFunction}$\int_{a}^{b}\left(f+g\right)=\int_{a}^{b}f+\int_{a}^{b}g$
\item \label{enu:homog}$\int_{a}^{b}\lambda f=\lambda\int_{a}^{b}f\quad\forall\lambda\in\rcrho$
\item \label{enu:additivityDomain}$\int_{a}^{b}f=\int_{a}^{c}f+\int_{c}^{b}f$
for all $c\in[a,b]$
\item \label{enu:chageOfExtremes}$\int_{a}^{b}f=-\int_{b}^{a}f$
\item \label{enu:fundamental}$\int_{a}^{b}f'=f(b)-f(a)$
\item \label{enu:intByParts}$\int_{a}^{b}f'\cdot g=\left[f\cdot g\right]_{a}^{b}-\int_{a}^{b}f\cdot g'$
\item \label{enu:intMonotone}If $f(x)\le g(x)$ for all $x\in[a,b]$, then
$\int_{a}^{b}f\le\int_{a}^{b}g$.
\end{enumerate}
\label{thm:changeOfVariablesInt}Let $f\in\gsf(T,\rcrho)$ and $\phi\in\gsf(S,T)$
be generalized smooth functions defined on arbitrary domains in $\rcrho$.
Let $a$, $b\in\rcrho$, with $a<b$, such that $[a,b]\subseteq S$,
$\phi(a)<\phi(b)$ and $[\phi(a),\phi(b)]\subseteq T$. Finally, assume
that $\phi([a,b])\subseteq[\phi(a),\phi(b)]$. Then 
\[
\int_{\phi(a)}^{\phi(b)}f(t)\,\diff{t}=\int_{a}^{b}f\left[\phi(s)\right]\cdot\phi'(s)\,\diff{s}.
\]
\end{thm}

\subsection{Classical theorems for GSF}

It is natural to expect that several classical theorems of differential
and integral calculus can be extended from the ordinary smooth case
to the generalized smooth framework. Once again, we underscore that
these faithful generalizations are possible because we do not have
a priori limitations in the evaluation $f(x)$ for GSF. For example,
one does not have similar results in Colombeau theory, where an arbitrary
generalized function can be evaluated only at compactly supported
points.

\noindent We start from the intermediate value theorem.
\begin{cor}
\label{cor:intermValue}Let $f\in\gsf(X,\rcrho)$ be a generalized
smooth function defined on the subset $X\subseteq\rcrho$. Let $a$,
$b\in\rcrho$, with $a<b$, such that $[a,b]\subseteq X$. Assume
that $f(a)<f(b)$. Then 
\[
\forall y\in\rcrho:\ f(a)\le y\le f(b)\ \Rightarrow\ \exists c\in[a,b]:\ y=f(c).
\]
\end{cor}

\noindent Using this theorem we can conclude that no GSF can assume
only a finite number of values which are comparable with respect to
the relation $<$ on any nontrivial interval $[a,b]\subseteq X$,
unless it is constant. For example, this provides an alternative way
of seeing that the function $i$ of Rem.~\ref{rem:I-function} cannot
be a generalized smooth map.

\noindent We note that the solution $c\in[a,b]$ of the previous generalized
smooth equation $y=f(x)$ need not even be continuous in $\eps$,
see e.g.~\cite{Gio-Kun-Ver19} for an explicit counter example. This
allows us to draw the following general conclusion: if we consider
generalized numbers as solutions of smooth equations, then we are
forced to work on a non-totally ordered ring of scalars derived from
discontinuous (in $\eps$) representatives. To put it differently:
if we choose a ring of scalars with a total order or continuous representatives,
we will not be able to solve every smooth equation, and the given
ring can be considered, in some sense, incomplete.

The next theorem deals with different version of the mean value theorem
\begin{thm}
\label{thm:classicalThms}Let $f\in\gsf(X,\rcrho^{d})$ be a generalized
smooth function defined in the sharply open set $X\subseteq\rcrho^{n}$.
Let $a$, $b\in\rcrho^{n}$ such that $[a,b]\subseteq X$. Then
\begin{enumerate}
\item \label{enu:meanValue}If $n=d=1$, then $\exists c\in[a,b]:\ f(b)-f(a)=(b-a)\cdot f'(c)$.
\item \label{enu:integralMeanValue}If $n=d=1$, then $\exists c\in[a,b]:\ \int_{a}^{b}f(t)\,\diff{t}=(b-a)\cdot f(c)$.
\item \label{enu:meanValueSevVars}If $d=1$, then $\exists c\in[a,b]:\ f(b)-f(a)=\nabla f(c)\cdot(b-a)$.
\item \label{enu:IMVSevVars}Let $h:=b-a$, then $f(a+h)-f(a)=\int_{0}^{1}\diff{f}(a+t\cdot h).h\,\diff{t}$.
\end{enumerate}
\end{thm}

Internal and sharply bounded sets generated by a net of compact sets
serve as a substitute for compact subsets for GSF, as can be seen
from the following extreme value theorem:
\begin{lem}
\label{lem:extremeValueCont}Let $\emptyset\ne K=[K_{\eps}]\subseteq\rti^{n}$
be an internal set generated by compact sets $K_{\eps}\comp\R^{n}$
such that $K$ is \emph{sharply bounded}, i.e.~$K\subseteq B_{R}(0)$
for some $R\in\rti_{>0}$. Assume that $\alpha:K\ra\rti$ is a well-defined
map given by $\alpha(x)=[\alpha_{\eps}(x_{\eps})]$ for all $x\in K$,
where $\alpha_{\eps}:K_{\eps}\ra\R$ are continuous maps (e.g.~$\alpha(x)=|x|$).
Then 
\[
\exists m,M\in K\,\forall x\in K:\ \alpha(m)\le\alpha(x)\le\alpha(M).
\]
\end{lem}

\begin{cor}
\label{cor:extremeValues}Let $f\in\gsf(X,\rcrho)$ be a generalized
smooth function defined in the subset $X\subseteq\rcrho^{n}$. Let
$\emptyset\ne K=[K_{\eps}]\subseteq X$ be as above, then 
\begin{equation}
\exists m,M\in K\,\forall x\in K:\ f(m)\le f(x)\le f(M).\label{eq:epsExtreme}
\end{equation}
\end{cor}

These results motivate the following
\begin{defn}
\label{def:functCmpt-1} A subset $K$ of $\rcrho^{n}$ is called
\emph{functionally compact}, denoted by $K\fcmp\rcrho^{n}$, if there
exists a net $(K_{\eps})$ such that
\begin{enumerate}
\item \label{enu:defFunctCmpt-internal-1}$K=[K_{\eps}]\subseteq\rcrho^{n}$
\item \label{enu:defFunctCmpt-sharpBound-1} $K$ is sharply bounded, i.e.~$\exists R\in\rti_{>0}:\ K\subseteq B_{R}(0)$
\item \label{enu:defFunctCmpt-cmpt-1}$\forall\eps\in I:\ K_{\eps}\Subset\R^{n}$
\end{enumerate}
If, in addition, $K\subseteq U\subseteq\rcrho^{n}$ then we write
$K\fcmp U$. Any net $(K_{\eps})$ such that $[K_{\eps}]=K$ is called
a \emph{representative} of $K$.
\end{defn}

\noindent We motivate the name \emph{functionally compact subset}
by noting that on this type of subsets, GSF have properties very similar
to those that ordinary smooth functions have on standard compact sets.
\begin{rem}
\noindent \label{rem:defFunctCmpt}\ 
\begin{enumerate}
\item \label{enu:rem-defFunctCmpt-closed}By Thm.~\ref{thm:strongMembershipAndDistanceComplement}.\ref{enu:internalAreClosed},
any internal set $K=[K_{\eps}]$ is closed in the sharp topology.
Therefore, functionally compact sets are sharply closed and bounded
subsets of $\rti^{n}$. In particular, the open interval $(0,1)\subseteq\Rtil$
is not functionally compact since it is not closed.
\item \label{enu:rem-defFunctCmpt-ordinaryCmpt}If $H\Subset\R^{n}$ is
a non-empty ordinary compact set, then the internal set $[H]$ is
functionally compact. In particular, $[0,1]=\left[[0,1]_{\R}\right]$
is functionally compact.
\item \label{enu:rem-defFunctCmpt-empty}The empty set $\emptyset=\widetilde{\emptyset}\fcmp\Rtil$.
\item \label{enu:rem-defFunctCmpt-equivDef}$\Rtil^{n}$ is not functionally
compact since it is not sharply bounded.
\item \label{enu:rem-defFunctCmpt-cmptlySuppPoints}The set of finite points
$\csp{\R}$ is not functionally compact because the GSF $f(x)=x$
does not satisfy the conclusion \eqref{eq:epsExtreme} of Cor.~\ref{cor:extremeValues}.
\end{enumerate}
\end{rem}

\noindent We also underscore the following properties of functionally
compact sets.
\begin{thm}
\label{thm:image}Let $K\subseteq X\subseteq\Rtil^{n}$, $f\in\gckf(X,\Rtil^{d})$.
Then $K\fcmp\Rtil^{n}$ implies $f(K)\fcmp\Rtil^{d}$.
\end{thm}

\noindent As a corollary of this theorem and Rem.\ \eqref{rem:defFunctCmpt}.\ref{enu:rem-defFunctCmpt-ordinaryCmpt}
we get
\begin{cor}
\label{cor:intervalsFunctCmpt}If $a$, $b\in\Rtil$ and $a\le b$,
then $[a,b]\fcmp\Rtil$.
\end{cor}

\noindent Let us note that $a$, $b\in\Rtil$ can also be infinite
numbers, e.g.~$a=\diff{\rho}^{-N}$, $b=\diff{\rho}^{-M}$ or $a=-\diff{\rho}^{-N}$,
$b=\diff{\rho}^{-M}$ with $M$, $N\in\N_{>0}$, so that e.g.~$[-\diff{\rho}^{-N},\diff{\rho}^{M}]\supseteq\R$.
Therefore, despite very similar properties shared by functionally
compact sets and classical compact sets, the former can also be unbounded
from the classical point of view.

\noindent Finally, in the following result we consider the product
of functionally compact sets:
\begin{thm}
\noindent \label{thm:product}Let $K\fcmp\Rtil^{n}$ and $H\fcmp\Rtil^{d}$,
then $K\times H\fcmp\Rtil^{n+d}$. In particular, if $a_{i}\le b_{i}$
for $i=1,\ldots,n$, then $\prod_{i=1}^{n}[a_{i},b_{i}]\fcmp\Rtil^{n}$.
\end{thm}

A theory of compactly supported GSF has been developed in \cite{Gio-Kun18a},
and it closely resembles the classical theory of LF-spaces of compactly
supported smooth functions. It establishes that for suitable functionally
compact subsets, the corresponding space of compactly supported GSF
contains extensions of all Colombeau generalized functions, and hence
also of all Schwartz distributions.

Note also that any interval $[a,b]\subseteq\rti$ with $b-a\in\R_{>0}$,
is functionally compact but \emph{not} connected: in fact if $c\in(a,b)$,
then both $c+D_{\infty}$ and $[a,b]\setminus\left(c+D_{\infty}\right)$
are sharply open in $[a,b]$. Once again, this is a general property
in several non-Archimedean frameworks (see e.g.~\cite{Rob73,Koc}).
On the other hand, as in the case of functionally compact sets, GSF
behave on intervals as if they were connected, in the sense that both
the intermediate value theorem Cor.~\ref{cor:intermValue} and the
extreme value theorem Cor.~\ref{cor:extremeValues} hold for them
(therefore, $f\left([a,b]\right)=\left[f(m),f(M)\right]$, where we
used the notations from the results just mentioned).

We close this section with generalizations of Taylor's theorem in
various forms. In the following statement, $\diff{^{k}f}(x):\rti^{dk}\ra\rti$
is the $k$-th differential of the GSF $f$, viewed as an $\rti$-multilinear
map $\rti^{d}\times\ptind^{k}\times\rti^{d}\ra\rti$, and we use the
common notation $\diff{^{k}f}(x)\cdot h^{k}:=\diff{^{k}f}(x)(h,\ldots,h)$.
Clearly, $\diff{^{k}f}(x)\in\gsf(\rti^{dk},\rti)$. For multilinear
maps $A:\rti^{p}\ra\rti^{q}$, we set $|A|:=[|A_{\eps}|]\in\rti$,
the generalized number defined by the norms of the operators $A_{\eps}:\R^{p}\ra\R^{q}$.
\begin{thm}
\label{thm:Taylor-1-1}Let $f\in\gsf(U,\rcrho)$ be a generalized
smooth function defined in the sharply open set $U\subseteq\rcrho^{d}$.
Let $a$, $b\in\rcrho^{d}$ such that the line segment $[a,b]\subseteq U$,
and set $h:=b-a$. Then, for all $n\in\N$ we have
\begin{enumerate}
\item \label{enu:LagrangeRest-1-1}$\exists\xi\in[a,b]:\ f(a+h)=\sum_{j=0}^{n}\frac{\diff{^{j}f}(a)}{j!}\cdot h^{j}+\frac{\diff{^{n+1}f}(\xi)}{(n+1)!}\cdot h^{n+1}.$
\item \label{enu:integralRest-1-1}$f(a+h)=\sum_{j=0}^{n}\frac{\diff{^{j}f}(a)}{j!}\cdot h^{j}+\frac{1}{n!}\cdot\int_{0}^{1}(1-t)^{n}\,\diff{^{n+1}f}(a+th)\cdot h^{n+1}\,\diff{t}.$
\end{enumerate}
\noindent Moreover, there exists some $R\in\rcrho_{>0}$ such that
\begin{equation}
\forall k\in B_{R}(0)\,\exists\xi\in[a,a+k]:\ f(a+k)=\sum_{j=0}^{n}\frac{\diff{^{j}f}(a)}{j!}\cdot k^{j}+\frac{\diff{^{n+1}f}(\xi)}{(n+1)!}\cdot k^{n+1}\label{eq:LagrangeInfRest-1-1}
\end{equation}
\begin{equation}
\frac{\diff{^{n+1}f}(\xi)}{(n+1)!}\cdot k^{n+1}=\frac{1}{n!}\cdot\int_{0}^{1}(1-t)^{n}\,\diff{^{n+1}f}(a+tk)\cdot k^{n+1}\,\diff{t}\approx0.\label{eq:integralInfRest-1-1}
\end{equation}
\end{thm}

Formulas \ref{enu:LagrangeRest-1-1} and \ref{enu:integralRest-1-1}
correspond to a plain generalization of Taylor's theorem for ordinary
smooth functions with Lagrange and integral remainder, respectively.
Dealing with GF, it is important to note that this direct statement
also includes the possibility that the differential $\diff{^{n+1}f}(\xi)$
may be infinite at some point. For this reason, in \eqref{eq:LagrangeInfRest-1-1}
and \eqref{eq:integralInfRest-1-1}, considering a sufficiently small
increment $k$, we get more classical infinitesimal remainders $\diff{^{n+1}f}(\xi)\cdot k^{n+1}\approx0$.

The following definitions allow us to state Taylor formulas in Peano
and in infinitesimal form. The latter has no remainder term thanks
to the use of an equivalence relation that permits the introduction
of a language of nilpotent infinitesimals, see e.g.~\cite{Gio10a,Gio10}
for a similar formulation. For simplicity, we only present the 1-dimentional
case.
\begin{defn}
\label{def:little-oh-equality_k-1-1}
\begin{enumerate}
\item \label{enu:little-oh-1-1}Let $U\subseteq\rcrho$ be a sharp neighborhood
of $0$ and $P$, $Q:U\longrightarrow\rcrho$ be maps defined on $U$.
Then we say that 
\[
P(u)=o(Q(u))\quad\text{as }u\to0
\]
if there exists a function $R:U\longrightarrow\rcrho$ such that 
\[
\forall u\in U:\ P(u)=R(u)\cdot Q(u)\quad\text{and}\quad\lim_{u\to0}R(u)=0,
\]
where the limit is taken in the sharp topology.
\item \label{enu:equalityUpTo-k-thOrderInf-1-1}Let $x$, $y\in\rcrho$
and $k$, $j\in\R_{>0}$, then we write $x=_{j}y$ if there exist
representatives $(x_{\eps})$, $(y_{\eps})$ of $x$, $y$, respectively,
such that 
\begin{equation}
|x_{\eps}-y_{\eps}|=O(\rho_{\eps}^{\frac{1}{j}}).\label{eq:equalityUpTo_j-1-1}
\end{equation}
We will read $x=_{j}y$ as $x$ \emph{is equal to} $y$ \emph{up to}
$j$\emph{-th order infinitesimals}. Finally, if $k\in\N_{>0}$, we
set $D_{kj}:=\left\{ x\in\rcrho\mid x^{k+1}=_{j}0\right\} $, which
is called the \emph{set of} $k$\emph{-th order infinitesimals for
the equality $=_{j}$}, and 
\[
D_{\infty j}:=\left\{ x\in\rcrho\mid\exists k\in\N_{>0}:\ x^{k+1}=_{j}0\right\} 
\]
which is called the \emph{set of infinitesimals for the equality $=_{j}$}.
\end{enumerate}
\end{defn}

Of course, the reformulation of Def.~\ref{def:little-oh-equality_k-1-1}
\ref{enu:little-oh-1-1} for the classical Landau's little-oh is particularly
suited to the case of a ring like $\rcrho$, instead of a field. The
intuitive interpretation of $x=_{j}y$ is that for particular (e.g.~physics-related)
problems one is not interested in distinguishing quantities whose
difference $|x-y|$ is less than an infinitesimal of order $j$. In
fact, if $x=_{j}y$ we can write $x_{\eps}=y_{\eps}+r_{\eps}$ with
$r_{\eps}\to0$ of order at most $\rho_{\eps}^{\frac{1}{j}}$. The
idea behind taking $\frac{1}{j}$ in \eqref{eq:equalityUpTo_j-1-1}
is to obtain the property that the greater the order $j$ of the infinitesimal
error, the greater the difference $|x-y|$ is allowed to be. This
is a typical property in rings with nilpotent infinitesimals (see
e.g.~\cite{Gio10a,Koc}). The set $D_{ki}$ represents the neighborhood
of infinitesimals of $k$-th order for the equality $=_{j}$. Once
again, the greater the order $k$, the bigger is the neighborhood
(see Theorem \ref{thm:TaylorPeano-Infinitesimals}.\ref{enu:chainOfInfNeighborhoods-1}
below). Note that if $x=_{j}y$, then $x_{\eps}=y_{\eps}+o\left(\rho_{\eps}^{\frac{1}{j}-a}\right)$
for all $a\in(0,1/j]_{\R}$. In particular, $x_{\eps}=y_{\eps}+o\left(\rho_{\eps}\right)$
implies $x=_{1}y$, whereas $x=_{1}y$ yields only $x_{\eps}=y_{\eps}+o\left(\rho_{\eps}^{1-a}\right)$
for all $a\in(0,1]_{\R}$. Finally, note that $x=_{j}y$ is equivalent
to $|x-y|\le C\diff\rho^{\frac{1}{j}}$ for some $C\in\R_{\ge0}$.
\begin{thm}
\label{thm:TaylorPeano-Infinitesimals}Let $f\in\gsf(U,\rcrho)$ be
a generalized smooth function defined in the sharply open set $U\subseteq\rcrho$.
Let $x$, $\delta\in\rcrho$, with $\delta>0$ and $[x-\delta,x+\delta]\subseteq U$.
Let $k$, $l$, $j\in\R_{>0}$. Then
\begin{enumerate}
\item \label{enu:Peano-1}$\forall n\in\N:\ f(x+u)=\sum_{r=0}^{n}\frac{f^{(r)}(x)}{r!}u^{r}+o(u^{n})$
as $u\to0$.
\item \label{enu:indepedRepr-1}The definition of $x=_{j}y$ does not depend
on the representatives of $x$, $y$.
\item \label{enu:equivalenceRel-1}$=_{j}$ is an equivalence relation on
$\rcrho$.
\item \label{enu:relationBetweenD_j-And-D_l}If $x=_{j}y$ and $l\ge j$,
then $x=_{l}y$. Therefore, $D_{nj}\subseteq D_{nl}$.
\item \label{enu:relationDandD_j}If $x=_{j}y$ for all $j\in\R_{>0}$ sufficiently
small, then $x=y$.
\item \label{enu:prodNilp}If $x=_{j}y$ and $z=_{j}w$ then $x+z=_{j}y+w$.
If $x$ and $z$ are finite, then $x\cdot z=_{j}y\cdot w$.
\item \label{enu:func=00003D_j}If $x=_{j}y$, $f\in\gsf([a,b],\rti)$,
$x$, $y\in[a,b]$, and $f'(c)$ is finite for all $c\in[a,b]$, then
$f(x)=_{j}f(y)$.
\item \label{enu:D_k-infinitesimal-1}$\forall h\in D_{kj}:\ h\approx0$.
\item \label{enu:chainOfInfNeighborhoods-1}$D_{mj}\subseteq D_{kj}\subseteq D_{\infty j}$
if $m\le k$.
\item \label{enu:D_j-AlmostIdeal-1}$D_{kj}$ is a subring of $\rcrho$.
For all $h\in D_{kj}$ and all finite $x\in\rcrho$, we have $x\cdot h\in D_{kj}$.
\item \label{enu:infTaylor}Let $n\in\N_{>0}$ and assume that $j$, k and
$f$ satisfy
\begin{equation}
\forall z\in\rcrho\,\forall\xi\in[x-\delta,x+\delta]:\ z=_{j}0\ \Rightarrow\ z\cdot f^{(n+1)}(\xi)=_{k}0.\label{eq:relBetween-i-AndInfinite-f^(n+1)}
\end{equation}
Then, we have 
\[
\forall u\in D_{nj}:\ f(x+u)=_{k}\sum_{r=0}^{n}\frac{f^{(r)}(x)}{r!}u^{r}.
\]
\item \label{enu:infTaylor2}For all $n\in\N_{>0}$ there exist $e\in\R_{>0}$
such that $e\le j$, and $\forall u\in D_{ne}:\ f(x+u)=_{j}\sum_{r=0}^{n}\frac{f^{(r)}(x)}{r!}u^{r}$.
\end{enumerate}
\end{thm}

\noindent We shall use the nilpotent Taylor formula \ref{enu:infTaylor2}
in Sec.~\ref{sec:Formal-deductions} for the deduction of the heat
and wave equation for GSF; we therefore note here that the index $e$
depends on the GSF $f$: in that case, we say that \emph{the nilpotent
Taylor formula of order $n$ holds for $f$ on $D_{ne}$}. From \ref{enu:relationBetweenD_j-And-D_l}
it hence follows that it also holds on $D_{ne'}$ for all $e'\le e$.

\subsection{\label{subsec:Multidimensional-integration}Multidimensional integration}

Finally, we can also introduce multidimensional integration of GSF
on suitable subsets of $\rcrho^{n}$.
\begin{defn}
\label{def:intOverCompact}Let $\mu$ be a measure on $\R^{n}$ and
let $K$ be a functionally compact subset of $\RC{\rho}^{n}$. Then,
we call $K$ $\mu$-\emph{measurable} if the limit 
\begin{equation}
\mu(K):=\lim_{m\to\infty}[\mu(\overline{\Eball}_{\rho_{\eps}^{m}}(K_{\eps}))]\label{eq:muMeasurable}
\end{equation}
exists for some representative $(K_{\eps})$ of $K$. Here $m\in\N$,
the limit is taken in the sharp topology on $\RC{\rho}$ since $[\mu(\overline{\Eball}_{\rho_{\eps}^{m}}(K_{\eps}))]\in\rcrho$,
and $\overline{\Eball}_{r}(A):=\{x\in\R^{n}:d(x,A)\le r\}$.
\end{defn}

\noindent In the following result, we show that this definition generates
a correct notion of multidimensional integration for GSF.
\begin{thm}
\label{thm:muMeasurableAndIntegral}Let $K\subseteq\RC{\rho}^{n}$
be $\mu$-measurable.
\begin{enumerate}
\item \label{enu:indepRepr}The definition of $\mu(K)$ is independent of
the representative $(K_{\eps})$.
\item \label{enu:existsRepre}There exists a representative $(K_{\eps})$
of $K$ such that $\mu(K)=[\mu(K_{\eps})]$.
\item \label{enu:epsWiseDefInt}Let $(K_{\eps})$ be any representative
of $K$ and let $f\in\gsf(K,\RC{\rho})$ be a GSF defined by the net
$(f_{\eps})$. Then 
\[
\int_{K}f\,\diff{\mu}:=\lim_{m\to\infty}\biggl[\int_{\overline{\Eball}_{\rho_{\eps}^{m}}(K_{\eps})}f_{\eps}\,\diff{\mu}\biggr]\in\rcrho
\]
exists and its value is independent of the representative $(K_{\eps})$.
\item \label{enu:existsReprDefInt}There exists a representative $(K_{\eps})$
of $K$ such that 
\begin{equation}
\int_{K}f\,\diff{\mu}=\biggl[\int_{K_{\eps}}f_{\eps}\,\diff{\mu}\biggr]\in\rcrho\label{eq:measurable}
\end{equation}
for each $f\in\gsf(K,\RC{\rho})$.
\item If $K=\prod_{i=1}^{n}[a_{i},b_{i}]$, then $K$ is $\lambda$-measurable
($\lambda$ being the Lebesgue measure on $\R^{n}$) and for all $f\in\gsf(K,\RC{\rho})$
we have
\[
\int_{K}f\,\diff{\lambda}=\left[\int_{a_{1,\eps}}^{b_{1,\eps}}\,dx_{1}\dots\int_{a_{n,\eps}}^{b_{n,\eps}}f_{\eps}(x_{1},\dots,x_{n})\,\diff{x_{n}}\right]\in\rcrho
\]
for any representatives $(a_{i,\eps})$, $(b_{i,\eps})$ of $a_{i}$
and $b_{i}$ respectively. Therefore, if $n=1$, this notion of integral
coincides with that of Thm.~\ref{thm:existenceUniquenessPrimitives}
and Def.~\ref{def:integral}.
\item Let $K\subseteq\RC{\rho}^{n}$ be $\lambda$-measurable, where $\lambda$
is the Lebesgue measure, and let $\phi\in\gsf(K,\RC{\rho}^{d})$ be
such that $\phi^{-1}\in\gsf(\phi(K),\RC{\rho}^{n})$. Then $\phi(K)$
is $\lambda$-measurable and 
\[
\int_{\phi(K)}f\,\diff{\lambda}=\int_{K}(f\circ\phi)\left|\det(\diff{\phi})\right|\,\diff{\lambda}
\]
for each $f\in\gsf(\phi(K),\RC{\rho})$.
\end{enumerate}
\end{thm}

\section{\label{sec:Differential-equations-1}Differential equations: the
Picard-Lindelöf theorem for ODE}

As in the classical case, thanks to the extreme value Lem.~\ref{lem:extremeValueCont}
and the properties of functionally compact sets $K$, we can naturally
define a topology on the space $\gckf(K,\rcrho^{d})$:
\begin{defn}
\label{def:genNormsSpaceGSF}Let $K\fcmp\rcrho^{n}$ be a functionally
compact set such that $K=\overline{\accentset{\circ}{K}}$ (so that
partial derivatives at sharply boundary points can be defined as limits
of partial derivatives at sharply interior points; such $K$ are called
\emph{solid} sets). Let $l\in\N_{\le k}$ and $v\in\gckf(K,\rcrho^{d})$.
Then
\[
\Vert v\Vert_{l}:=\max_{\substack{|\alpha|\le l\\
1\le i\le d
}
}\max\left(\left|\partial^{\alpha}v^{i}(M_{ni})\right|,\left|\partial^{\alpha}v^{i}(m_{ni})\right|\right)\in\rcrho,
\]
where $M_{ni}$, $m_{ni}\in K$ satisfy
\[
\forall x\in K:\ \partial^{\alpha}v^{i}(m_{ni})\le\partial^{\alpha}v^{i}(x)\le\partial^{\alpha}v^{i}(M_{ni}).
\]
\end{defn}

\noindent The following result permits us to calculate the (generalized)
norm $\Vert v\Vert_{l}$ using any net $(v_{\eps})$ that defines
$v$.
\begin{lem}
\label{lem:normSpaceGSF}Under the assumptions of Def.~\ref{def:genNormsSpaceGSF},
let $[K_{\eps}]=K\fcmp\rcrho^{n}$ be any representative of $K$.
Then we have:
\begin{enumerate}
\item \label{enu:normAndDefNet}If the net $(v_{\eps})$ defines $v$, then
$\Vert v\Vert_{l}=\left[\max_{\substack{|\alpha|\le l\\
1\le i\le d
}
}\max_{x\in K_{\eps}}\left|\partial^{\alpha}v_{\eps}^{i}(x)\right|\right]\in\rcrho$;
\item \label{enu:normPos}$\Vert v\Vert_{l}\ge0$;
\item $\Vert v\Vert_{l}=0$ if and only if $v=0$;
\item $\forall c\in\rcrho:\ \Vert c\cdot v\Vert_{l}=|c|\cdot\Vert v\Vert_{l}$;
\item \label{enu:normTriang}For all $u\in\gckf(K,\rcrho^{d})$, we have
$\Vert u+v\Vert_{l}\le\Vert u\Vert_{l}+\Vert v\Vert_{l}$ and $\Vert u\cdot v\Vert_{l}\le c_{l}\cdot\Vert u\Vert_{l}\cdot\Vert v\Vert_{l}$
for some $c_{l}\in\rcrho_{>0}$.
\end{enumerate}
\end{lem}

\noindent Using these $\rcrho$-valued norms, we can naturally define
a topology on the space $\gckf(K,\rcrho^{d})$.
\begin{defn}
\label{def:sharpTopSpaceGSF}Let $K\fcmp\rcrho^{n}$ be a solid set.
Let $l\in\N_{\le k}$, $u\in\gckf(K,\rcrho^{d})$, $r\in\rcrho_{>0}$,
then
\begin{enumerate}
\item $B_{r}^{l}(u):=\left\{ v\in\gckf(K,\rcrho^{d})\mid\Vert v-u\Vert_{l}<r\right\} $
\item If $U\subseteq\gckf(K,\rcrho^{d})$, then we say that $U$ is a \emph{sharply
open set} if 
\[
\forall u\in U\,\exists l\in\N_{\le k}\,\exists r\in\rcrho_{>0}:\ B_{r}^{l}(u)\subseteq U.
\]
\end{enumerate}
\end{defn}

\noindent One can easily prove that sharply open sets form a sequentially
Cauchy complete topology on $\gckf(K,\rcrho^{d})$, see e.g.~\cite{Gio-Kun18a,LuGi20b}.
The structure $\left(\gckf(K,\rcrho^{d}),\left(\norm{-}_{l}\right)_{l\le k}\right)$
has the usual properties of a graded Fréchet space if we replace everywhere
the field $\R$ with the ring $\rcrho$, and for this reason it is
called an $\rcrho$-graded Fréchet space.

The Banach fixed point theorem can be easily generalized to spaces
of generalized continuous functions with the sup-norm $\Vert-\Vert_{0}$
(see Def.~\ref{def:genNormsSpaceGSF}). As a consequence, we have
the following Picard-Lindelöf theorem for ODE in the $\gckf$ setting,
see also \cite{ErlGross,LuGi20b}.
\begin{thm}
\label{thm:PicLindFiniteContr}Let $t_{0}\in\rti$, $y_{0}\in\RC{\rho}^{d}$,
$\alpha$, $r\in\rti_{>0}$. Let $F\in\gckf([t_{0}-\alpha,t_{0}+\alpha]\times\overline{B_{r}(y_{0})},\rti^{d})$.
Set $M:={\displaystyle \max_{\substack{t_{0}-\alpha\le t\le t_{0}+\alpha\\
|y-y_{0}|\le r
}
}}|F(t,y)|$, $L:={\displaystyle \max_{\substack{t_{0}-\alpha\le t\le t_{0}+\alpha\\
|y-y_{0}|\le r
}
}}\left|\partial_{y}F(t,y)\right|\in\rti$ and assume that
\begin{align*}
\alpha\cdot M & \le r,
\end{align*}
\begin{equation}
\lim_{n\to+\infty}\alpha^{n}L^{n}=0,\label{eq:limitAssPL}
\end{equation}
where the limit in \eqref{eq:limitAssPL} is clearly taken in the
sharp topology. Then there exists a unique solution $y\in\gcf{k+1}\left([t_{0}-\alpha,t_{0}+\alpha],\rti^{d}\right)$
of the Cauchy problem
\begin{equation}
\begin{cases}
y'(t)=F(t,y(t))\\
y(t_{0})=y_{0}.
\end{cases}\label{eq:ODE}
\end{equation}
This solution is given by
\begin{align*}
y & =\lim_{n\to+\infty}P^{n}(y_{0})\\
P(y)(t): & =y_{0}+\int_{t_{0}}^{t}F(s,y(s))\,\diff{s}\quad\forall t\in[t_{0}-\alpha,t_{0}+\alpha],
\end{align*}
and for all $n\in\N$ satisfies $\Vert y-P^{n}(y_{0})\Vert_{0}\le\alpha M\sum_{k=n}^{+\infty}\frac{\alpha^{n}L^{n}}{n!}$
and $\Vert y-y_{0}\Vert_{0}\le r$.
\end{thm}

Finally, we have the following Grönwall-Bellman inequality in integral
form:
\begin{thm}
\label{thm:Groenwall}Let $\alpha\in\RC{\rho}_{>0}$. Let $u$, $a$,
$b\in\gckf\left([0,\alpha],\RC{\rho}\right)$ and assume that $\left\Vert a\right\Vert _{0}\cdot\alpha<N\cdot\log\left(\diff{\rho}^{-1}\right)$
for some $N\in\mathbb{N}$. Assume that $a(t)\geq0$ for all $t\in[0,\alpha]$,
and that $u(t)\leq b(t)+\int_{0}^{t}a(s)u(s)\,\diff{s}$. Then
\begin{enumerate}
\item \label{enu: Gronw 1}For every $t\in[0,\alpha]$ we have
\[
u(t)\leq b(t)+\int_{0}^{t}a(s)b(s)e^{\int_{s}^{t}a(r)\,\diff{r}}\,\diff{s}.
\]
\item \label{enu: Gronw 2}If $b(t)\le b(s)$ for all $t\le s$, i.e.~if
$b$ is non-decreasing, then for every $t\in[0,\alpha]$ we have
\[
u(t)\leq b(t)e^{\int_{0}^{t}a(s)\,\diff{s}}.
\]
\end{enumerate}
\end{thm}

Finally, the following theorem considers global solutions of homogeneous
linear ODE:
\begin{thm}[Solution of homogeneous linear ODE]
\label{thm:linearODE}Let $A\in\gsf([a,b],\rti^{d\times d})$, where
$a$, $b\in\rti$, $a<b$, and $t_{0}\in[a,b]$, $y_{0}\in\rti^{d}$.
Assume that 
\begin{equation}
\left|\int_{t_{0}}^{t}A(s)\diff{s}\right|\le-C\cdot\log\diff\rho\quad\forall t\in[a,b],\label{eq:logHP}
\end{equation}
where $C\in\R_{>0}$. Then there exists one and only one $y\in\gsf([a,b],\rti^{d})$
such that 
\begin{equation}
\begin{cases}
y'(t)=A(t)\cdot y(t) & \text{if }t\in[a,b]\\
y(t_{0})=y_{0}
\end{cases}\label{eq:LinearODE}
\end{equation}
Moreover, this $y$ is given by $y(t)=\exp\left(\int_{t_{0}}^{t}A(s)\diff{s}\right)\cdot y_{0}$
for all $t\in[a,b]$.
\end{thm}

In general, the solution of a differential equation in a non-Archimedean
setting is defined on an infinitesimal neighborhood of the initial
condition. This is a general fact of every non-Archimedean theory
having at least one positive and invertible infinitesimal $h$. If
fact, the Cauchy problem
\begin{equation}
\begin{cases}
y'=-\frac{t}{1+y}\cdot\frac{1}{h}\\
y(0)=0
\end{cases}\label{eq:infinitesimalSolution}
\end{equation}
has solution $y(t)=-1+\sqrt{1-\frac{t^{2}}{h}}$ which is defined
and smooth only in the infinitesimal interval $(-\sqrt{h},\sqrt{h})$.
Moreover, we have that $\lim_{t\to\pm\sqrt{h}}y'(t)=+\infty$ (in
the sharp topology) and this clearly shows that the solution cannot
be extended. However, very general sufficient conditions to have non-infinitesimal
domains can be proved, considering e.g.~the case where the right
hand side $F$ in \eqref{eq:ODE} is an ordinary smooth function,
or when we extend the theory of Picard iterations $P^{n}$ to an infinite
natural number $n=[n_{\eps}]\in\rti$, $n_{\eps}\in\N$, see \cite{LuGi20b}.
We also finally state that a very general Picard-Lindelöf theorem
can also be proved for PDE, see \cite{GiLu20a,Gio-L22,DeMO}.

\section{\label{sec:Formal-deductions}Formal deductions corresponding to
informal reasonings}

In the previous sections, we reviewed GSF theory and we hope we persuaded
the reader that a meaningful and sufficiently complete theory containing
infinitesimal and infinite numbers is possible. This non-Archimedean
theory does not require any background in mathematical logic, has
clear connections with the usual standard calculus, is intuitively
clear, but also solves non trivial problems such as the possibility
to consider generalized functions with infinite derivatives, making
non-linear operations on Schwartz distributions and sharing several
results of ordinary smooth functions.

Now, the framework of GSF theory allows one to formalize several informal
reasonings with the intuitive use of infinitesimal and infinite numbers
we can find in physics, engineering and even in mathematics. The main
goal is absolutely not the empty searching for the mathematical rigour,
but the learning of the true rules of infinitesimal calculus instead
of unclear foggy explanations and, mainly, the flexibility to create
new and simpler mathematical models of real-world problems. As a trivial
example, using the Taylor formula with nilpotent infinitesimals Thm.~\ref{thm:TaylorPeano-Infinitesimals},
if $\frac{v^{2}}{c^{2}}\in D_{1j}$, we can write \eqref{eq:EinsteinInfinitesimal}
as $1/\sqrt{1-{\displaystyle v^{2}/c^{2}}}=_{j}1+\frac{v^{2}}{2c^{2}}$
for all $j\in\R_{>0}$ and Einstein calculations remain essentially
unchanged. In the next sections, we will see that this method not
only allows one to obtain a rigorous version of the usual informal
deductions of the heat and wave equations, but that these same proofs
show the validity of these equations for GSF, opening new applications
for example to optics of different materials and geophysics.

A frequently underestimated consequence of seeing generalized functions,
e.g.~any Schwartz distribution $T$, as set-theoretical functions
is that pointwise values $T(x_{0})$ are now always well-defined.
Therefore, non-linear boundary value problems are now conceivable
(see e.g.~\eqref{eq:ODE}), and this is a solution of a non trivial
drawback of Schwartz theory having important consequences for mathematical
modeling.

\subsection{Derivation of the heat equation for GSF}

In this section, we derive the heat equation in a similar way to \cite{Vla71,Gio10},
with the difference that here we extend the applicability to GSF and
not only to smooth functions. Let $(\vec{e}_{1},\vec{e}_{2},\vec{e}_{3})$
denotes the standard basis of $\mathbb{R}^{3}$, so that any vector
$a\in\rcrho^{3}$ is of the form $a=\lambda_{1}\cdot\vec{e}_{1}+\lambda_{2}\cdot\vec{e}_{2}+\lambda_{3}\cdot\vec{e}_{3}$
for $\lambda_{1}$, $\lambda_{2}$, $\lambda_{3}\in\rti$. In the
following, a symbol of the form $\delta y\in\rcrho$ intuitively means
that the infinitesimal increment $\delta y$ is associated to the
variable $y$.

Let us consider a body $B\subseteq\rti^{3}$ represented by a solid
set, i.e.~$B=\overline{\accentset{\circ}{B}}$, so that values of
GSF on the boundary of $B$ can be computed as limits of values at
interior points. We consider the following GSF:
\begin{itemize}
\item $\varrho:B\to\rcrho$ (mass density) ,
\item $c:B\to\rcrho$ (heat capacity),
\item $k:B\to\rcrho$ (thermal conductivity coefficient).
\end{itemize}
\noindent Note that we do not make any assumptions on the favoured
directions of these functions on their domain $B$. This assumption
corresponds to the isotropy condition for $B$. The next GSF we need
represents the temperature of the body $B$ at each point $x\in B$
and time $t\in[0,\infty)$ and is denoted by $u:B\times[0,\infty)\to\rcrho$.

\noindent We choose an interior point $x\in\accentset{\circ}{B}$
and an infinitesimal volume $V\subset\rcrho^{3}$ of the form

\noindent 
\begin{equation}
V=V(x,\delta\bar{x})=\{y\in\rcrho^{3}|-\delta x_{i}\leq2(y-x)\cdot\vec{e}_{i}\leq\delta x_{i}\ \forall i=1,2,3\},
\end{equation}
where $\delta x_{i}\in\rti_{>0}$ and $\delta\bar{x}:=(\delta x_{1},\delta x_{2},\delta x_{3})$.
Such a set is said to be an infinitesimal parallelepiped if $\delta v:=\delta x_{1}\cdot\delta x_{2}\cdot\delta x_{3}\approx0$,
that is, if the corresponding volume is infinitesimal. Note that since
$x\in\accentset{\circ}{B}$, we have $\exists\delta\bar{x}\in\rti_{>0}^{3}:\ V=V(x,\delta\bar{x})\subseteq B$,
and hence we can view $V$ as the subbody of $B$ corresponding to
the infinitesimal parallelepiped centered at $x$ with sides parallel
to the coordinate axes. This subbody interacts thermally with its
complement $CV:=B\setminus V$ and with external heat sources. In
this type of deductions, the physical part frequently consists, from
the mathematical point of view, in physically meaningful definitions
or assumptions corresponding to physical principles or constitutive
relations. For example, we now recall Fourier's law, which states
that during the infinitesimal time interval $\delta t$ the heat $Q_{CV,V}$
flowing perpendicularly to the surface of $V$ defines the exchange
between $V$ and $CV$, and this yields the following

\noindent 
\begin{align}
Q_{CV,V} & :=Q_{CV,V}(x,t,\delta t,\delta\bar{x})\label{eq:QCVV}\\
 & =\delta t\cdot\sum_{i=1}^{3}\delta s_{i}\cdot[k(x+\vec{\delta h_{i}})\cdot\dfrac{\partial u}{\partial x_{i}}(x+\vec{\delta h_{i}},t)-k(x-\vec{\delta h_{i}})\cdot\dfrac{\partial u}{\partial x_{i}}(x-\vec{\delta h_{i}},t)],
\end{align}

\noindent where $\vec{\delta h_{i}}:=\dfrac{1}{2}\delta x_{i}\cdot\vec{e_{i}}\in\rcrho^{3}$
and $\delta s_{i}:=\prod_{j\neq i}\delta x_{j}\in\rcrho.$ Note explicitly
that $Q_{CV,V}$ depends on $x$, $t$, $\delta t$, $\delta x_{i}$.
The heat exchange of $V$ due to thermal interactions with external
sources is given by the expression

\noindent 
\begin{equation}
Q_{\text{ext},V}:=Q_{\text{ext},V}(x,t,\delta t,\delta\bar{x})=F(x,t)\cdot\delta v\cdot\delta t,\label{eq:QextV}
\end{equation}

\noindent where $F(x,t):B\to{}^{\rho}\tilde{\mathbb{R}}$ is a GSF
representing the intensity of the thermal sources. The total heat
is $Q_{CV,V}+Q_{\text{ext},V}$ and it corresponds to the increment
$u(x,t+\delta t)-u(x,t)$ of the temperature of $V$ and hence to
an exchange of heat with the environment $Q_{\text{env},V}$ that
reads

\noindent 
\begin{align}
Q_{\text{env},V} & :=Q_{\text{env},V}(x,t,\delta t,\delta\bar{x})=[u(x,t+\delta t)-u(x,t)]\cdot c(x)\cdot\varrho(x)\cdot\delta v,\label{eq:QenvV}\\
 & =Q_{CV,V}+Q_{\text{ext},V}.\label{eq:heatFluxes}
\end{align}

\noindent We now want to apply the first order nilpotent Taylor formula
Thm.~\ref{thm:TaylorPeano-Infinitesimals}.\ref{enu:infTaylor2},
at \eqref{eq:QCVV} and \eqref{eq:QenvV}, i.e.~at the GSF $k$,
$\frac{\partial u}{\partial x_{i}}(-,t)$ and $u(x,-)$. From \ref{enu:infTaylor2}
and \ref{enu:relationBetweenD_j-And-D_l} of Thm.~\ref{thm:TaylorPeano-Infinitesimals},
if these formulas hold respectively on $D_{1e'}$, $D_{1e''}$ and
$D_{1\bar{e}}$, then they also hold on $D_{1e}$, where $e=\min(e',e'',\bar{e},j)$.
We choose our infinitesimals in such a way that $\delta v\cdot\delta t\in D_{1e}$,
$\delta t\cdot\delta s_{i}\cdot(\delta x_{i})^{2}=_{j}0$ and $(\delta t)^{2}\delta v=_{j}0$.
Using these infinitesimals, second order terms using nilpotent Taylor
formula Thm.~\ref{thm:TaylorPeano-Infinitesimals}.\ref{enu:infTaylor}
in \eqref{eq:QCVV} and \eqref{eq:QenvV} will not give a contribution
if we use the equality $=_{j}$. We will see later that infinitesimals
$\delta t$ and $\delta x_{i}$ satisfying all the needed conditions
actually exist.

\noindent This allows us to rewrite \eqref{eq:QCVV} and \eqref{eq:QenvV}
as follows

\noindent 
\begin{align}
Q_{CV,V} & =_{j}\text{div}[k\cdot\text{grad}(u)](x,t)\cdot\delta v\cdot\delta t,\label{eq:QCV-V}\\
Q_{\text{env},V} & =_{j}c(x)\cdot\varrho(x)\cdot\dfrac{\partial u}{\partial t}(x,t)\cdot\delta v\cdot\delta t.\label{eq:Qenv-V}
\end{align}

\noindent Note that the calculations with the nilpotent Taylor formula
to get \eqref{eq:QCV-V} correspond to the divergence theorem. From
\eqref{eq:QCV-V}, \eqref{eq:QextV} and \eqref{eq:Qenv-V} we therefore
get that the equality $Q_{\text{env},V}=_{j}Q_{CV,V}+Q_{\text{ext},V}$
is equivalent to
\begin{equation}
c(x)\cdot\varrho(x)\cdot\dfrac{\partial u}{\partial t}(x,t)\delta t\delta v=_{j}\left[\text{div}[k\cdot\text{grad}(u)](x,t)+F(x,t)\right]\delta t\delta v.\label{eq:heatBeforeCanc}
\end{equation}
More precisely: \eqref{eq:heatFluxes} implies \eqref{eq:heatBeforeCanc},
and the latter implies the former but with $=_{j}$ replacing $=$.
The following theorem allows us to cancel the nilpotent factor $\delta t\delta v$
in \eqref{eq:heatBeforeCanc}:
\begin{thm}
\noindent \label{thm:cancLaw}Let $x$, $r$, $s\,\in\rcrho$, $\lvert x\rvert\geq\diff\rho^{q}$,
$j\in\mathbb{R}_{>0}.$ Assume that $x\cdot r=_{j}x\cdot s$ and $\frac{1}{j}-q=:\frac{1}{k}>0$.
Then $r=_{k}s.$ Vice versa, if $r=_{k}s$, and $x$ is finite, then
$x\cdot r=_{k}x\cdot s$.
\end{thm}

\begin{proof}
Assume that $x\cdot r=_{j}x\cdot s$. Then $|xr-xs|\le C\diff\rho^{\frac{1}{j}}$,
with $C\in\R_{\ge0}$. Then , $\lvert r-s\rvert=\lvert x\rvert\cdot\lvert r-s\rvert\cdot\dfrac{1}{\lvert x\rvert}\leq\dfrac{C\cdot\diff\rho^{\frac{1}{j}}}{\diff\rho^{q}}=C\cdot\diff\rho^{\frac{1}{j}-q}=C\diff\rho^{\frac{1}{k}}$
since $\frac{1}{k}=\frac{1}{j}-q$. For the second part of the conclusion,
$x$ finite means $|x|\le K\in\R_{>0}$, so that $|r-s|\le C\diff\rho^{\frac{1}{k}}$
implies $|xr-xs|\le KC\diff\rho^{\frac{1}{k}}$.
\end{proof}
\noindent This derivation is summed up in the following Lemma which
we just have proven.
\begin{lem}
\noindent \label{lem:heat}Let $B\subseteq\rti^{3}$, $B=\overline{\accentset{\circ}{B}}$,
and consider the GSF $\varrho$, $c$, $k:B\to{}^{\rho}\tilde{\mathbb{R}}$,
$u$, $F:B\times[0,\infty)\to{}^{\rho}\tilde{\mathbb{R}}$. Take a
point $(x,t)\in\accentset{\circ}{B}\times[0,\infty)$ and define $V$,
$Q_{CV,V}$, $Q_{\text{\emph{ext}},V}$ and $Q_{\text{\emph{env}},V}$
as in \eqref{eq:QCVV}, \eqref{eq:QextV}, \eqref{eq:QenvV}, where
the infinitesimals $\delta t$, $\delta x_{i}\in\rti_{>0}$ satisfy
\begin{align}
\delta v\cdot\delta t & \in D_{1e},\ \delta t\cdot\delta s_{i}\cdot(\delta x_{i})^{2}=_{j}0,\ (\delta t)^{2}\delta v=_{j}0\label{eq:condInf}\\
\delta v\cdot\delta t & \ge\diff\rho^{q},\ \frac{1}{k}=\frac{1}{j}-q,\nonumber 
\end{align}
and where the first order nilpotent Taylor formula for $k$, $\frac{\partial u}{\partial x_{i}}(-,t)$
and $u(x,-)$ holds in $D_{1e}$. Then the following are equivalent:
\begin{enumerate}
\item \label{enu:fluxes_j}$Q_{\text{\emph{env}},V}=_{j}Q_{CV,V}+Q_{\text{\emph{ext}},V}$,
\item \label{enu:heat_k}$c(x)\cdot\varrho(x)\cdot\dfrac{\partial u}{\partial t}(x,t)=_{k}\text{\emph{div}}[k\cdot\text{\emph{grad}}(u)](x,t)+F(x,t)$.
\end{enumerate}
\end{lem}

\noindent Note that this result corresponds to the usual informal
derivation, but it is now stated as a formal theorem where the use
of nilpotent infinitesimals and the corresponding Taylor formula is
now precise.

The next natural steps thus concern the existence of infinitesimals
satisfying \eqref{eq:condInf} and how to obtain a true equality $=$
in the final heat equation for GSF. Conditions \eqref{eq:condInf}
hold if e.g.~we choose $\delta t=\diff\rho^{\frac{1}{3e}}$ and $\delta x_{i}=\diff\rho^{\frac{1}{5e}}$
(recall that $e\le j$ and note that these infinitesimals depend on
$j$); thereby, it easily follows that we can take $q=q(j)=\frac{14}{15j}$
and hence $k=k(j)=15j$.

\noindent Finally, assume that $Q_{\text{env,}V}(x,t,\delta t,\delta\bar{x})=Q_{CV,V}(x,t,\delta t,\delta\bar{x})+Q_{\text{ext,}V}(x,t,\delta t,\delta\bar{x})$
holds at $(x,t)$ and for all infinitesimals $\delta t$, $\delta\bar{x}$.
Thereby (using simplified notations)
\begin{equation}
Q_{\text{env,}V}=_{j}Q_{CV,V}+Q_{\text{ext,}V}\quad\forall j\in\mathbb{R}_{>0}.\label{eq:Qj}
\end{equation}

\noindent Lemma \ref{lem:heat} yields the heat equation with equality
up to order $k(j)=15j$. If we now let $j\to0^{+}$, then also $k(j)\to0^{+}$
and hence Thm.~\ref{thm:TaylorPeano-Infinitesimals}.\ref{enu:relationDandD_j}
proves the heat equation with $=$.

Even if it is true that the full equality $=$ implies $=_{k(j)}$
in the heat equation, the opposite implication (i.e.~that \ref{enu:heat_k}
above but with $=$ instead of $=_{k}$, implies \ref{enu:fluxes_j}
above with $=$ instead of $=_{j}$) cannot be proved simply by reversing
the previous steps because we would arrive at \eqref{eq:Qj} with
infinitesimals $\delta t=\delta t(j)$, $\delta\bar{x}=\delta\bar{x}(j)$
satisfying \eqref{eq:condInf} that would depend on $j$: taking $j\to0^{+}$
in \eqref{eq:Qj} would not get anything meaningful because $\delta t(j)$,
$\delta\bar{x}(j)\to0$.

\noindent The final result is then stated as follows:
\begin{thm}
\noindent \label{thm:heat}Let $B\subseteq\rti^{3}$, $B=\overline{\accentset{\circ}{B}}$,
and consider the GSF $\varrho$, $c$, $k:B\to{}^{\rho}\tilde{\mathbb{R}}$,
$u$, $F:B\times[0,\infty)\to{}^{\rho}\tilde{\mathbb{R}}$. Take a
point $(x,t)\in\accentset{\circ}{B}\times[0,\infty)$ and define $V$,
$Q_{CV,V}$, $Q_{\text{\emph{ext}},V}$ and $Q_{\text{\emph{env}},V}$
as in \eqref{eq:QCVV}, \eqref{eq:QextV}, \eqref{eq:QenvV}. Finally
assume that $Q_{\text{env,}V}(x,t,\delta t,\delta\bar{x})=Q_{CV,V}(x,t,\delta t,\delta\bar{x})+Q_{\text{ext,}V}(x,t,\delta t,\delta\bar{x})$
holds at $(x,t)$ and for all infinitesimals $\delta t$, $\delta\bar{x}$.
Then
\begin{equation}
c(x)\cdot\varrho(x)\cdot\dfrac{\partial u}{\partial t}(x,t)=\text{\emph{div}}[k\cdot\text{\emph{grad}}(u)](x,t)+F(x,t).\label{eq:heatEq}
\end{equation}
Moreover, if these conditions hold at all points $x\in\accentset{\circ}{B}$,
then equation \eqref{eq:heatEq} holds on the entire $B$ because
$B=\overline{\accentset{\circ}{B}}$.
\end{thm}

\subsection{Derivation of the wave equation for GSF}

\noindent In this section, we derive the wave equation in a similar
way to \cite{Vla71,Gio10}, with the difference that we extend its
applicability to GSF and not only to smooth functions. Consider a
string with given equilibrium position located on an interval $[a,b]\subseteq\rti$
for $a$, $b\in\rcrho$, $a<b$. Let this string now make small transversal
oscillations around its equilibrium position. The position $s_{t}\subseteq\rcrho^{2}$
of the string is always represented by the graph of a curve $\gamma:[a,b]\times[0,\infty)\to\rcrho^{2}$.
Furthermore, we set $\gamma_{xt}:=\gamma(x,t)$, $s_{t}:=\{\gamma_{xt}\in\rcrho^{2}\mid a\leq x\leq b\}$
for all $t\in[0,\infty)$. The curve $\gamma$ is supposed to be injective
with respect to $x\in(a,b)$, i.e.~$\gamma_{x_{1}t}\ne\gamma_{x_{2}t}$
for all $t\in[0,+\infty)$ and all $x_{1}$, $x_{2}\in(a,b)$ such
that $x_{1}\ne x_{2}$; therefore, the order relation on $(a,b)$
implies an order relation on the support $s_{t}$. For all pairs of
points $p=\gamma_{x_{p}t}$, $q=\gamma_{x_{q}t}\in s_{t}$ on the
string at time $t$, we can define the subbodies $[p:=\{\gamma_{xt}\mid x_{p}\leq x\leq b\}$,
$p]:=\{\gamma_{xt}\mid a\leq x\leq x_{p}\}$ and $[p|q]:=\{\gamma_{xt}\mid x_{p}\leq x\leq x_{q}\}$
corresponding to the parts of the string after the point $p\in s_{t}$,
before the same point and between the points $p$, $q\in s_{t}$.
Clearly, every subbody of the form $p]$ exerts a force on every other
subbody it is in contact with, i.e.~$[p|q]$ and $p]$. Moreover,
the force $\vec{F}(A,B)\in\rcrho$ exerted by the subbody $A$ on
the subbody $B$ satisfies the following equalities:

\noindent 
\begin{align}
\vec{F}([p|q],p]) & =\vec{F}([p,p])\label{eq:note1}\\
\vec{F}([q,[p|q]) & =\vec{F}([q,q])\label{eq:note2}\\
\vec{F}(p],[p|q]) & =-\vec{F}([p|q],p]),\label{eq:note3}
\end{align}

\noindent for all pairs of points $p$, $q\in s_{t}$ and time $t\in[0,\infty)$.
The third equation \eqref{eq:note3} corresponds to the action-reaction
principle.

\noindent We can now define the tension $\vec{T}$ at the point $\gamma_{xt}\in s_{t}$
and time $t\in[0,\infty)$ as

\noindent 
\begin{equation}
\vec{T}(x,t):=\vec{F}([\gamma_{xt},\gamma_{xt}]).\label{eq:5.21}
\end{equation}

\noindent Consider now the infinitesimal subbody $[x|x+\delta x]:=[\gamma_{xt}|\gamma_{x+\delta x,t}]\subseteq s_{t}$
located at time $t$ between the points $\gamma_{xt}\in s_{t}$ and
$\gamma_{x+\delta x,t}\in s_{t},$ and defined by the first order
infinitesimal $\delta x\in D_{1j}$, $\delta x>0$. We have an action
on this infinitesimal subbody due to mass forces of linear density
$\vec{G}:[a,b]\times[0,\infty)\to\rcrho^{2}$ that allows us to represent
Newton's law as follows:

\noindent 
\begin{equation}
\rho\cdot\delta x\cdot\dfrac{\partial^{2}\gamma}{\partial t^{2}}=\vec{F}(\gamma_{xt}],[x|x+\delta x])+\vec{F}([\gamma_{x+\delta x,t},[x|x+\delta x])+\vec{G}\cdot\varrho\cdot\delta x,\label{eq:58}
\end{equation}

\noindent where $\varrho:[a,b]\times[0,\infty)\to\rcrho^{2}$ is the
linear mass density, and all functions, unless stated otherwise, are
evaluated at $(x,t)\in(a,b)\times[0,\infty)$.

\noindent The contact forces appearing in Newton's law are caused
by the interaction of the infinitesimal subbody with other contacting
subbodies along the border $\partial[x|x+\delta x]=\{\gamma_{xt},\gamma_{x+\delta x,t}\}\subseteq\rcrho^{2}$.
Using now relations \eqref{eq:note2} and \eqref{eq:note3} with $q=\gamma_{x+\delta x,t}$
and $p=\gamma_{xt}$, so that $[p|q]=[x|x+\delta x]$, we see by \eqref{eq:58}
that 
\begin{equation}
\rho\cdot\delta x\cdot\dfrac{\partial^{2\gamma}}{\partial t^{2}}=-\vec{F}([x|x+\delta x],\gamma_{xt}])+\vec{F}([\gamma_{x+\delta x,t},\gamma_{x+\delta x,t}])+\vec{G}\cdot\rho\cdot\delta x.\label{5.22}
\end{equation}

\noindent By \eqref{eq:note1}, the definition of tension \eqref{eq:5.21}
and inserting it in \eqref{5.22}, we obtain
\begin{align}
\rho\cdot\delta x\cdot\dfrac{\partial^{2}\gamma}{\partial t^{2}} & =-\vec{F}([\gamma_{xt},\gamma_{xt}])+\vec{F}([\gamma_{x+\delta x,t},\gamma_{x+\delta x,t}])+\overrightarrow{G}\cdot\rho\cdot\delta x\nonumber \\
 & =-\vec{T}(x,t)+\vec{T}(x+\delta x,t)+\vec{G}\cdot\rho\cdot\delta x.\label{eq:59}
\end{align}

\noindent Note that, up to now, we have not used neither the small
oscillation nor the transversal oscillation hypothesis of the force
$\vec{G}$. As for the latter, it can be introduced with the assumption
\begin{equation}
\vec{G}(x,t)\cdot\vec{e}_{1}=0\quad\forall x,t,\label{5.23}
\end{equation}

\noindent where $(\vec{e}_{1},\vec{e}_{2})$ are the axial unit vectors.
Let now $\varphi(x,t)$ denote the non-oriented angle between the
tangent unit vector $\vec{t}(x,t):=\frac{\partial\gamma}{\partial x}(x,t)/\left|\frac{\partial\gamma}{\partial x}(x,t)\right|$
(a subsequent assumption will guarantee that $\vec{t}$ always exists)
at the point $\gamma_{x,t}$ and the $x$-axis, i.e.~the unique $\phi(x,t)\in[0,\pi]\subseteq\rti$
defined by
\begin{equation}
\vec{t}(x,t)=\cos(\phi(x,t))\vec{e}_{1}+\sin(\phi(x,t))\vec{e}_{2}.\label{eq:tangVers}
\end{equation}
Setting $(\gamma_{1},\gamma_{2})=\gamma$ for the two components of
the curve $\gamma$, from this equality directly follows
\begin{equation}
\frac{\partial\gamma_{1}}{\partial x}\sin\phi=\frac{\partial\gamma_{2}}{\partial x}\cos\phi\label{eq:sincos}
\end{equation}
The small oscillation hypothesis can then be formalized with the assumption
that this angle $\phi(x,t)$ is a first order infinitesimal (in the
following Thm.~\ref{thm:waveeq}, we will assume a weaker assumption),
i.e. 
\begin{equation}
\varphi(x,t)\in D_{1j}\quad\forall x,t.\label{5.24}
\end{equation}

\noindent This allows us to recreate the classical derivation in the
most faithful possible way. Furthermore, in the standard proof of
the wave equation, only curves of the specific form $\gamma_{xt}=(x,u(x,t))$
are considered (this implies that the tangent unit vector $\vec{t}(x,t)$
always exists). The Tayor-formula for nilpotent infinitesimals Thm.~\ref{thm:TaylorPeano-Infinitesimals}.\ref{enu:infTaylor}
yields $\sin(\phi)=_{j}\phi\in D_{1j}$ and $\cos(\phi)=_{j}1$ (note
that assumption \eqref{eq:relBetween-i-AndInfinite-f^(n+1)} holds
for any $j$ and $k$ for both $\sin(x)$ and $\cos(x)$), and hence
$\dfrac{\partial\gamma_{2}}{\partial x}=_{j}\varphi$ from \eqref{eq:sincos}.
Therefore, $\left(\dfrac{\partial\gamma_{2}}{\partial x}\right)^{2}=_{j}0$
and the total length of the string becomes 
\begin{equation}
L=\int_{a}^{b}\sqrt{1+\left[\dfrac{\partial\gamma_{2}}{\partial x}(x,t)\right]^{2}}\,\diff x=_{j}b-a\quad\forall t\in[0,\infty).\label{5.25}
\end{equation}

\noindent Following Hooke's law, this allows us to assume that the
tension is of constant modulus $T=|\vec{T}(x,t)|$ that is neither
depending on the position $x$ nor on the time $t$, i.e. 
\begin{equation}
\vec{T}(x,t)=T\cdot\vec{t}(x,t)\quad\forall x\in(a,b)\,\forall t\in[0,\infty).\label{5.26}
\end{equation}

\noindent Note that, as a second part of the hypothesis about nontransversal
oscillations of the string, we have that the tension $\vec{T}$ is
parallel to the tangent vector. We then project equation \ref{eq:59}
to the $y$-axis and obtain 
\begin{align}
\rho\cdot\delta x\cdot\dfrac{\partial^{2}u}{\partial t^{2}} & =-T\cdot\vec{t}(x,t)\cdot\vec{e}_{2}+T\cdot\vec{t}(x+\delta x,t)\cdot\vec{e}_{2}+\vec{G}\cdot\vec{e}_{2}\cdot\rho\cdot\delta x\nonumber \\
 & =-T\cdot\sin(\varphi(x,t))+T\cdot\sin(\varphi(x+\delta x),t)+G_{2}\cdot\rho\cdot\delta x\nonumber \\
 & =T\cdot\left[\frac{\partial u}{\partial x}(x+\delta x,t)\cos(\phi(x+\delta x,t))-\frac{\partial u}{\partial x}(x,t)\cos(\phi(x,t))\right]\label{eq:waveLastBut1}\\
 & \phantom{=}+G_{2}\cdot\rho\cdot\delta x,
\end{align}
where $G_{2}=\vec{G}\cdot\vec{e}_{2}$ is the second component of
$\vec{G}$. Now, assume that the first order Taylor formula for $\frac{\partial u}{\partial x}(-,t)$
holds on $D_{1e}$, with $e\le j$, and take $\delta x\in D_{1e}$,
$\delta x\ge\diff\rho^{q}$ (e.g.~$\delta x=\diff\rho^{\frac{1}{2e}+\frac{1}{2}}$).
Then, $\cos(\phi(x+\delta x,t))=_{j}1=_{j}\cos(\phi(x,t))$ and $\frac{\partial u}{\partial x}(x+\delta x,t)-\frac{\partial u}{\partial x}(x,t)=_{j}\frac{\partial^{2}u}{\partial x^{2}}(x,t)\delta x$,
and from \eqref{eq:waveLastBut1} we hence get
\[
\rho\cdot\delta x\cdot\dfrac{\partial^{2}u}{\partial t^{2}}=_{j}T\frac{\partial^{2}u}{\partial x^{2}}(x,t)\delta x+G_{2}\cdot\rho\cdot\delta x.
\]
We can now use the cancellation law Thm.~\ref{thm:cancLaw} to cancel
out the $\delta x$ obtaining
\begin{equation}
\rho\dfrac{\partial^{2}u}{\partial t^{2}}=_{k}T\frac{\partial^{2}u}{\partial x^{2}}(x,t)+G_{2}\rho,\label{eq:waveApprox}
\end{equation}
for $\frac{1}{k}=\frac{1}{j}-q$.

Can we take $j\to0^{+}$ (and hence $k\to0^{+}$) in \eqref{eq:waveApprox}?
Actually no, because all this deduction depends on the small oscillations
assumption \eqref{5.24}, and the only $\phi\in D_{1j}$ for all $j$
is $\phi=0$, i.e.~the string is not oscillating at all. In order
to underscore that this classical deduction of the wave equation leads
to an approximate equality only, we generalize the previous proof
in the following
\begin{thm}
\label{thm:waveeq}Let $a$, $b\in\mathbb{\rcrho}$, with $a<b$,
$\gamma:[a,b]\times[0,\infty)\to\rcrho^{2}$, $\rho:[a,b]\times[0,\infty)\to\rcrho$,
$\vec{G}$, $\vec{T}:[a,b]\times[0,\infty)\to\rcrho^{2}$ be GSF,
and let $T\in\rcrho$ be an invertible generalized number such that
both $T$ and $\frac{1}{T}$ are finite. Suppose that $\gamma(x,t)=(x,u(x,t))$
for all $x$, $t$, and let $\vec{t}(x,t)$ be the unit tangent vector
to $\gamma$. Assume that at least an approximate version of Hooke's
law and the second Newton's law
\begin{align}
\vec{T}(x,t) & =_{j}T\cdot\vec{t}(x,t),\tag{{Hooke}}\label{eq:71}\\
\rho\cdot\delta x\cdot\dfrac{\partial^{2}\gamma}{\partial t^{2}}(x,t) & =\vec{T}(x+\delta x,t)-\vec{T}(x,t)+\vec{G}\cdot\rho\cdot\delta x,\tag{{II\ Newton}}\label{eq:72}
\end{align}

\noindent hold for every point $(x,t)\in(a,b)\times[0,\infty)$ and
for an infinitesimal $\delta x=\diff\rho^{q}$ such that $\delta x\in D_{1e}$,
where the first order Taylor formula for $\frac{\partial u}{\partial x}(-,t)$
holds on $D_{1e}$ and $e\le j$. Finally, let $\phi(x,t)$ be the
non-ordered angle between $\vec{t}(x,t)$ and the $x$-axis, and suppose
that $\dfrac{\partial\varphi}{\partial x}(x,t)\ge\diff\rho^{p}$,
$\phi(x,t)<\frac{\pi}{2}$. Then we have:
\begin{enumerate}
\item \label{enu:wave_j}If $\rho(x,t)\cdot\dfrac{\partial^{2}u}{\partial t^{2}}(x,t)=_{j}T\cdot\dfrac{\partial^{2}u}{\partial x^{2}}(x,t)+G_{2}(x,t)\cdot\rho(x,t)$,
then $\cos^{3}(\varphi(x,t))=_{h}1$, where $\frac{1}{h}=\frac{1}{j}-p-2q$.
\item \label{enu:small-phi}If $\cos^{3}(\varphi(x,t))=_{j}1$, then $\rho(x,t)\cdot\dfrac{\partial^{2}u}{\partial t^{2}}(x,t)=_{k}T\cdot\dfrac{\partial^{2}u}{\partial x^{2}}(x,t)+G_{2}(x,t)\cdot\rho(x,t)$,
where $\frac{1}{k}=\frac{1}{j}-q$.
\end{enumerate}
\end{thm}

\noindent For example, the assumption of \ref{enu:small-phi} holds
if $\phi(x,t)\in D_{k\hat{\jmath}}$ and $\frac{(k+1)}{2}\hat{\jmath}=j$.
Finally, if $\phi(x,t)\in D_{3j}$ for all $x$, $t$, and $b-a$
is finite, then $\text{length}(\gamma(-,t))=_{2j}b-a$.
\begin{proof}
As usual, if the arguments of a function are missing, we mean they
are evaluated at $(x,t)$.

\noindent \ref{enu:wave_j}: Projecting \eqref{eq:72} on $\vec{e}_{2}$
and using \eqref{eq:71} and \eqref{eq:tangVers} we get
\[
\rho\delta x\frac{\partial^{2}u}{\partial t^{2}}=T\sin(\phi(x+\delta x,t))-T\sin(\phi(x,t))+G_{2}\rho\delta x.
\]
Therefore, the assumption of \ref{enu:wave_j} implies
\[
T\sin(\phi(x+\delta x,t))-T\sin(\phi(x,t))+G_{2}\rho\delta x=_{j}T\dfrac{\partial^{2}u}{\partial x^{2}}\delta x+G_{2}\rho\delta x.
\]
Since $\delta x\in D_{1e}$ and $e\le j$, we can use the first order
Taylor formula with $\frac{\partial u}{\partial x}(-,t)$ to get
\begin{multline*}
T\sin(\phi(x+\delta x,t))-T\sin(\phi(x,t))+G_{2}\rho\delta x=_{j}\\
T\delta x\left[\frac{\partial u}{\partial x}(x+\delta x,t)-\frac{\partial u}{\partial x}(x,t)\right]+G_{2}\rho\delta x.
\end{multline*}
Multiply by $\frac{1}{T}$ (which is finite, see Thm.~\ref{thm:cancLaw})
and use \eqref{eq:sincos} considering that $\phi(x,t)<\frac{\pi}{2}$
to obtain
\[
\left[\sin(\phi(x+\delta x,t))-\sin(\phi(x,t))\right]\delta x=_{j}\left[\tan(\phi(x+\delta x,t))-\tan(\phi(x,t))\right]\delta x.
\]
Using the cancellation law Thm.~\ref{thm:cancLaw} with $\frac{1}{k}:=\frac{1}{j}-q$,
this yields
\[
\sin(\phi(x+\delta x,t))-\sin(\phi(x,t))=_{k}\tan(\phi(x+\delta x,t))-\tan(\phi(x,t)).
\]
We can use the first order Taylor formula Thm.~\ref{thm:TaylorPeano-Infinitesimals}.\ref{enu:infTaylor}
both with $\sin(\phi(-,t))$ and $\tan(\phi(-,t))$ because $e\le j$
and hence $\delta x\in D_{1e}\subseteq D_{1j}\subseteq D_{1k}$ (note
that the derivatives of these functions are always finite because
$\phi(x,t)<\frac{\pi}{2}$)
\[
\delta x\cdot\cos(\phi)\cdot\frac{\partial\phi}{\partial x}=_{k}\delta x\frac{1}{\cos^{2}(\phi)}\cdot\frac{\partial\phi}{\partial x}.
\]
Simplifying by $\delta x\cdot\frac{\partial\phi}{\partial x}\ge\diff\rho^{q+p}$,
we obtain $\cos(\phi)=_{h}\frac{1}{\cos^{2}(\phi)}$, where $\frac{1}{h}:=\frac{1}{k}-p-q=\frac{1}{j}-p-2q$.
Since $\cos^{2}(\phi)$ is finite, using Thm.~\ref{thm:cancLaw}
we obtain the conclusion.

\noindent \ref{enu:small-phi}: It suffices to invert all the previous
steps starting from $\cos^{3}(\phi)=_{j}1$ and considering that we
always have to multiply by finite numbers. Only in the last step we
need to simplify by $\delta x$ and hence we switch from $=_{j}$
to $=_{k}$.

From Taylor formula with Peano remainder Thm.~\ref{thm:TaylorPeano-Infinitesimals}.\ref{enu:Peano-1}
we have $\cos^{3}(\phi)=\left(1-\frac{\phi^{2}}{2}+o(\phi^{3})\right)^{3}=1-\frac{3}{2}\phi^{2}+o(\phi^{3})$.
If $\phi\in D_{k\hat{\jmath}}$, then $|\phi^{k+1}|\le C\diff\rho^{\frac{1}{\hat{\jmath}}}$
and hence $\phi^{2}\le C\diff\rho^{\frac{2}{(k+1)\hat{\jmath}}}=C\diff\rho^{\frac{1}{j}}$
and $\left|\cos^{3}(\phi)-1\right|=\left|\frac{3}{2}\phi^{2}+o(\phi^{3})\right|\le\bar{C}\diff\rho^{\frac{1}{j}}$.
Note that this property includes the classical case $\phi\in D_{1j}$,
but also e.g.~$\phi\in D_{2j-1,1}$.

Finally, assume that $\phi(x,t)\in D_{3j}$ for all $x$, $t$. From
Taylor formula $\sin(\phi)=_{j}\phi-\frac{\phi^{3}}{6}$ and $\cos(\phi)=_{j}1-\frac{\phi^{2}}{2}$.
Therefore, \eqref{eq:sincos} yields $\phi-\frac{\phi^{3}}{6}=_{j}\frac{\partial u}{\partial x}\left(1-\frac{\phi^{2}}{2}\right)$.
Taking the square and considering that $\phi^{4}=_{j}0$, this implies
$\phi^{2}=_{j}\left(\frac{\partial u}{\partial x}\right)^{2}\left(1-\phi^{2}\right)$.
Multiplying both sides by $1+\phi^{2}$ and using again that $\phi^{4}=_{j}0$
we obtain $\left(\frac{\partial u}{\partial x}\right)^{2}(x,t)=_{j}\phi^{2}(x,t)$
for all $x$, $t$. The mean value theorem Thm.~\ref{thm:classicalThms}.\ref{enu:integralMeanValue}
and Thm.~\ref{thm:TaylorPeano-Infinitesimals}.\ref{enu:func=00003D_j}
yield $\text{length}(\gamma(-,t))=\sqrt{1+\left(\frac{\partial u}{\partial x}\right)^{2}(c,t)}\cdot(b-a)=_{j}\sqrt{1+\phi^{2}(c,t)}\cdot(b-a)$
for some $c\in[a,b]$. The Taylor formula with Peano remainder applied
to the function $\sqrt{1+x}$ gives $\text{length}(\gamma(-,t))=_{j}b-a+\frac{b-a}{2}\phi^{2}(c,t)+o(\phi)$,
which implies the conclusion because $\left|\frac{b-a}{2}\right|\phi^{2}(c,t)\le C\diff\rho^{\frac{1}{2j}}$.
\end{proof}
This theorem suggests the following comments and potential applications:
\begin{enumerate}
\item It highlights that the wave equation is intrinsically approximated
because it implies $\cos^{3}(\phi)=_{h}1$, which is necessarily only
an approximated relation.
\item It is formulated as a general mathematical theorem depending on two
assumptions corresponding to physical laws.
\item In our deduction, we do not conclude by ``magically'' transforming
approximate equalities $\simeq$ into true equalities $=$ or neglecting
little-oh terms despite keeping true equalities.
\item The validity of the wave equation for GSF can find possible applications
in geophysics. In seismology, we have for example elastodynamical
oscillations after earthquakes or simply the elastodynamical properties
of materials that have a rapid change in density like the seabed or
earth's crust. This leads to the seismological base equations of elastodynamics
with a special case being the isotropic wave equation where the setting
of GSF could be used to treat the special case with non-smooth coefficients.
A motivation for this topic can be found in \cite{BrHo,BuKe}.
\item Other potential applications can also be considered in global seismology,
where one is dealing with seismic wave propagation. In fact, hyberbolic
PDE in global seismology do have generalized functions as coefficients,
together with a singular structure created by geological and physical
processes. These processes are supposed to behave in a fractal way.
In the so-called \emph{seismic transmission problem}, we want to diagonalize
a first order system of PDE and then transform it to the second order
wave equation. This requires us to differentiate the coefficients,
which means that even though the original model medium varies continuously,
coefficients that are (highly) discontinuous will naturally appear
in this procedure. A possible way to deal with this is to embed the
fractal coefficients into GSF or in a Colombeau algebra. See e.g.~\cite{Hoer2}. 
\item We can finally think at using GSF in mathematical general relativity,
where one considers wave equations on Lorentzian manifolds with non-smooth
metric, i.e.~non-smooth coefficients in the corresponding wave equation,
see for example \cite{Hoer3}. Colombeau generalized functions is
already a tool used to prove local well-posedness of the wave equation
in space times that are of conical type. Cosmic strings are e.g.~objects
that can be treated within this theory. There has even been a generalization
of this result to a class of locally bounded space-times with discussion
of a static case and an extension to non-scalar equations. Similar
applications can hence be considered using GSF, because of their better
properties with respect to Colombeau theory.
\end{enumerate}

\section{\label{sec:Examples-of-applications}Examples of applications}

Nature is made up of different bodies, having boundaries and frequently
interacting in a non-smooth way. Even the simple motion of an elastic
bouncing ball seems to be more easily modeled using non-differentiable
functions than classical $\mathcal{C}^{2}$ ones, at least if we are
not interested to model the non-trivial behavior at the collision
times. Therefore, the motivation to introduce a suitable kind of generalized
functions formalism in a mathematical model is clear, and this would
undoubtedly be of an applicable advantage, since many relevant systems
are described by singular mathematical objects: non-smooth constraints,
collisions between two or more bodies, motion in different or in granular
media, discontinuous propagation of rays of light, even turning on
the switch of an electrical circuit, to name but a few, and only in
the framework of classical physics. In this section we show several
applications of the theory of GSF we reviewed above.

We will \emph{not} consider mathematical models of singular dynamical
systems \emph{at} the times when singularities occur. Indeed, this
would clearly require new physical ideas, e.g.~in order to consider
the nonlinear behavior of objects or materials for the entire duration
of the singularity. Like in every mathematical model, the correct
point of view concerns J.~von Neumann's \emph{reasonably wide area
}of applicability of a mathematical model, i.e.~the range of phenomena
where our model is expected to work (see \cite[pag.~492]{vonN}).
Therefore, it is \emph{not} epistemologically correct to use the theory
described in the present article to deduce a physical property of
our modeled systems when a singularity occurs. Stating it with a language
typically used in physics, \emph{we consider physical systems where
the duration of the singularity is negligible with respect to the
durations of the other phenomena that take place in the system}. Mathematically,
this means to consider as infinitesimal the duration of the singularities.
As a consequence, several quantities changing during this infinitesimal
interval of time have infinite derivatives. We can hence paraphrase
the latter sentence saying that the amplitude (of the derivatives)
of these physical quantities is much larger than all the other (finite)
quantities we can estimate in the system. However, this is a logical
consequence of our lacking of interest to include in our mathematical
model what happens during the singularity, constructing at the same
time a beautiful and sufficiently powerful mathematical model, and
not because these quantities really become infinite. Thereby, it is
\emph{not} epistemologically correct to state that, e.g., if a speed
is infinite at some singularity, this means that we must use relativity
theory: on the contrary, relativity theory is exactly a modeling setting
where infinite speeds are impossible!

On the other hand, the aforementioned ``wide area'' is now able
to include in a single equation the dynamical properties of our modeled
systems, without being forced to subdivide into cases of the type
``before/after the occurrence of each singularity''. Which can be
considered as not reasonable in several cases, e.g.~in the motion
of a particle in a granular medium or of a ray of light in an optical
fiber.

Finally, note that remaining far from the singularity (from the point
of view of the physical interpretation), is what allow us to state
that in several cases this kind of models are already experimentally
validated.\medskip{}

Moreover, the applications we are going to present always end up with
an ODE. Existence and uniqueness of the solution is therefore guaranteed
by Thm\@.~\ref{thm:PicLindFiniteContr}. Clearly, if an explicit
analytic solution is possible, this is preferable, but this is a rare
event, and frequently we have to opt for a numerical solution, usually
simply solving the corresponding $\eps$-wise ODE, for several values
of sufficiently small $\eps$. This mean that we are considering numerical
solutions of our differential equations as empirical laboratories
helping us to guess suitable properties and hence \emph{conjecture}
on the solutions. In principle, these properties \emph{must} be justified
by corresponding theorems. From this point of view, the fact that
GSF share with ordinary smooth functions a lot of classical theorems
(such as the intermediate value, the extreme value, the mean value,
Taylor theorems, etc.) is usually of great help. For example, pictures
of Heaviside's function and Dirac's delta in Fig.~\ref{fig:MollifierHeaviside}
are clearly obtained in the same way by numerical methods, but their
properties can be fully justified by suitable theorems, see e.g.~Rem.~\ref{rem:embedding_properties}.\ref{enu:embDelta}
and \ref{enu:embH} or Example \ref{enu:deltaCompDelta}.

Finally, we already saw in Sec.~\ref{sec:Embeddings} that if $\mu$
is a $1$-dimensional Colombeau mollifier, and $\delta$ is the $\iota^{b}$-embedding
of the Dirac delta, then $\delta(x)=b\mu(bx)$ for all $x\in\rti$.
Thereby, the Heaviside function is $H(x)=\int_{c}^{x}\delta(t)\,\diff t=\int_{bc}^{bx}\mu(t)\,\diff t$,
for all $x\in\rti$ and all $c\in\rti$ sufficiently far from $0$,
i.e.~such that $c<r<0$ for some $r\in\R_{<0}$. If the oscillations
in an infinitesimal neighborhood of $0$ shown in Fig.~\ref{fig:MollifierHeaviside}
have no modelling meaning, one can easily implement e.g.~a non-decreasing
Heaviside-like function by smoothly interpolating the constant functions
$y=0$ and $y=1$ in the intervals $(-\infty,a_{\eps}]_{\R}$ and
$[b_{\eps},+\infty)_{\R}$, where $a=[a_{\eps}]<0<b=[b_{\eps}]$ are
chosen depending on the model requirements.

\subsection{\label{subsec:Singular-variable-length}Singular variable length
pendulum}

As a first example, we want to study the dynamics of a pendulum with
singularly variable length, e.g.~because it is wrapping on a parallelepiped
(see Fig.~\ref{fig:variable_length_pendulum}; see \cite{MaHoAh12}
for a similar but non-singular case, and \cite{MuGl} for a similar
problem of jumps in the Lagrangian, but without the explicit use of
infinitesimals and generalized functions).

The pendulum length function is therefore $\Lambda(\theta)=H(\theta_{0}-\theta)L_{1}+L_{2}$,
where $H$ is the (embedding of the) Heaviside function. We always
assume that $L_{1}$, $L_{2}\in\rti_{>0}$ are finite and non-infinitesimal
numbers. From this it follows that for all $\theta$, $H(\theta_{0}-\theta)>\frac{\diff{\rho}-L_{2}}{L_{1}}\approx-\frac{L_{2}}{L_{1}}$
and hence that also $\Lambda(\theta)>\diff{\rho}$ is invertible.
The kinetic energy is given by: 
\begin{equation}
T(\theta,\dot{\theta})=\frac{1}{2}m\dot{\theta}^{2}\Lambda(\theta)^{2}.\label{e1c}
\end{equation}
The potential energy (the zero level being the suspension point of
the pendulum) is: 
\begin{equation}
U(\theta)=-mg\Lambda(\theta)\cos\theta-mg(1-H(\theta_{0}-\theta))L_{1}\cos\theta_{0}.\label{e1d}
\end{equation}

\noindent Let us define the Lagrangian \textit{$L$} for this problem
as 
\begin{equation}
L(\theta,\dot{\theta}):=T(\theta,\dot{\theta})-U(\theta).\label{e1a}
\end{equation}
\begin{figure}
\centering{}\includegraphics[scale=0.2]{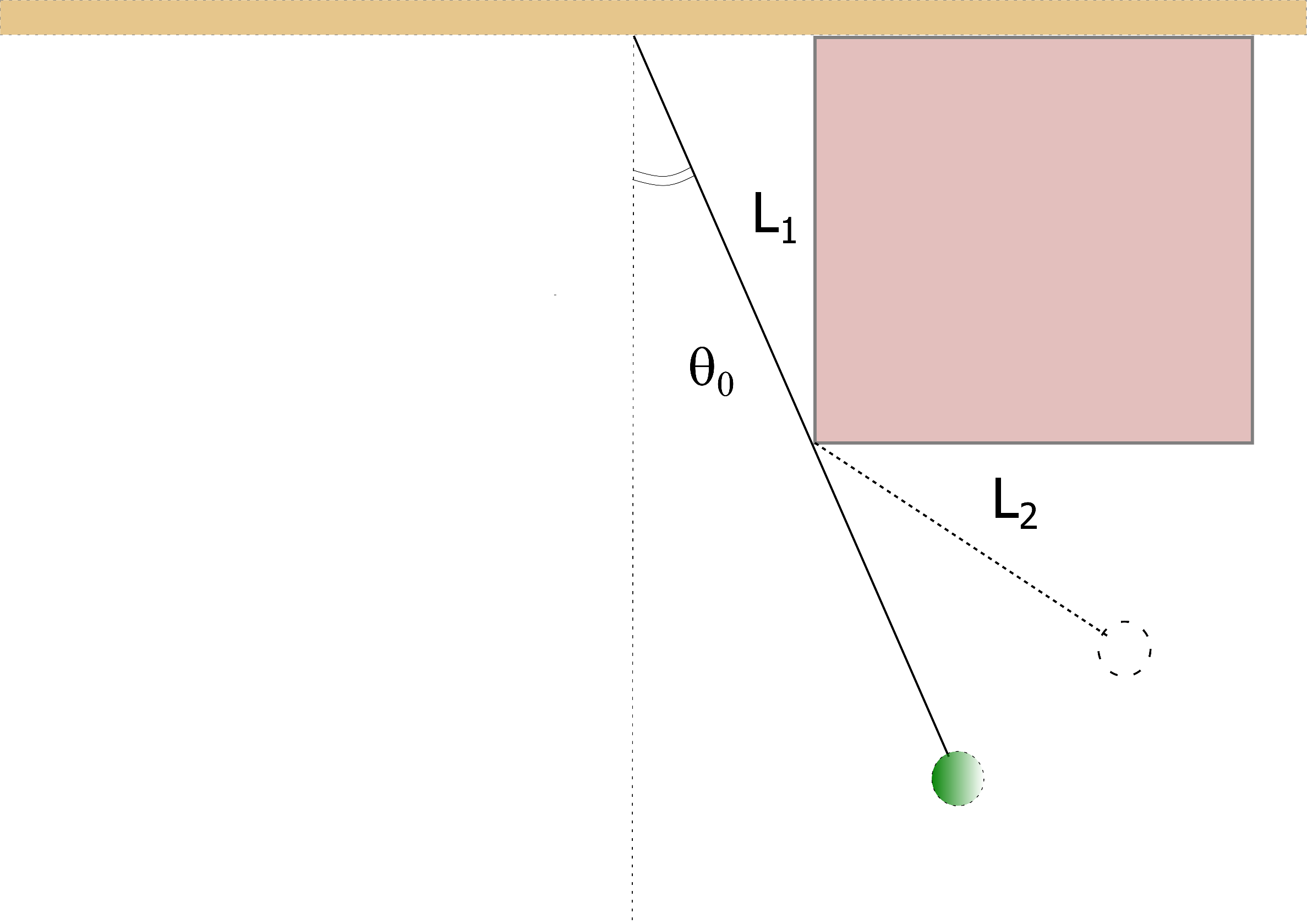}\caption{\label{fig:variable_length_pendulum} Oscillations of a pendulum wrapping
on a parallelepiped}
\end{figure}

\noindent The equation of motion is assumed to satisfy the Euler--Lagrange
equation, see also \cite{FGBL}, and can be written as: 
\begin{equation}
\frac{\partial L}{\partial\theta}=\frac{\diff{}}{\diff{t}}\frac{\partial L}{\partial\dot{\theta}}.\label{e1b}
\end{equation}

\noindent Thereby
\begin{equation}
\frac{\diff{}}{\diff{t}}\frac{\partial L}{\partial\dot{\theta}}=\frac{\diff{}}{\diff{t}}\frac{\partial}{\partial\dot{\theta}}\left(\frac{1}{2}m\dot{\theta}^{2}\Lambda(\theta)^{2}\right)=\frac{\diff{}}{\diff{t}}\left(m\dot{\theta}\Lambda(\theta)^{2}\right)=m\Lambda(\theta)^{2}\ddot{\theta}+2m\dot{\theta}\Lambda(\theta)\dot{\Lambda}(\theta),\label{e1e}
\end{equation}
where $\dot{\Lambda}(\theta):=\frac{\diff{}}{\diff{t}}\Lambda(\theta(t))$.
From \eqref{e1d}, the left side of the Euler--Lagrange equation
\eqref{e1b} reduces to 
\begin{equation}
\frac{\partial L}{\partial\theta}=\frac{\partial T}{\partial\theta}+\frac{\partial(-U)}{\partial\theta}=m\dot{\theta}^{2}\Lambda(\theta)\Lambda'(\theta)+mg\Lambda'(\theta)\left(\cos\theta-\cos\theta_{0}\right)-mg\Lambda(\theta)\sin\theta,\label{e1f}
\end{equation}
where
\begin{equation}
\Lambda'(\theta)=\frac{\diff{}}{\diff{\theta}}\left(H(\theta_{0}-\theta)L_{1}+L_{2}\right)=-\delta(\theta_{0}-\theta)L_{1},\label{e1g}
\end{equation}
and $\delta$ is the Dirac delta function. We then obtain the following
equation of motion: 
\begin{equation}
m\dot{\theta}^{2}\Lambda(\theta)\Lambda'(\theta)+mg\Lambda'(\theta)\left(\cos\theta-\cos\theta_{0}\right)-mg\Lambda(\theta)\sin\theta=m\Lambda(\theta)^{2}\ddot{\theta}+2m\dot{\theta}\Lambda(\theta)\dot{\Lambda}(\theta).\label{e1h}
\end{equation}
Taking into account that $\dot{\Lambda}(\theta)=\Lambda'(\theta)\dot{\theta}$,
we finally obtain the equation of motion for the variable length pendulum:
\begin{equation}
\ddot{\theta}+\dot{\theta}\frac{\dot{\Lambda}(\theta)}{\Lambda(\theta)}-g\frac{\dot{\Lambda}(\theta)}{\dot{\theta}\Lambda(\theta)^{2}}\left(\cos\theta-\cos\theta_{0}\right)+\frac{g}{\Lambda(\theta)}\sin\theta=0.\label{e1j}
\end{equation}
Note in \eqref{e1j} the nonlinear operations on the Schwartz distribution
$\Lambda$, on the GSF $\theta$ and the composition $t\mapsto\Lambda(\theta(t))$.
Before showing the numerical solution of \eqref{e1j}, let us consider
the simplest case of the dynamics far from the singularity and that
of small oscillations. The former, as we mentioned above, is the only
physically meaningful one.

\subsubsection{Description far from singularity and small oscillations}

For simplicity, let us consider the simplest case $\theta_{0}=0$.
Furthermore, we consider that the pendulum is initially at rest and
starts its movement at $t_{1}\in\rti$. The initial conditions we
use are hence: 
\begin{equation}
\begin{cases}
\theta(t_{1})=\theta_{1};\\
\dot{\theta}(t_{1})=0,
\end{cases}\label{e1l}
\end{equation}
with $\theta_{1}<0$. Assuming that at some time $t_{3}\in\rti$ we
have $\theta(t_{3})>0$, by the intermediate value theorem for GSF,
there exists $t_{2}\in\rti$ where we have the singularity, i.e.~$\theta(t_{2})=0$
and the length of the pendulum \emph{smoothly} (in the sharp topology)
changes from $L_{1}+L_{2}$ to $L_{2}$ after the rope touches the
parallelepiped. This change happens in an infinitesimal interval,
because by contradiction it is possible to prove that if $\Lambda(\theta)\in(L_{2},L_{1}+L_{2})$,
then $|\theta|\le\frac{-1}{\log\diff{\rho}}\approx0$.
\begin{defn}
Let $x$, $y\in\rti$. We say that $x$ \emph{is far from $y$ }if
$|x-y|\ge\diff{\rho}^{a}$ for all $a\in\R_{>0}$. More generally,
we say that $x$ \emph{is far from} $y$ \emph{with respect to the
class of infinitesimals $\mathcal{I}\subset\rti$}, if $|x-y|\ge i$
for all $i\in\mathcal{I}$.
\end{defn}

\noindent For example, if $|x|\ge r$ for some $r\in\R_{>0}$, then
$x$ is far from $0$, but also the infinitesimal number $x=\frac{-1}{k\log\diff{\rho}}$
($k\in\R_{>0})$ is far from $0$; similarly, the infinitesimal $x=\frac{-1}{k\log\log\diff{\rho}}$
if far from $0$ with respect to all the infinitesimals of the type
$\frac{-1}{h\log\diff{\rho}}$ for $h\in\R_{>0}$.

If $\theta$ is far from $0$ and $b\ge\diff{\rho}^{-a}$, $a\in\R_{>0}$,
then $|b\theta|\ge\diff{\rho}^{-a}|\theta|\ge\diff{\rho}^{-a/2}\ge1$.
Therefore, $H(-\theta)\in\{0,1\}$ and hence $\dot{\Lambda}(\theta(t))=0$.
Equation \eqref{e1j} becomes
\begin{equation}
\theta(t)\text{ is far from }0\ \Rightarrow\ \begin{cases}
\ddot{\theta}+\frac{g}{L_{1}+L_{2}}\sin\theta(t)=0 & \text{if }\theta(t)<0,\\
\ddot{\theta}+\frac{g}{L_{2}}\sin\theta(t)=0 & \text{if }\theta(t)>0.
\end{cases}\label{eq:ODEfar}
\end{equation}
If we assume that $\theta(t_{1})=\theta_{1}<0$ and $\theta(t_{3})>0$
are far from $0$, the sharp continuity of $\theta$ yields the existence
of $\delta_{1}$, $\delta_{3}\in\rti_{>0}$ such that
\begin{align}
\forall t & \in[t_{1},t_{1}+\delta_{1})\cup(t_{3}-\delta_{3},t_{3}]:\ \theta(t)\text{ is far from }0\nonumber \\
\forall t & \in[t_{1},t_{1}+\delta_{1}):\ \theta(t)<0\label{eq:intFar}\\
\forall t & \in(t_{3}-\delta_{3},t_{3}]:\ \theta(t)>0\nonumber 
\end{align}
(and hence $t_{2}\notin[t_{1},t_{1}+\delta_{1})\cup(t_{3}-\delta_{3},t_{3}]$
because $\theta(t_{2})=0$). Assuming that $t_{1}$, $t_{3}$ are
far from $t_{2}$, without loss of generality we can also assume to
have taken $\delta_{i}$ so small that also $t_{1}+\delta_{1}$ and
$t_{3}-\delta_{3}$ are far from $t_{2}$.

We now employ the non Archimedean framework of $\rti$ in order to
formally consider small oscillations, i.e.~$\theta_{1}\approx0$.
We first note that we cannot only assume $\theta_{1}$ infinitesimal,
because if $\theta_{1}$ is not far from $0$ then our solution will
not be physically meaningful. However, we already have seen that we
can take $\theta_{1}$ far from $0$ and infinitesimal at the same
time, e.g.~$\theta_{1}=\frac{-1}{\log\diff{\rho}}$. In other words,
$\theta_{1}$ is a ``large'' infinitesimal with respect to all the
infinitesimals of the form $\diff{\rho}^{a}$. Let $\vartheta_{1}$,
$\vartheta_{3}$ be the solution of the linearized problems
\begin{equation}
\begin{cases}
\ddot{\vartheta}_{1}+\frac{g}{L_{1}+L_{2}}\vartheta_{1}=0, & t_{1}\le t<t_{1}+\delta_{1}\\
\dot{\vartheta}_{1}(t_{1})=0,\ \vartheta_{1}(t_{1})=\theta_{1}
\end{cases}\label{eq:linODE}
\end{equation}
\[
\begin{cases}
\ddot{\vartheta}_{3}+\frac{g}{L_{2}}\vartheta_{3}=0, & t_{3}-\delta_{3}<t\le t_{3}\\
\dot{\vartheta}_{1}(t_{3})=\dot{\theta}(t_{3}),\ \vartheta_{1}(t_{3})=\theta(t_{3}),
\end{cases}
\]

\noindent i.e.~$\vartheta_{1}(t)=\theta_{1}\cos\left(\omega(t-t_{1})\right)$,
$\omega:=\sqrt{\frac{g}{L_{1}+L_{2}}}$, and $\vartheta_{3}(t)=\theta(t_{3})\cos\left(\omega'(t_{3}-t)\right)-\frac{\dot{\theta}(t_{3})}{\omega'}\sin\left(\omega'(t_{3}-t)\right)$,
$\omega'=\sqrt{\frac{g}{L_{2}}}$. We want to show that $\theta(t)\approx\vartheta_{i}(t)$
at least in an infinitesimal neighborhood of $t_{1}$ and $t_{3}$
exactly because $\theta_{1}\approx0$. For simplicity, we proceed
only for $\vartheta_{1}$, the other case being similar. For any $t\in[t_{1},t_{1}+\delta_{1})$,
we have that $\theta(t)<0$ is far from $0$ from \eqref{eq:intFar},
and hence $\ddot{\theta}+\frac{g}{L_{1}+L_{2}}\sin\theta(t)=0$ from
\eqref{eq:ODEfar}. Recalling the initial conditions, we obtain
\[
\theta(t_{1}+h)-\theta_{1}=-\omega^{2}\int_{t_{1}}^{t_{1}+h}\sin\theta(s)\,\diff{s}\quad\forall h\in(0,\delta_{1}).
\]
Similarly, integrating \eqref{eq:linODE}, we get
\[
\vartheta_{1}(t_{1}+h)-\theta_{1}=-\omega^{2}\int_{t_{1}}^{t_{1}+h}\vartheta(s)\,\diff{s}\quad\forall h\in(0,\delta_{1}).
\]
Using Taylor Thm.~\ref{thm:Taylor-1-1} at $t_{1}$ with increment
$h$ of these integral GSF, we obtain
\begin{align*}
\theta(t_{1}+h)-\vartheta(t_{1}+h) & =-\omega^{2}\left\{ \sin\theta_{1}-\theta_{1}+h\cos\theta_{1}\cdot\dot{\theta}(t_{1})-h\dot{\vartheta}_{1}(t_{1})+h^{2}R(h)\right\} =\\
 & =-\omega^{2}\left\{ \sin\theta_{1}-\theta_{1}+h^{2}R(h)\right\} ,
\end{align*}
where $R(-)$ is a suitable GSF. Thereby, $\theta(t_{1}+h)-\vartheta(t_{1}+h)\approx-\omega^{2}h^{2}R(h)\approx0$
for all $h\approx0$ sufficiently small because $\sin\theta_{1}\approx\theta_{1}$
since $\theta_{1}\approx0$.

Since each $t\in[t_{1},t_{1}+\delta_{1})\cup(t_{3}-\delta_{3},t_{3}]$
is far from $t_{2}$, we can also formally join the two solutions
$\vartheta_{i}$ using the Heaviside's function:
\begin{equation}
\theta(t)\approx\vartheta_{1}(t)+H(t_{2}-t)\left(\vartheta_{3}(t)-\vartheta_{1}(t)\right)\quad\forall t\in[t_{1},t_{1}+h)\cup(t_{3}-h,t_{3}].\label{eq:join}
\end{equation}
For the epistemological motivations previously stated, this infinitesimal
approximation cannot be extended to a neighborhood of $t_{2}$.

We close this section noting that all these deductions can be repeated
using any GSF $H\in\gsf(\rti,\rti)$ satisfying for all $x$ far from
zero $H(x)=1$ if $x>0$ and $H(x)=0$ if $x<0$.

\subsubsection{Numerical Solution}

The numerical solution of equation \eqref{e1j} has been computed
using Mathematica Solver NDSolve (see \cite{NDSolve}). Initial conditions
we used are: 
\begin{equation}
\begin{cases}
\theta(0)=0\text{ rad},\\
\dot{\theta}(0)=1\text{ rad/s}.
\end{cases}\label{eq:initial_cond_case_1}
\end{equation}
The graph of $\theta(t)$, and its derivative $\dot{\theta}(t)$,
based on the Mathematica definitions of $H(x)$ and $\delta(x)$ (see
\cite{Hdelta}) are shown in Figure \ref{fig:damping osc - solution and first derivative}.
\begin{figure}
\centering{}\includegraphics[width=0.75\textwidth]{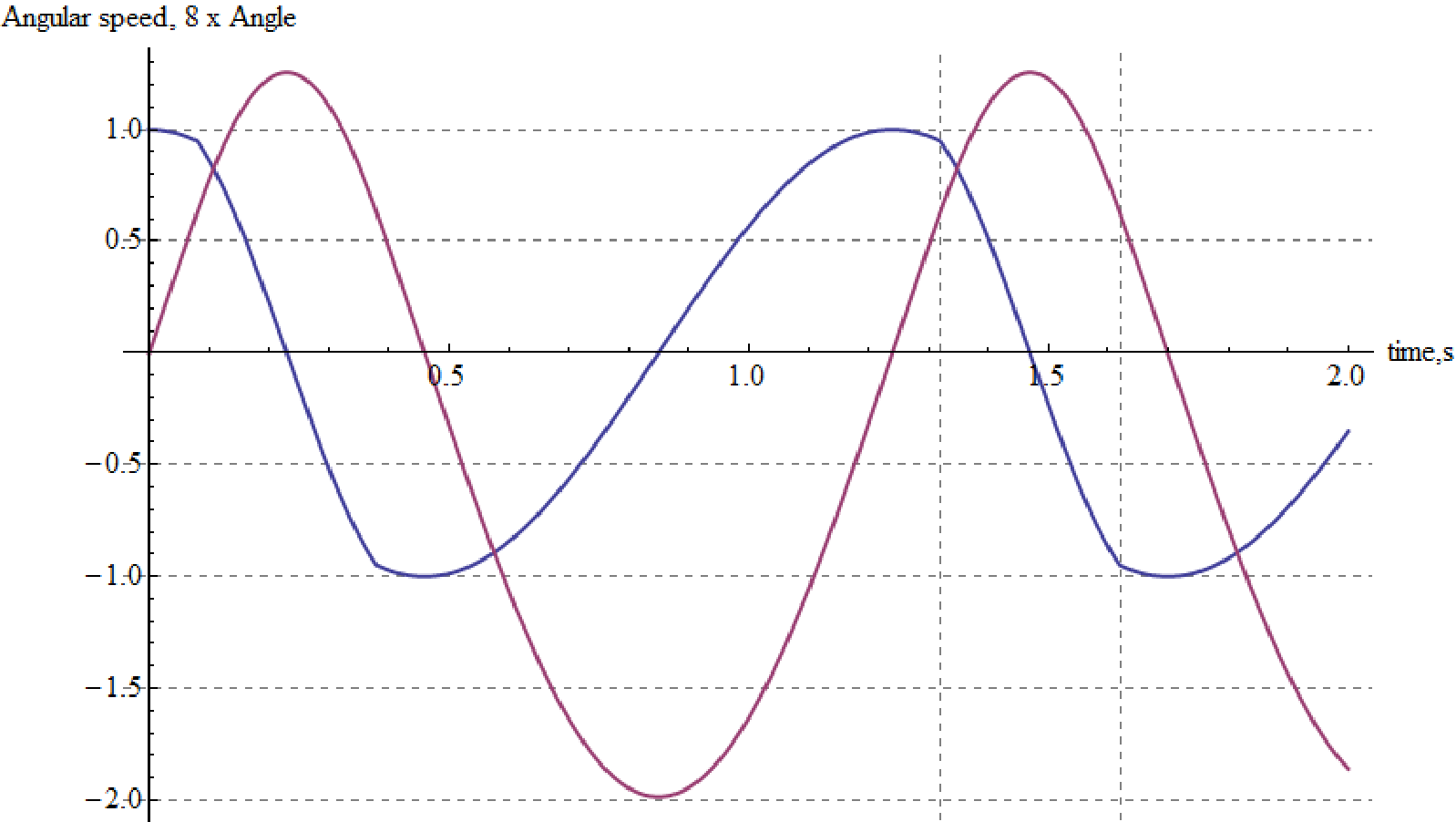}\caption{\label{fig:damping osc - solution and first derivative}8 times re-scaled
solution (violet line) in radians and its derivative in rad/s (at
$\theta=\theta_{0}=\pi/40$ rad we can see a corner point). Parameters
used: $L_{1}=0.4\text{ m}$, $L_{2}=0.2\text{ m}$, $g=9.8\text{ m/\ensuremath{\text{s}^{2}}}$.}
\end{figure}

In Figure \ref{fige2}, we show the second derivative graph. Directly
from \eqref{e1j} and \eqref{e1g} we can prove that when $\theta(t)=\theta_{0}$,
$\ddot{\theta}(t)$ is an infinite number and hence $\dot{\theta}(t)$
has a corner point. Because of the classical Mathematica implementation
of $H$ and $\delta$ we can say that these graphs represent the solution
far from the singularities.
\begin{figure}
\centering{}\includegraphics[width=0.75\textwidth]{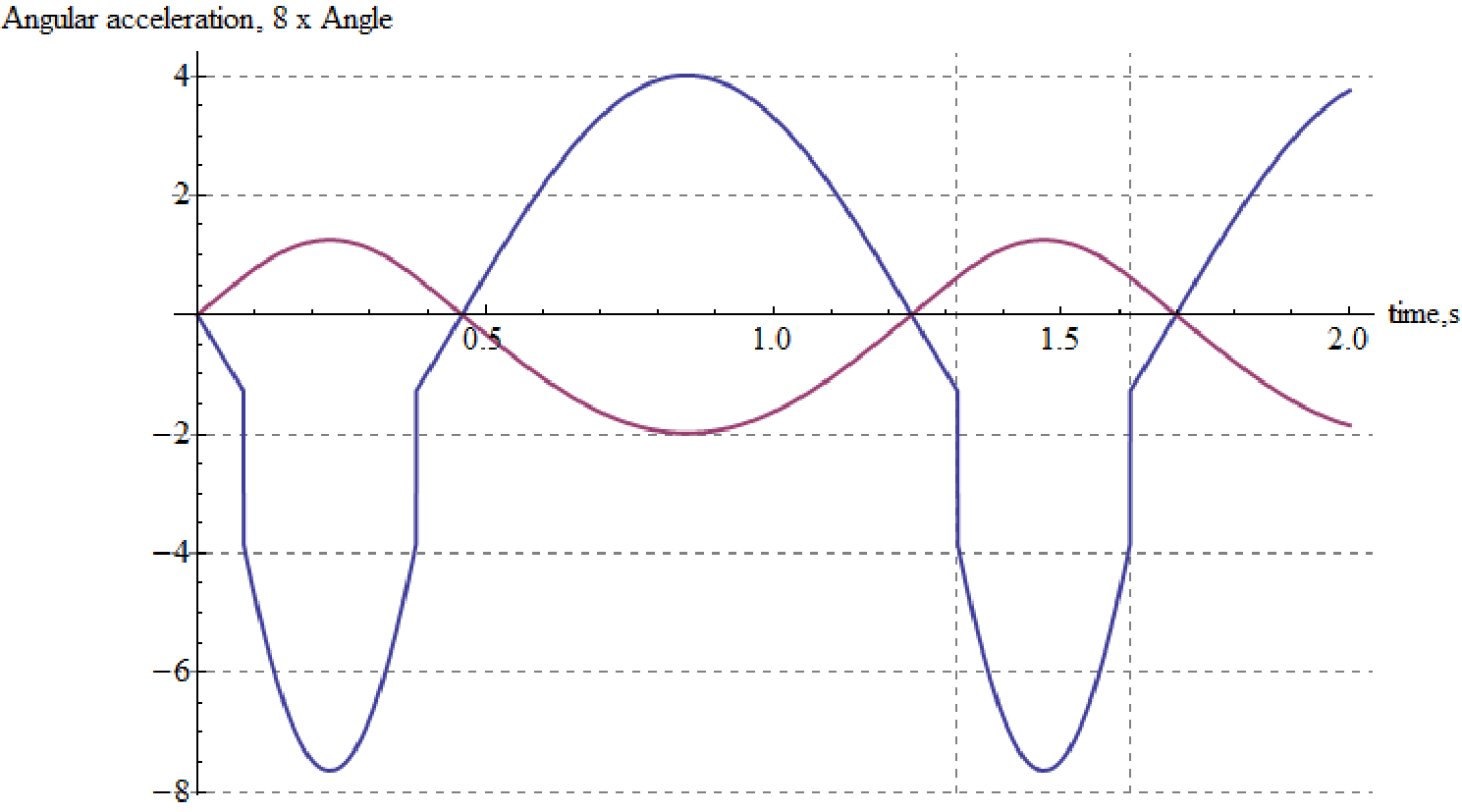}\caption{\label{fige2}8 times re-scaled solution in radians (violet line)
and its second derivative in rad/$\text{s}^{2}$}
\end{figure}

\subsection{Oscillations damped by two media}

\begin{figure}
\centering{}\includegraphics[scale=0.2]{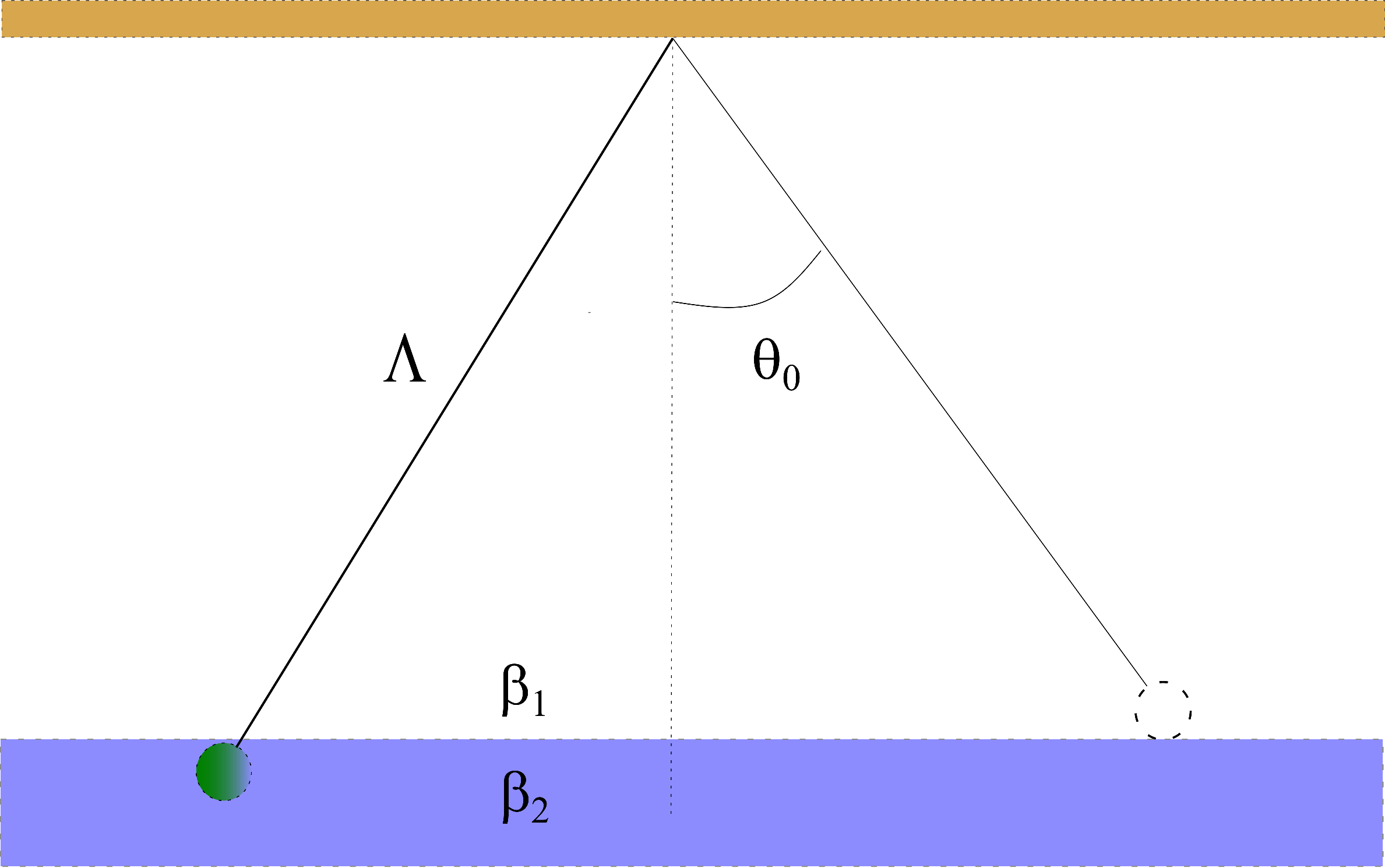}\caption{\label{fig:water_oscillations}Simple pendulum moving in two media}
\end{figure}

The second example concerns oscillations of a pendulum in the interface
of two media. Since we are not interested at the dynamics occurring
at singular times (i.e.~at the changing of the medium), this can
be considered only a toy model approximating the case of a very small
but sufficiently heavy moving particle.

We hence want to model the system employing a ``jump'' in the damping
coefficient $\beta$, i.e.~a finite change occurring in an infinitesimal
interval of time, see Fig.~\ref{fig:water_oscillations}. Since the
frictional forces acting in this case are not conservative, it is
well-known that the Euler-Lagrange equations cannot be assumed to
describe the dynamics of the system and we have to use the D'Alembert
principle, see \cite{FGBL} for details.

The kinetic energy is given by: 
\begin{equation}
T(\dot{\theta})=\frac{1}{2}m\dot{\theta}^{2}\Lambda^{2},\label{e2c}
\end{equation}
and the potential energy (the zero level is the suspension point of
the pendulum) is: 
\begin{equation}
U(\theta)=-mg\Lambda\cos\theta.\label{e2d}
\end{equation}

\noindent In case of fluid resistance proportional to the velocity,
we can introduce the generalized forces $Q$ as: 
\begin{equation}
Q(\dot{\theta})=-r\Lambda^{2}\dot{\theta},\label{e2e}
\end{equation}
where $r$ is a proportional coefficient depending on the media. Let's
define the Lagrangian \textit{$L$} as 
\begin{equation}
L(\theta,\dot{\theta}):=T(\dot{\theta})-U(\theta).\label{e2a}
\end{equation}
We hence assume that the equation of motion for this non-conservative
system is given by the D'Alembert's principle, i.e.
\begin{equation}
\frac{\diff{}}{\diff{t}}\frac{\partial L}{\partial\dot{\theta}}-\frac{\partial L}{\partial\theta}=Q.\label{e2b}
\end{equation}
Inserting \eqref{e2c}, \eqref{e2d} and \eqref{e2e} into \eqref{e2b}
we obtain the following equation of motion: 
\begin{equation}
m\Lambda^{2}\ddot{\theta}+mg\Lambda\sin\theta=-r\Lambda^{2}\dot{\theta}.\label{e2f}
\end{equation}
By introducing the damping coefficient $\beta(\theta):=r(\theta)/(2m)$
(we clearly assume that the mass $m\in\rti>0$ is invertible) we obtain
the classical form of the equation of motion for damped oscillations:
\begin{equation}
\ddot{\theta}+2\beta(\theta)\dot{\theta}+\frac{g\sin\theta}{\Lambda}=0.\label{eq:damping_equation}
\end{equation}

If the pendulum crosses the boundary between two media with damping
coefficients $\beta_{1}$ and $\beta_{2}$, we can model the system
using the Heaviside function $H$:
\begin{equation}
\beta(\theta)=\beta_{1}+\left(H(\theta+\theta_{0})-H(\theta-\theta_{0})\right)(\beta_{2}-\beta_{1}),\label{eq:beta_dependence}
\end{equation}
where $\theta=\pm\theta_{0}$ are the angles at which we have the
changing of the medium (singularities).

The numerical solution of \eqref{eq:damping_equation} with $\beta$
defined by \eqref{eq:beta_dependence} and initial conditions \eqref{eq:initial_cond_case_1}
is presented in Fig.~\ref{fig:Solution-beta-gen-1}. The numerical
solution has been computed using Mathematica Solver NDSolve, but with
an implementation of the Heaviside's function $H$ corresponding to
Thm\@.~\ref{thm:embeddingD'}, i.e.~as represented in Fig.~\ref{fig:MollifierHeaviside}.

\begin{figure}
\centering{}\includegraphics[scale=0.3]{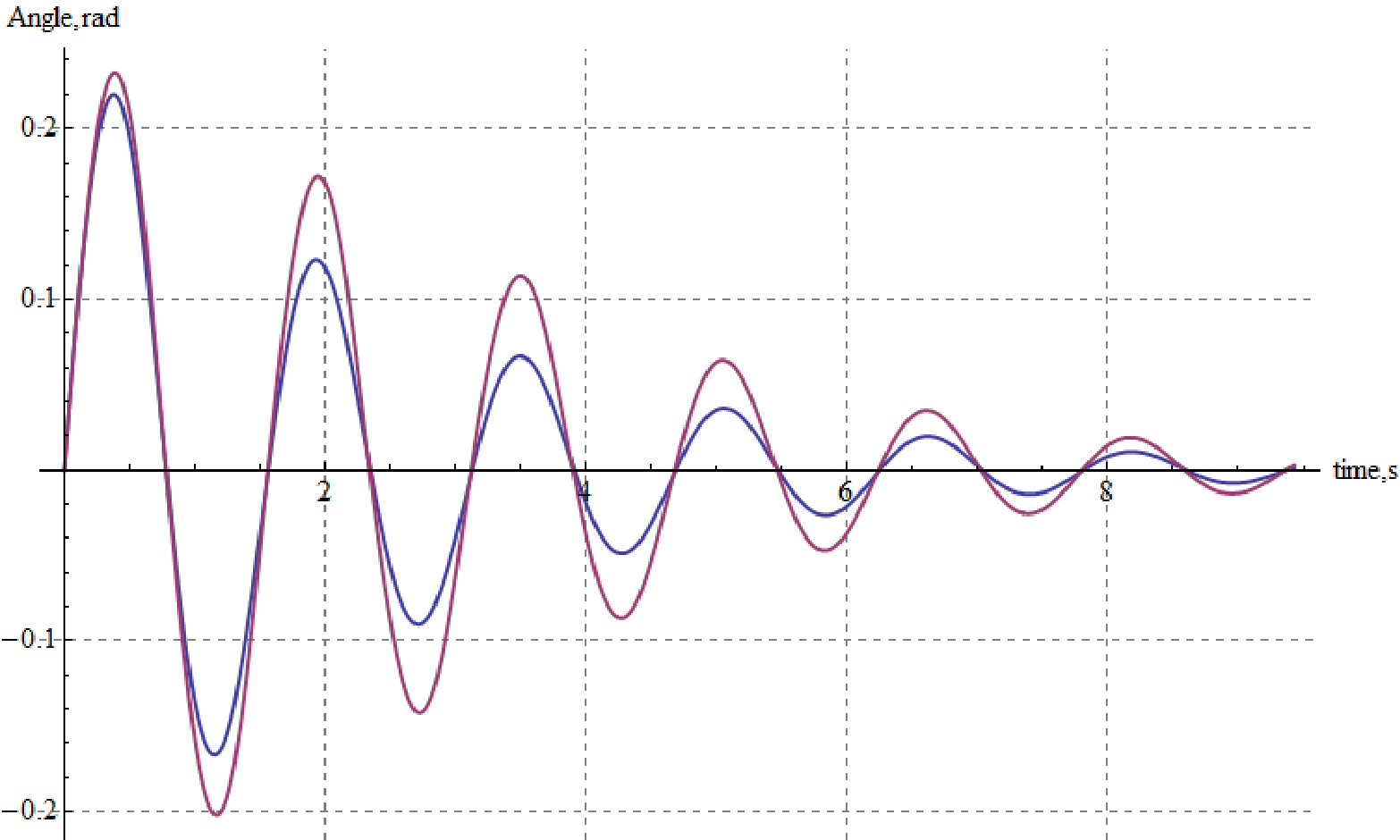}\caption{\label{fig:Solution-beta-gen-1}Solution $\theta$ of \eqref{eq:damping_equation}
(blue line). For comparison, the violet line is the case $\beta=\text{const}.=\beta_{1}$.
Used parameters: $\beta_{1}=0.0064$ (air), $\beta_{2}=0.3859$ (water),
$\theta_{0}=\pi/40\text{ rad}$, $\Lambda=0.6\text{ m}$, $g=9.8\text{ m/\ensuremath{\text{s}^{2}}}$.}
\end{figure}

We also include the graphs of the angular frequency $\dot{\theta}$
(which shows corner points) and of the angular acceleration $\ddot{\theta}$
(which shows ``jumps'', i.e.~infinite derivatives at singular times,
as we can directly see from \eqref{eq:damping_equation} and \eqref{eq:beta_dependence}).

\begin{figure}
\centering{}\includegraphics[scale=0.25]{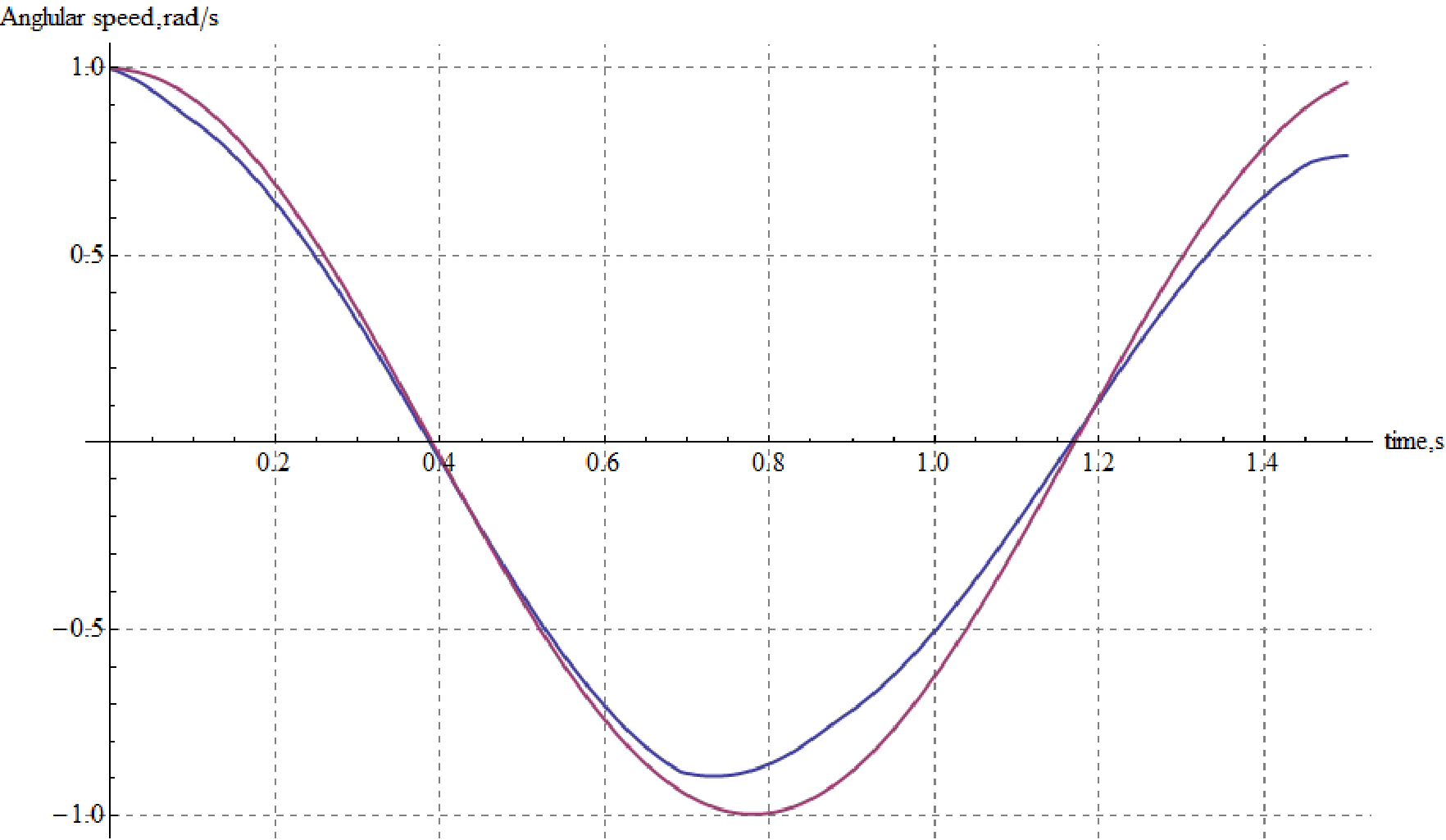}\includegraphics[scale=0.4]{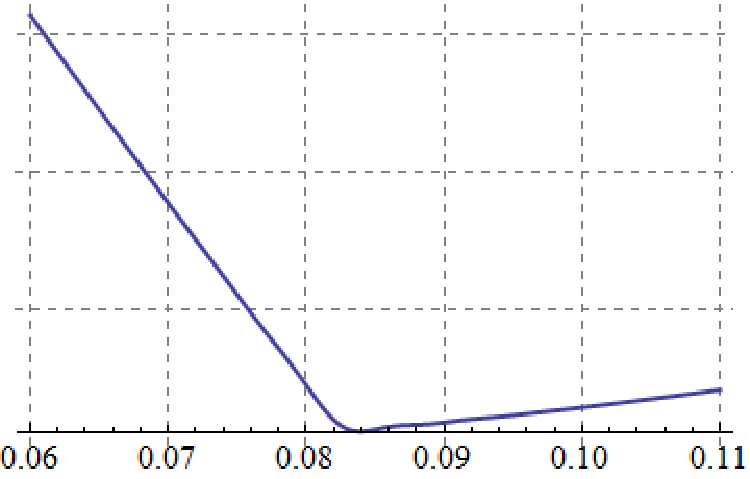}\caption{\label{fig:Difference-in-Derivative-Damping-1}First derivative $\dot{\theta}$
of the solution of \eqref{eq:damping_equation} (blue line). The case
with $\beta=\text{const}=\beta_{1}$ is also shown for comparison
(violet line). Note the corner points at the singular moments, for
example at $t=0.083\text{ s}$ (scaled in the right figure).}
\end{figure}

\begin{figure}
\begin{centering}
\includegraphics[scale=0.2]{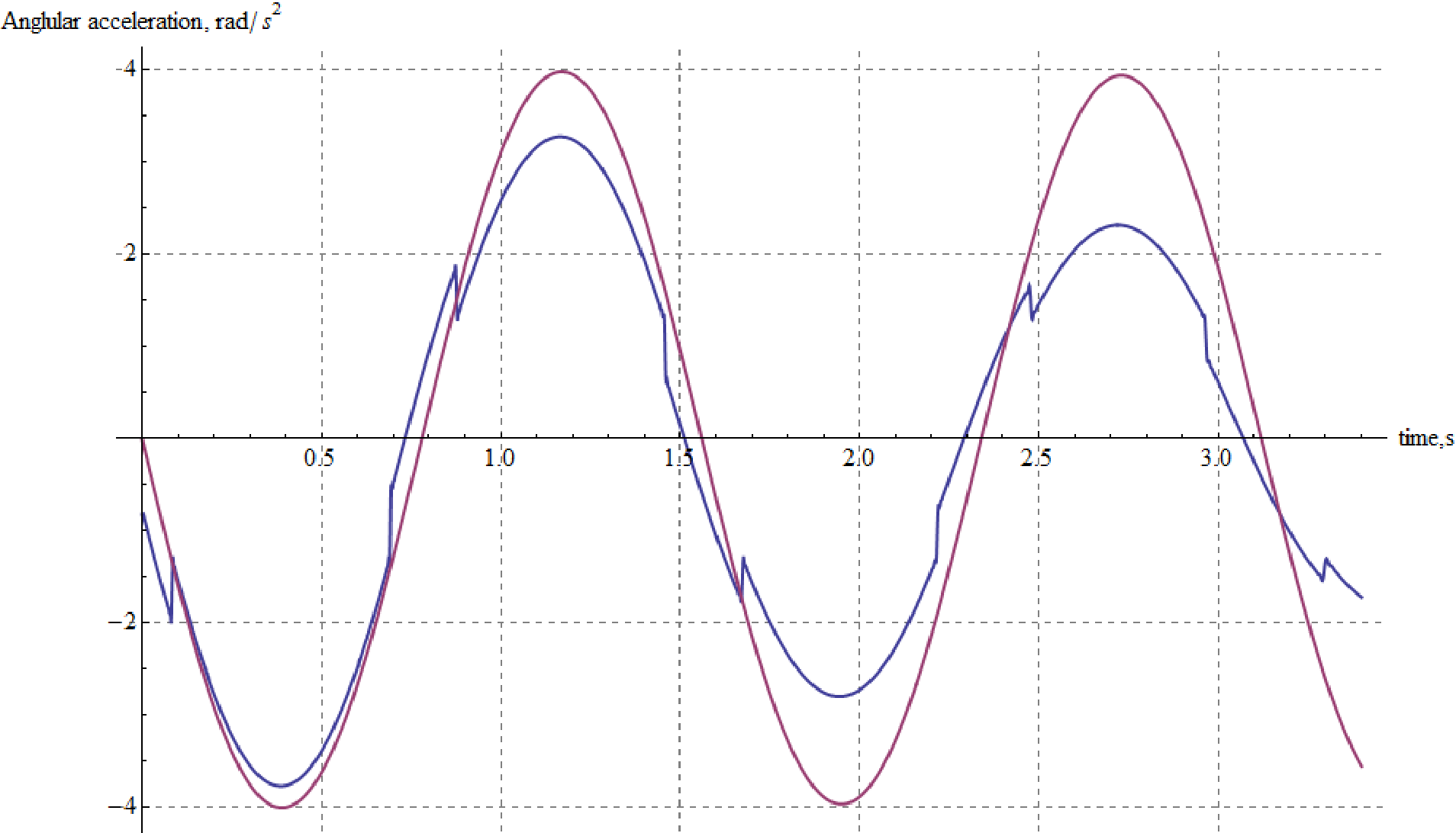}\includegraphics[scale=0.45]{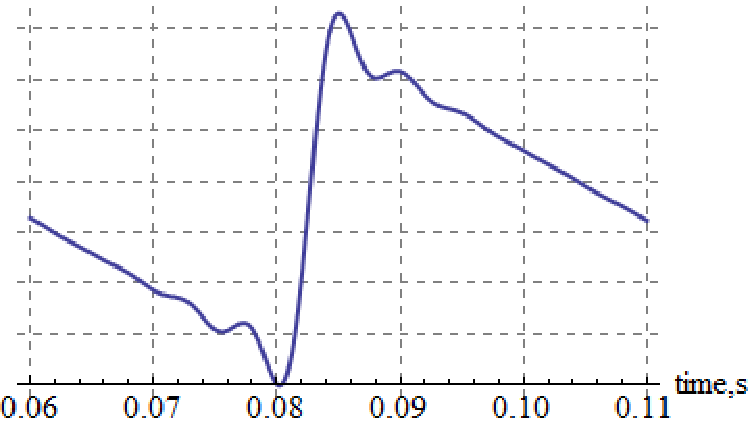}
\par\end{centering}
\caption{Second derivative $\ddot{\theta}$ of the solution of \eqref{eq:damping_equation}
(blue line). The case with $\beta=\text{const}=\beta_{1}$ is also
shown for comparison (violet line). Note the \textquotedblleft jumps\textquotedblright{}
at the singular moments, for example at $t=0.083\text{ s}$ (scaled
in the right figure). The infinitesimal oscillations are caused by
the embedding as GSF of the Heaviside function.}
\end{figure}

\subsection{Non linear strain-stress model}

In this section, we want to show how to construct a mathematical model
starting from an empirical function (the strain-stress relation for
a steel sample) and representing it as a GSF. Starting from Newton's
second law, we hence arrive at a single nonlinear equation describing
the behaviour of the steel sample. Since the empirical function is
not differentiable at the end of the linear part, the use of GSF is
therefore essential.

The strain-stress curve we consider is shown in Fig.~\ref{fig:Strain-stress-experiment}.

\begin{figure}
\includegraphics[scale=0.4]{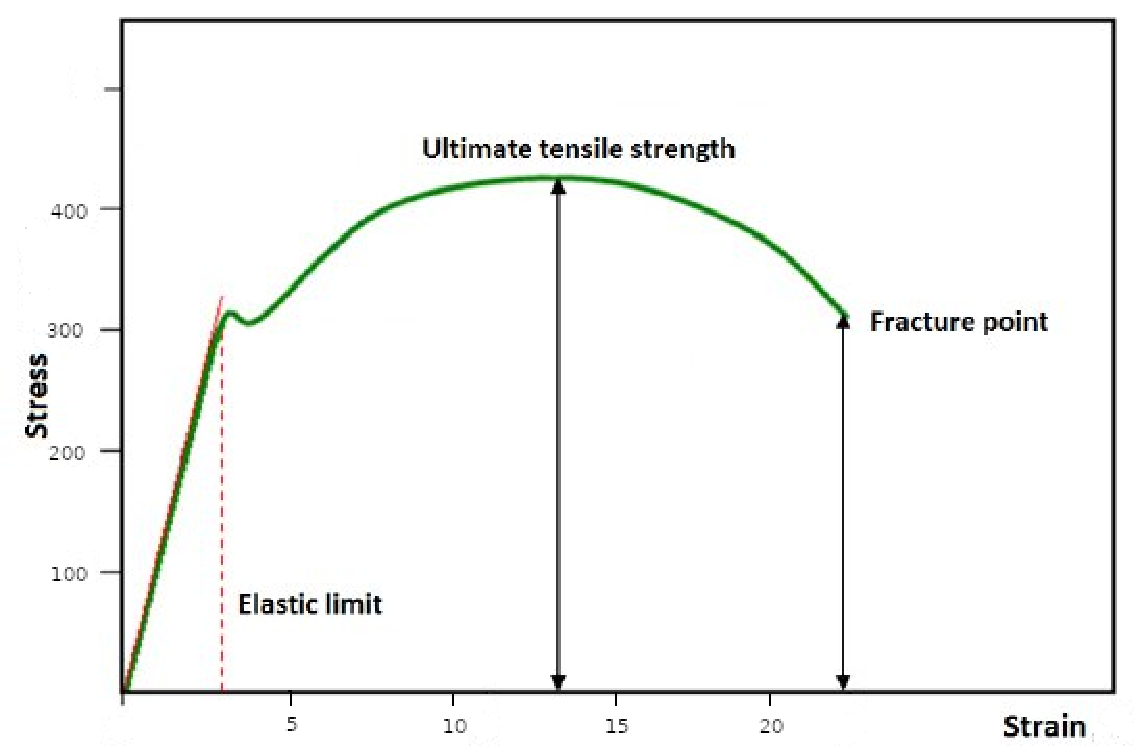}\caption{\label{fig:Strain-stress-experiment}Strain-stress empirical model
of the steel (see \cite{Otani}).}
\end{figure}

\noindent We recall that stress can be defined as $\sigma=\frac{F}{S_{0}}$,
where $F$ is the force applied to the sample, and $S_{0}$ is the
initial cross-section of the cylindrical sample. The strain $\varepsilon$
is usually introduced as $\varepsilon=\frac{L-L_{0}}{L_{0}}$, where
$L_{0}$ is the unstressed length and $L$ is the length after force
application. In order to reproduce the experimental dependence of
Fig.~\ref{fig:Strain-stress-experiment} we considered the parameters
$d=0.37\,$mm for the diameter of the steel cilinder, and $L_{0}=2.2\,$m
for the unstressed length of the sample. Thus, during the elastic
behaviour (linear part) we have a Young's modulus $E=\frac{\sigma}{\varepsilon}=2.13\cdot10^{11}\,\text{Pa}$,
a stiffness $k=\frac{ES_{0}}{L_{0}}=10423\,\frac{\text{N}}{\text{m}}$,
and hence the magnitude of the linear part of the force is given by
$F_{\text{l}}(x)=kx$. For the nonlinear part of the empirical law,
we use the Mathematica built-in function NonlinearModelFit, see \cite{NonlinModFit}.
The result is shown in Fig.~\ref{fig:Non-linear-part-modelling}.

\begin{figure}
\includegraphics[scale=0.3]{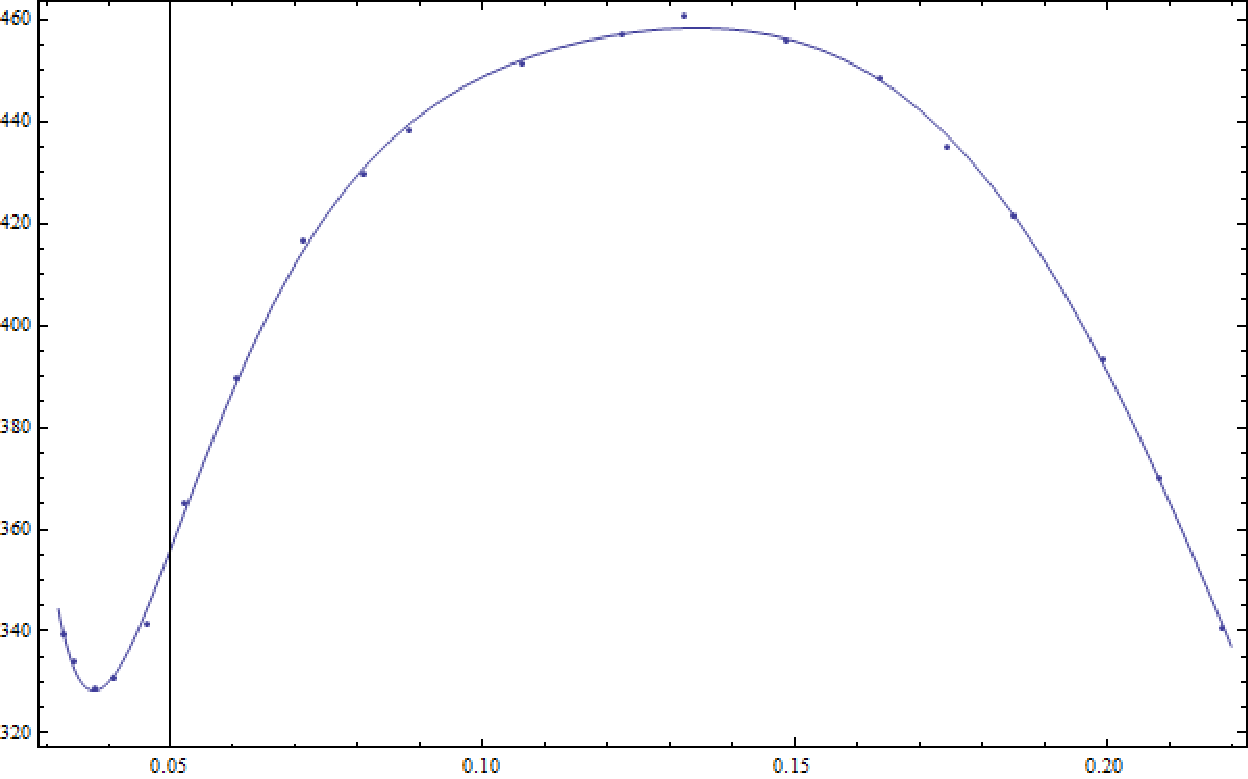}\caption{\label{fig:Non-linear-part-modelling}Non-linear part modelling of
the strain-stress curve using the NonlinearModelFit. The dots are
values extracted from the original data.}
\end{figure}

\noindent The resulting expression is

\begin{equation}
F_{\text{n}}(x)=a_{1}\exp(a_{2}x)+a_{3}\cos(a_{4}x)+a_{5}x+a_{6}x^{2}+a_{7}x^{3}+a_{8}x^{4}+a_{9}x^{5}+a_{10}x^{6}+a_{11}x^{7},\label{eq:nonlinear_force_part}
\end{equation}
where the coefficients $a_{k}$ are given in Tab.~\ref{tab:Coeff}.

\begin{table}
\begin{tabular}{|c|c|c|}
\hline 
$a_{1}=1.5\cdot10^{3}$ & $a_{5}=-9.9\cdot10^{4}$ & $a_{9}=-1.8\cdot10^{9}$\tabularnewline
\hline 
\hline 
$a_{2}=3.9$ & $a_{6}=2.8\cdot10^{6}$ & $a_{10}=4.8\cdot10^{9}$\tabularnewline
\hline 
$a_{3}=3.0$ & $a_{7}=-4.4\cdot10^{7}$ & $a_{11}=-5.1\cdot10^{9}$\tabularnewline
\hline 
$a_{4}=1.0\cdot10^{2}$ & $a_{8}=3.8\cdot10^{8}$ & \tabularnewline
\hline 
\end{tabular}\caption{\label{tab:Coeff}Coefficients used in non-linear part $F_{\text{n}}$
of the force.}
\end{table}
Using the Heaviside function, we can hence write the force $F$ acting
on the steel sample as

\begin{equation}
F(x)=-F_{\text{l}}(x)-(F_{\text{n}}(x)-F_{\text{l}}(x))H(x-x_{0}),\label{eq:force_with_nonlinear_part}
\end{equation}
and it is represented in Fig.~\ref{fig:Recomputed-strain-stress-model}.
Note that the negative sign is due to the fact that the force is directed
opposite to the elongation of the sample; moreover, $x_{0}=0.033$
according to Fig.~\ref{fig:Recomputed-strain-stress-model}. Thus,
the position $x$ of the steel sample satisfies the differential equation

\begin{equation}
\ddot{x}-\frac{F(x)}{m}=0.\label{eq:equation_with_nonlinear_force}
\end{equation}
Once again, note that the GSF $F$ is nonlinear and the term $F(x(t))$
is a composition of GSF.

\begin{figure}
\includegraphics[scale=0.25]{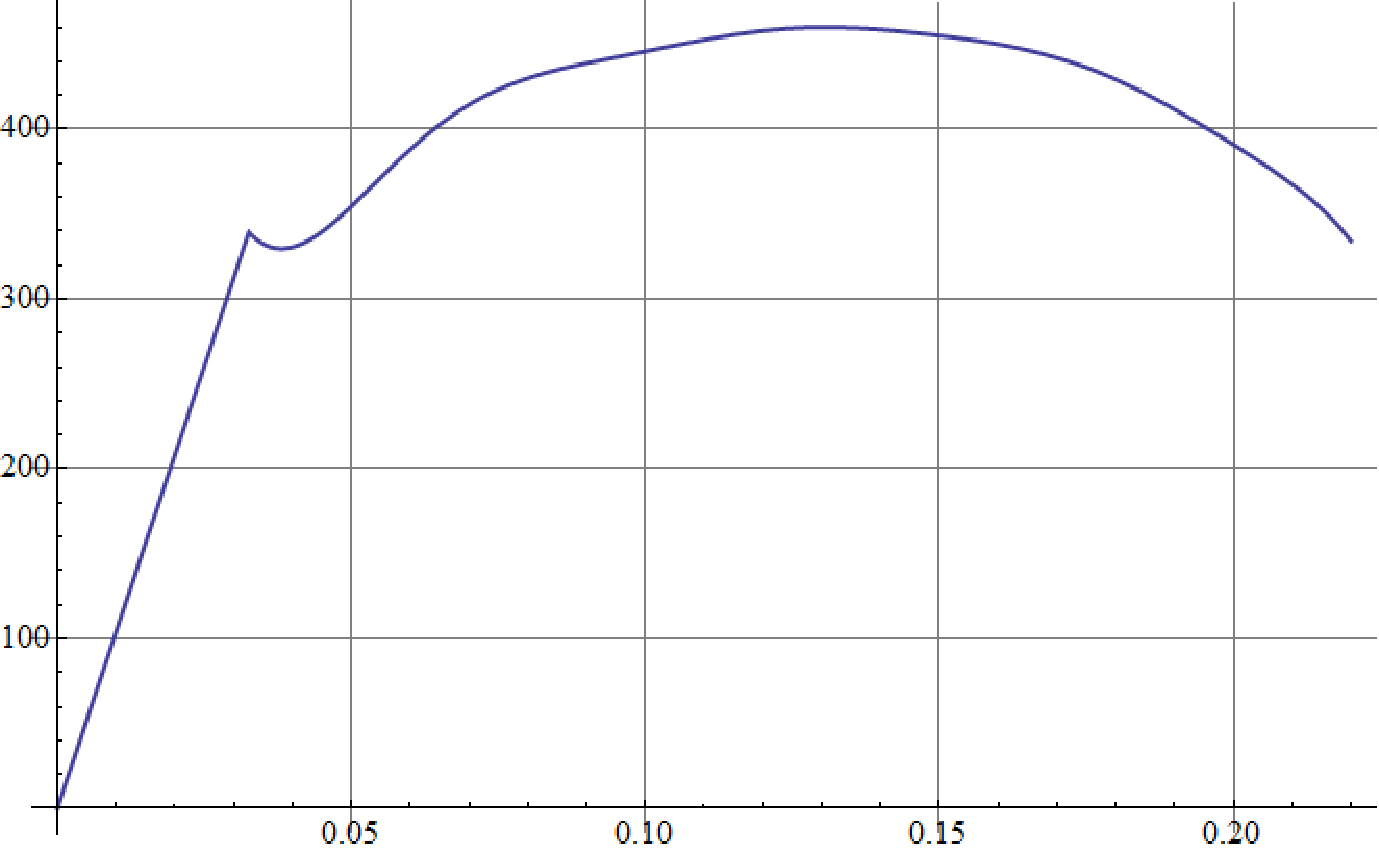}\caption{\label{fig:Recomputed-strain-stress-model}Recomputed strain-stress
model.}
\end{figure}

\noindent Far from the singularity $x=x_{0}$, the validity of \eqref{eq:equation_with_nonlinear_force}
can also be seen using the conservation of the mechanical energy.
In fact, if $x$ is far from $x_{0}$, in the sense that $x\le x_{1}$
for some $x_{1}\in\R_{<0}$, then $F(x)=-F_{\text{l}}(x)$ and we
are in the zone of Hooke's law; we thus have the potential energy:
$U(x)=\frac{kx^{2}}{2}$. Similarly, if $x\ge x_{2}$ for some $x_{2}\in\R_{>0}$,
then $F(x)=-F_{\text{n}}(x)$ and $U(x)=\frac{a_{1}}{a_{2}}\exp(a_{2}x)+\frac{a_{3}}{a_{4}}\sin(a_{4}x)+\frac{a_{5}}{2}x^{2}+\frac{a_{6}}{3}x^{3}+\frac{a_{7}}{4}x^{4}+\frac{a_{8}}{5}x^{5}+\frac{a_{9}}{6}x^{6}+\frac{a_{10}}{7}x^{7}+\frac{a_{11}}{8}x^{8}$.
Therefore, far from the singularity, the conservation of the mechanical
energy is equivalent to \eqref{eq:equation_with_nonlinear_force}.

Clearly, the nonlinear behaviour depends on the initial conditions:
if $x(0)=0.0\,\text{m}$ and $\dot{x}(0)=5\,\frac{\text{m}}{\text{s}}$,
we remain in the linear setting, whereas for $x(0)=0.0\,\text{m}$,
$\dot{x}(0)=15\,\frac{\text{m}}{\text{s}}$ we enter into the nonlinear
one, see Fig.~\ref{fig:solStressStrain}.

\begin{figure}
\includegraphics[scale=0.3]{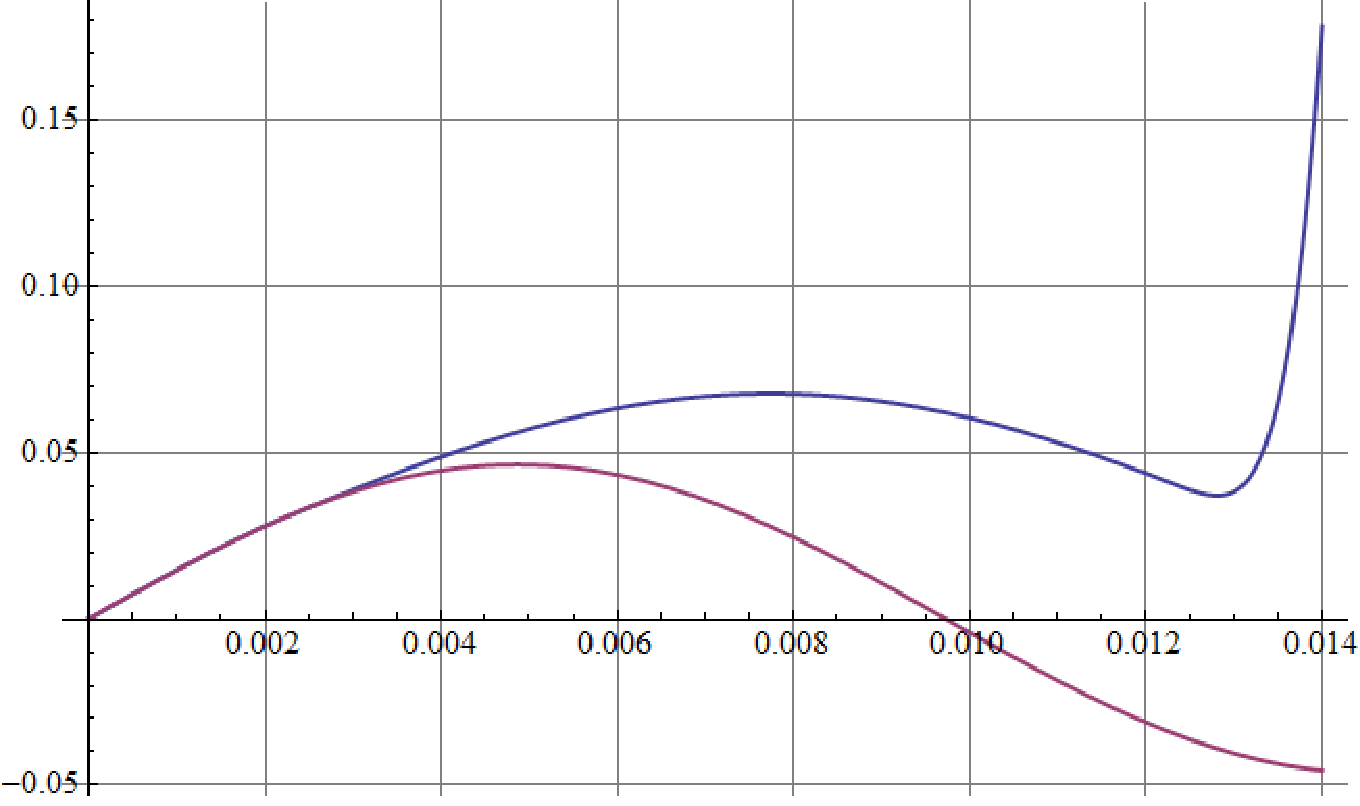}\caption{\label{fig:solStressStrain}The solution $x(t)$ of equation \ref{eq:equation_with_nonlinear_force}
with initial conditions $x(0)=0.0\,\text{m},$$\dot{x}(0)=15\,\frac{\text{m}}{\text{s}}$
(blue line) in comparison with the solution in the linear setting
for $x(0)=0.0\,\text{m}$, $\dot{x}(0)=5\,\frac{\text{m}}{\text{s}}$
(violet line).}
\end{figure}

\noindent See also \cite{C1,PruRaj2011,PruRaj2016,PruReTu2017} for
more complete models of this type in the setting of Colombeau theory.

\subsection{Discontinuous Lagrangians in optics (Snell's law derivation)}

A typical example where one would like to use the usual results of
calculus despite dealing with non differentiable functions, is geometrical
optics at the interface of two media, where usually the Lagrangian
function is not smooth. It is well-know, e.g., that the rigorous derivation
of Snell's law is a paradigmatic example, see e.g.~\cite{Zat}. The
main aim of this section is to show the features of the nonlinear
calculus of GSF by deriving Snell's law for plane stratified media
from the classical Fermat's principle. For example, in the following
deduction, the refraction index $n(x)$ can be any GSF, e.g.~the
embedding of a locally integrable function. For more general versions
of Snell's law in different media, see e.g.~\cite{Cole,RoFe}. For
the classical deduction where the refraction index $n(x)$ and the
light path $x=x(s)$ are $\mathcal{C}^{2}$ functions, see e.g.~\cite{LaGhTh11,NoSe}.

We assume that we are considering a body $B:=\langle B_{\eps}\rangle\subseteq\R^{3}$
represented by the strongly internal set generated by the net $B_{\eps}\subseteq\R^{3}$,
and that our light path satisfies the classical Fermat principle,
i.e.~the path of light between two given points $P$, $Q\in B$ is
the one which minimizes the travel time. In order to mathematically
state this principle for GSF, we introduce the space of paths (see
\cite{LeLuGi17,FGBL} for details):
\begin{equation}
\gsf_{\text{bd}}(P,Q):=\left\{ r\in\gsf([0,1],B)\mid r(0)=P,\ r(1)=Q\right\} ,\label{eq:gsfBD}
\end{equation}
and the travel time functional
\begin{equation}
T[r]:=\frac{1}{c}\int_{0}^{1}n(r)\,\diff r:=\frac{1}{c}\int_{0}^{1}n(r(s))\left|\dot{r}(s)\right|\,\diff s\quad\forall r\in\gsf_{\text{bd}}(P,Q).\label{eq:timeFunct}
\end{equation}
As usual, $c$ is the speed of light in vacuum, and $n\in\gsf(B,\rti_{\ge0})$
is the refraction index of the media $B$ we are considering. The
Fermat principle hence implies that the light path $r\in\gsf_{\text{bd}}(P,Q)$
is a weak extremal of the travel time functional $T[-]$, i.e.~it
satisfies
\begin{equation}
\delta T(r;h):=\left.\frac{\diff{}}{\diff x}T[r+xh]\right|_{x=0}=0\quad\forall h\in\gsf_{\text{bd}}(0,0).\label{eq:derT}
\end{equation}
Note that, since in \eqref{eq:gsfBD} we consider only paths $r\in\gsf([0,1],B)$
valued in the strongly internal set $B=\langle B_{\eps}\rangle$,
Thm.~\ref{thm:strongMembershipAndDistanceComplement} implies
\[
\forall h\in\gsf_{\text{bd}}(0,0)\,\exists\delta\in\rti_{>0}\,\forall x\in(-\delta,\delta):\ r+xh\in\gsf([0,1],B),
\]
and therefore, it is correct to consider the derivative in \eqref{eq:derT}.
Physically this means that we are considering only paths which lay
completely inside the body $B=\langle B_{\eps}\rangle$. The weak
extremal condition \eqref{eq:derT} is equivalent to the Euler-Lagrange
equations (see \cite{LeLuGi17}) for the Lagrangian $(r,v)\mapsto L(r,v):=n(r)\sqrt{v\cdot v}\in\gsf(B\times\rti^{3},\rti)$.
We use the notation $L(r,v)=L(r_{1},r_{2},r_{3},v_{1},v_{2},v_{3})$
for the variables of $L$. We also explicitly note the nonlinear operations
in this Lagrangian, and the composition $n(r(s))$ of GSF. We use
the customary notations $\vec{v}(s):=\frac{\diff{r}}{\diff s}(s)\in\rti^{3}$,
$v(s):=|\vec{v}(s)|\in\rti$, and $L[r](s):=L(r(s),\vec{v}(s))$.
We hence get
\begin{equation}
\frac{\diff{}}{\diff s}\left(\frac{\partial L}{\partial v_{j}}[r](s)\right)=\frac{\partial L}{\partial r_{j}}[r](s),\quad\forall j=1,2,3,\ \forall s\in[0,1].\label{eq:Lagr}
\end{equation}
We always assume that the frame of reference is chosen so that the
light path satisfies $v(s)\in\rti_{>0}$. Calculating the derivatives
$\partial L/\partial v_{j}$ and $\partial L/\partial r_{j}$ in Euler-Lagrange
equations, we get
\begin{align*}
\frac{\partial L}{\partial v_{j}}(r,v) & =n(r)\frac{v_{j}}{\sqrt{v\cdot v}}\\
\frac{\partial L}{\partial r_{j}}(r,v) & =\frac{\partial n}{\partial r_{j}}(r)\sqrt{v\cdot v},
\end{align*}
for all $(r,v)\in B\times\rti^{3}$ and all $j=1,2,3$. Substituting
in \eqref{eq:Lagr}, we obtain the eikonal equation

\begin{equation}
\frac{\diff{}}{\diff s}\left(n(r(s))\frac{\vec{v}(s)}{v(s)}\right)=\nabla n(r(s))v(s),\quad\forall s\in[0,1].\label{eq:Eikonal}
\end{equation}

We now consider the case of a plane stratified media, i.e.~where
$n$ changes only along one direction $\vec{k}\in\rti^{3}$, $|\vec{k}|=1$,
so that 
\begin{equation}
\nabla n(r(s))\parallel\vec{k}\quad\forall s\in[0,1].
\end{equation}

\noindent For simplicity, where it is clear, we omit the evaluation
at $s$. Thus, the cross product of the two vectors $\nabla n(r)$
and $\vec{k}$ is $\vec{k}\times\nabla n(r)=0=\vec{k}\times\nabla n(r)v$.
Using \eqref{eq:Eikonal} we get $\vec{k}\times\frac{\diff{}}{\diff s}(n(r)\frac{\vec{v}}{v})=0$,
and hence $\frac{\diff{}}{\diff s}(\vec{k}\times n(r(s))\frac{\vec{v}(s)}{v(s)})=0$
for all $s\in[0,1]$, i.e.~the function $s\in[0,1]\mapsto\vec{k}\times n(r(s))\frac{\vec{v}(s)}{v(s)}\in\rti^{3}$
is a constant $\vec{C}\in\rti^{3}$. Taking its magnitude $C=|\vec{C}|$
\[
\forall s\in[0,1]:\ C=|\vec{k}|\left|n(r(s))\right|\left|\frac{\vec{v}(s)}{v(s)}\right|\sin\phi(s),
\]
where $\phi(s)$ is the angle between $\vec{k}$ and $\vec{v}(s)=\frac{\diff{r}}{\diff s}(s)$.
We proved Snell's law for GSF:
\begin{thm}
\label{thm:Snell}Let $B_{\eps}\subseteq\R^{3}$, $B:=\langle B_{\eps}\rangle$,
$P$, $Q\in\langle B_{\eps}\rangle$, $n\in\gsf(B,\rti_{\ge0})$.
Assume that $r\in\gsf_{\text{\emph{bd}}}(P,Q)$ is a weak extremal
of the travel time functional \eqref{eq:timeFunct}, i.e.~it satisfies
\eqref{eq:derT}. Set $\vec{v}(s):=\frac{\diff{r}}{\diff s}(s)\in\rti^{3}$,
$v(s):=|\vec{v}(s)|\in\rti$ and assume that $v(s)>0$ for all $s\in[0,1]$.
Then the eikonal equation \eqref{eq:Eikonal} holds. Moreover, if
$\nabla n(r(s))\parallel\vec{k}$, where $\vec{k}\in\rti^{3}$, $|\vec{k}|=1$,
and $\phi(s)$ is the angle between $\vec{k}$ and $\vec{v}(s),$then
the quantity $n(r(s))\cdot\sin\phi(s)$ is constant for all $s\in[0,1]$.
\end{thm}

\subsection{Finite and infinite step potential}

Models of quantum mechanics such as the potential well or the step
potential with finite or infinite walls are clear and simple examples
showing features of various quantum mechanical effects, see e.g.~\cite{CoDiLa}.
However, the mathematics of such models is not very clear sometimes,
see again e.g.~\cite[pag.~34-40,~pag.~68]{CoDiLa} and authors' comments
about mathematical rigour. Once again, in this section we see how
the formalism of GSF theory allows one to completely recover a mathematically
and physically clear proof by formalizing the intuitive steps of \cite[pag.~68]{CoDiLa}.
We consider the step-potential problem, where the high of the potential
can be any finite or infinite generalized number; a similar approach
can be used for the rectangular potential wells.

In the following, we write $x\ll0$ if $\exists r\in\R_{<0}:\ x<r$,
and similarly for $x\gg0$, and we simply say that $x$ \emph{is far
from $0$}. The step function potential for the one-dimensional stationary
Schrodinger equation is a GSF $U\in\gsf(\rti,\rti)$ such that

\begin{equation}
U(x)=\begin{cases}
0 & x\ll0,\\
U_{0} & x\gg0,
\end{cases}\label{eq:step_potential}
\end{equation}
where $U_{0}\in\rti_{>0}$ is an arbitrary generalized number (finite
or infinite). For example, $U(x)=H(x)\cdot U_{0}$ satisfies these
conditions. However, as stated in \cite[pag.~34]{CoDiLa}, this is
actually an idealized model of the potential, and we cannot say it
is a physically meaningful model for infinitesimal $x\approx0$ (similarly
to what we have already seen e.g.~in Sec.~\ref{subsec:Singular-variable-length}
for the singular variable length pendulum).

The system satisfies the stationary Schrodinger's equation 
\begin{equation}
\left[-\frac{\hbar^{2}}{2m}\frac{\diff{}^{2}}{\diff x^{2}}+U(x)\right]\psi(x)=E\psi(x),\label{eq:Schrodinger-stationary}
\end{equation}
where $\hbar$ is the Planck's constant, $m\in\R_{>0}$ the mass of
the particle, $E$ the energy, and $\psi(x)$ the wave function. Using
Thm.~\ref{thm:linearODE}, we can state that there exists a $\psi\in\gsf(\rti,\rti)$
satisfying \eqref{eq:Schrodinger-stationary}. Repeating exactly the
usual calculations, for $x\ll0$ we have:
\begin{equation}
\psi(x)=\frac{1}{\sqrt{k_{1}}}\left(A_{1}e^{ik_{1}x}+A_{2}e^{-ik_{1}x}\right),\label{eq:psiLeft}
\end{equation}
where $k_{1}:=\sqrt{2mE/\hbar^{2}}$ and $A_{1}$, $A_{2}\in\rti$
are undefined constants. For $x\gg0$, we have:

\begin{equation}
\psi(x)=\frac{1}{\sqrt{k_{2}}}\left(B_{1}e^{ik_{2}x}+B_{2}e^{-ik_{2}x}\right),\label{eq:psiRight}
\end{equation}
where $k_{2}=\sqrt{2m(E-V_{0})/\hbar^{2}},$ $B_{1}$, $B_{2}\in\rti$
are undefined constants. As stated in \cite[pag.~68]{CoDiLa}, in
order to find these constants, we need some mathematically careful
steps to justify the corresponding initial conditions. Take any \emph{standard
real number} $\eta\in\R_{>0}$ and integrate \eqref{eq:Schrodinger-stationary}
on $[-\eta,\eta]\subseteq\rti$ to get
\begin{equation}
\frac{\diff{\psi}}{\diff x}(\eta)-\frac{\diff{\psi}}{\diff x}(-\eta)=\frac{2m}{\hbar^{2}}\int_{-\eta}^{\eta}\left[U(x)-E\right]\psi(x)\,\diff x.\label{eq:diffDer}
\end{equation}
As in \cite[pag.~68]{CoDiLa}, we assume that
\begin{align}
U(x)-E\text{ is finite for all finite }x\in\rti\label{eq:finite1}\\
\frac{\diff{\psi}}{\diff x}(\eta)\text{ is finite for all }\eta\in\R_{>0} & .\label{eq:finite2}
\end{align}
From \eqref{eq:diffDer} and the first of these assumptions, we obtain
\[
\left|\frac{\diff{\psi}}{\diff x}(\eta)-\frac{\diff{\psi}}{\diff x}(-\eta)\right|\le\frac{4m}{\hbar^{2}}\cdot C\cdot\eta
\]
for some $C\in\R_{>0}$ (coming from \eqref{eq:finite1}), i.e.
\begin{equation}
\lim_{\substack{\eta\to0^{+}\\
\eta\in\R_{>0}
}
}\left|\frac{\diff{\psi}}{\diff x}(\eta)-\frac{\diff{\psi}}{\diff x}(-\eta)\right|=0.\label{eq:FermatLimDer}
\end{equation}
Similarly, from \eqref{eq:finite2} and the fundamental theorem of
calculus for GSF Thm.~\ref{thm:intRules}.\ref{enu:fundamental},
we have
\[
\left|\psi(\eta)-\psi(0)\right|=\left|\int_{0}^{\eta}\frac{\diff{\psi}}{\diff x}(x)\,\diff x\right|\le\bar{C}\cdot\eta,
\]
for some $\bar{C}\in\R_{>0}$ (coming from \eqref{eq:finite2}), and
hence
\begin{equation}
\lim_{\substack{\eta\to0^{+}\\
\eta\in\R_{>0}
}
}\left|\psi(\eta)-\psi(0)\right|=0.\label{eq:FermaLimFnc}
\end{equation}
 Recall that from Thm.~\ref{thm:GSF-continuity}.\ref{enu:GSF-cont}
and Thm.~\ref{thm:FR-forGSF} it directly follows that both $\psi(x)$
and its derivative $\frac{\diff{\psi}}{\diff x}(x)$ are GSF, and
hence they are continuous in the sharp topology at each point $x\in\rti$.
Stated explicitly at $x=0$, this means that
\begin{align*}
\exists\lim_{\substack{x\to0\\
x\in\,\rti
}
}\frac{\diff{\psi}}{\diff x}(x) & =\frac{\diff{\psi}}{\diff x}(0),\\
\exists\lim_{\substack{x\to0\\
x\in\,\rti
}
}\psi(x) & =\psi(0),
\end{align*}
and these are different than \eqref{eq:FermatLimDer} and \eqref{eq:FermaLimFnc},
where $\eta\in\R_{>0}$. Indeed, balls $B_{\eta}(c)\subseteq\rti$,
for radii $\eta\in\R_{>0}$, generate a different topology on $\rti$
(called \emph{Fermat topology}, see e.g.~\cite{Gio-Kun-Ver15,Gio-Kun-Ver19}).
See Fig.~\ref{fig:The-fist-derivative-and-second-derivatives} for
an intuitive diagram of the solution $\psi$ in an infinitesimal neighborhood
of $x=0$: whereas $\frac{\diff{\psi}}{\diff x}(\eta)$ is continuous
for $\eta\to0$, $\eta\in\R_{>0}$, it is well-known (see \cite{CoDiLa})
that the same property does not hold for the second derivative. In
Fig.~\ref{fig:The-fist-derivative-and-second-derivatives}, the green
lines represents the solution for $x\ll0$ or $x\gg0$, and the blue
one the GSF function $\psi=[\psi_{\eps}(-)]$ (we actually represented
$\psi_{\eps}$ for $\eps$ sufficiently small); we therefore have
to think as infinitesimal the differences between blue and green lines,
and hence as infinite the second derivative at $x=0$.

From \eqref{eq:psiLeft}, \eqref{eq:psiRight} and \eqref{eq:FermatLimDer},
\eqref{eq:FermaLimFnc} we obtain that the constants are uniquely
determined by the system

\begin{equation}
\begin{cases}
\left(A_{1}+A_{2}\right)=\left(B_{1}+B_{2}\right)\\
k_{1}\left(A_{1}-A_{2}\right)=k_{2}\left(B_{1}-B_{2}\right)
\end{cases}
\end{equation}

\begin{figure}
\includegraphics[scale=0.27]{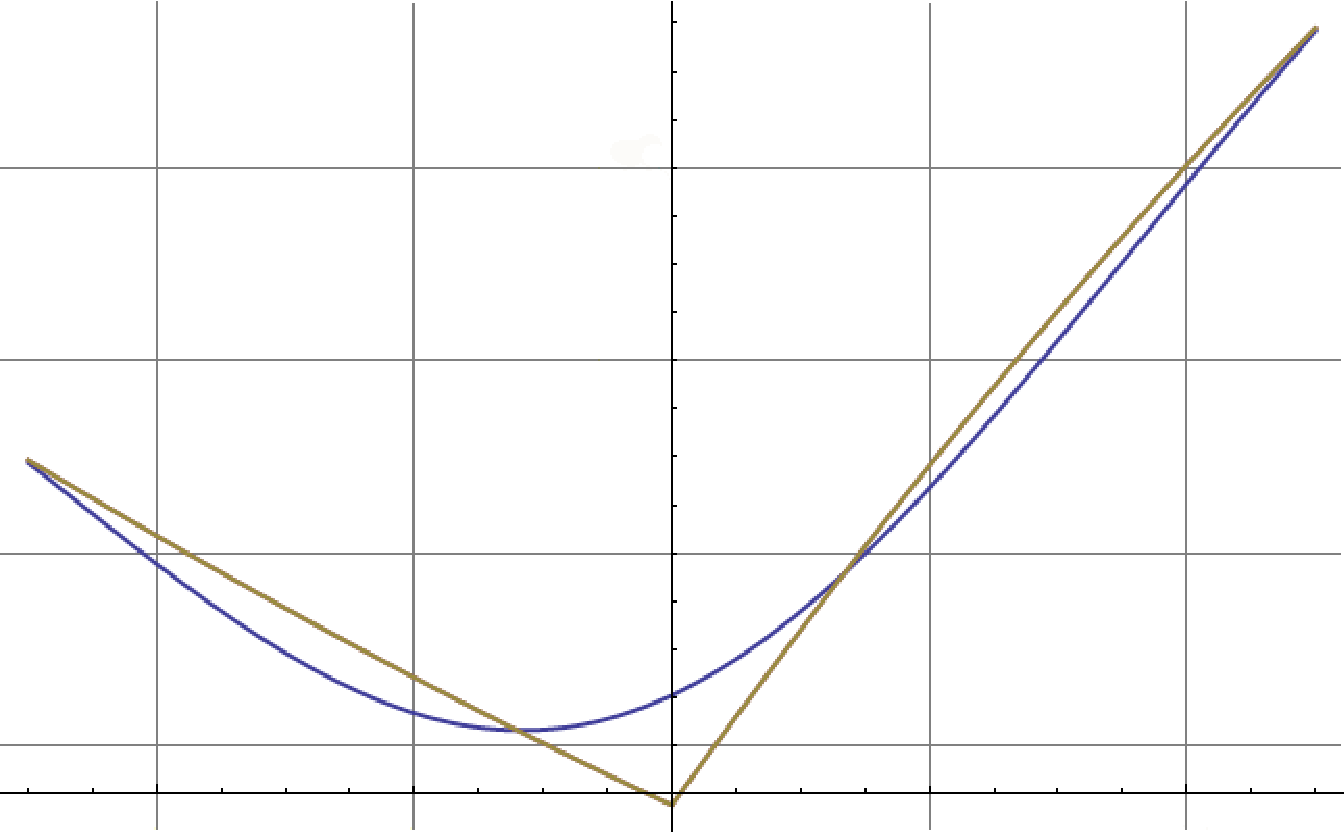}\includegraphics[scale=0.27]{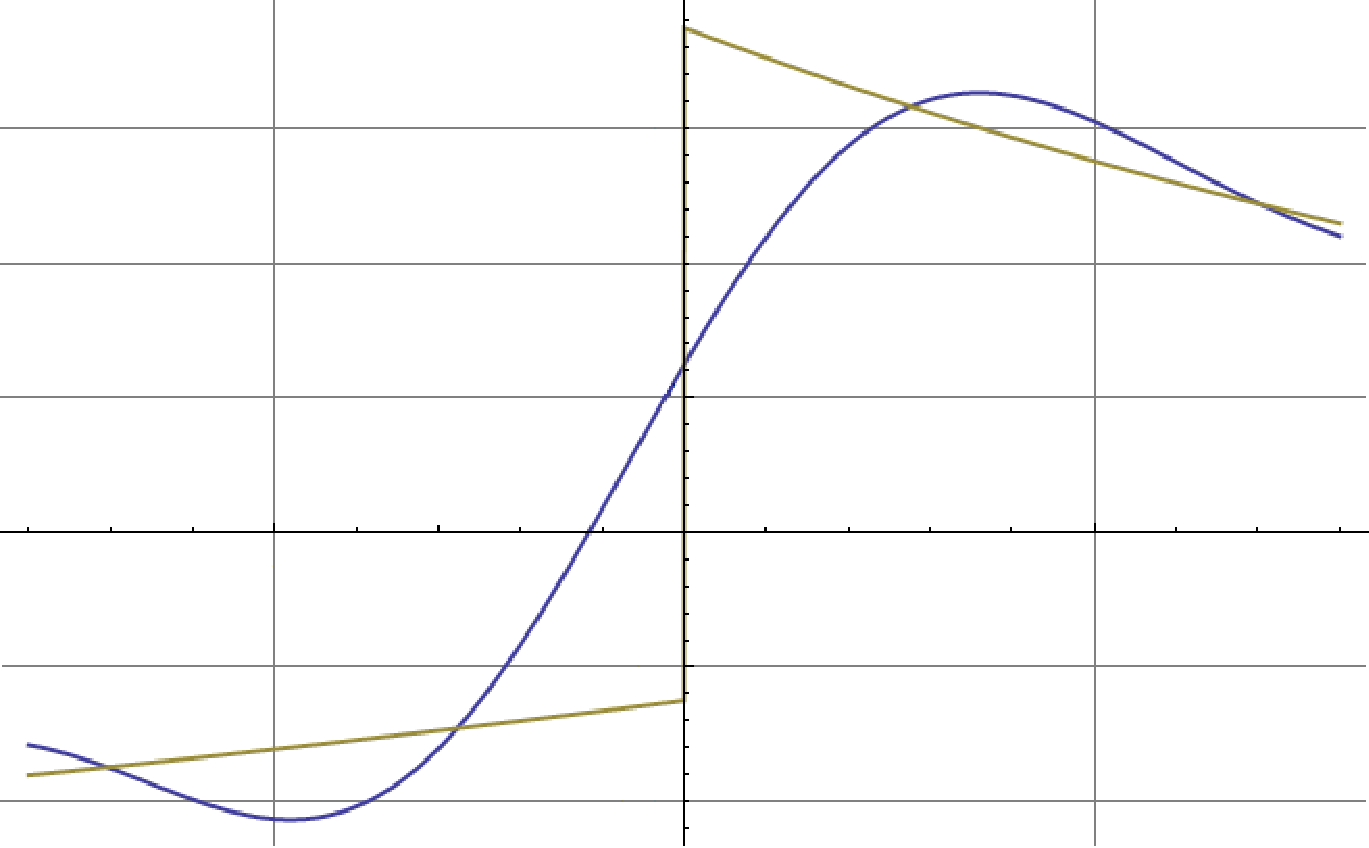}\caption{\label{fig:The-fist-derivative-and-second-derivatives}The fist derivative
(left) and the second derivative (right) of the wave function $\psi$
in an infinitesimal neighborhood of $0$ (blue lines). Green lines
are the same derivatives for $x\ll0$ and $x\gg0$.}
\end{figure}

\subsection{Heisenberg uncertainty principle}

We close this section of applications by mentioning how we can use
infinitesimal and infinite numbers and GSF theory to fully justify
the most frequent example of the Heisenberg uncertainty principle.

We do not have sufficient space here to present a complete list of
result about the so called \emph{hyperfinite Fourier transform} (see
\cite{MT-Gio22}), but we can surely present the main ideas:
\begin{enumerate}
\item Since in the ring $\rti$ we have infinite numbers $k\in\rti_{>0}$,
e.g.~$k=\diff\rho^{-1}$, and since every GSF is always integrable
on a functionally compact set of the form $K:=[-k,k]^{n}\subseteq\rti^{n}$
(see Thm.~\ref{thm:muMeasurableAndIntegral}), we can simply define
the \emph{hyperfinite Fourier transform} $\mathcal{F}_{k}\left(f\right)$
of any $f\in\gsf(K,\rccrho)$ as
\begin{equation}
\mathcal{F}_{k}\left(f\right)\left(\omega\right):=\intop_{K}f\left(x\right)e^{-ix\cdotp\omega}\,\diff{x}=\intop_{-k}^{k}\,\diff{x_{1}}\ldots\intop_{-k}^{k}f\left(x_{1},\ldots,x_{n}\right)e^{-ix\cdotp\omega}\,\diff{x_{n}}.\label{eq:Hpfinite_FT}
\end{equation}
\item The main feature of this transform is that, despite the fact that
essentially all the usual classical properties of the Fourier transform
can be proved for $\mathcal{F}_{k}\left(f\right)$, it is well-defined
for \emph{all} GSF $f\in\gsf(K,\rccrho)$, even if they are not of
tempered type. Clearly, this allows one to use the Fourier method
to find non-tempered solutions of differential equations. For example,
let $f\left(x\right)=e^{x}$ for all $\left|x\right|\leq k$, where
$k:=-\log\left(\diff\rho\right)$. The hyperfinite Fourier transform
$\mathcal{F}_{k}$ of $f$ is
\begin{align*}
\mathcal{F}_{k}\left(f\right)\left(\omega\right) & =\frac{1}{1-i\omega}\left(\frac{\diff{\rho}^{i\omega}}{\diff{\rho}}-\frac{\diff{\rho}}{\diff{\rho}^{i\omega}}\right)\quad\forall\omega\in\rti.
\end{align*}
Therefore, $\mathcal{F}_{k}(f)(\omega)$ is always an infinite complex
number for all finite numbers $\omega$ and hence the non-Archimedean
language is essential here.
\item The set $\text{supp}\left(f\right):=\overline{\left\{ x\in X\mid\left|f\left(x\right)\right|>0\right\} }$,
where $\overline{\left(\cdot\right)}$ denotes the relative closure
in $X$ with respect to the sharp topology, is called the \emph{support}
of $f$. Let $H\fcmp\rti^{n}$ be a functionally compact set (see
Def.~\ref{def:functCmpt-1}), we say that $f\in\Dgsf\left(H\right)$
if $f\in\gsf(\rti^{n},\rccrho)$ and $\text{supp}\left(f\right)\subseteq H$.
Such an $f$ is called \emph{compactly supported}.
\end{enumerate}
We can now state the uncertainty principle (see \cite{MTAG} for the
proof):
\begin{thm}
\label{thm:uncertainty}If $\psi\in\Dgsf(\rti)$, then $\left(\intop x^{2}\left|\psi\left(x\right)\right|^{2}\,\diff{x}\right)\left(\intop\omega^{2}\left|\mathcal{F}\left(\psi\right)\left(\omega\right)\right|^{2}\,\diff{\omega}\right)\geq\frac{1}{4}\Vert\psi\Vert_{2}\Vert\mathcal{F}(\psi)\Vert_{2}$.
\end{thm}

On the contrary with respect the classical formulation in $L^{2}(\R)$
of the uncertainty principle, in Thm.~\ref{thm:uncertainty} we can
e.g.~consider $\psi=\delta\in\Dgsf(\rti)$, and we have
\[
\int x^{2}\delta(x)^{2}\,\diff{x}=\left[\int_{-1}^{1}x^{2}b_{\eps}^{2}\mu_{\eps}(b_{\eps}x)^{2}\,\diff{x}\right]
\]
where $\mu(x)=[\mu_{\eps}(x_{\eps})]$ is a Colombeau mollifier and
$b=[b_{\eps}]\in\rti$ satisfies $b\ge\diff{\rho}^{-a}$ for some
$a\in\R_{>0}$ (see embedding Thm.~\ref{thm:embeddingD'}). Since
normalizing the function $\eps\mapsto b_{\eps}^{2}\mu_{\eps}(b_{\eps}x)^{2}$
we get an approximate identity, we have $\lim_{\eps\to0^{+}}\int_{-1}^{1}x^{2}b_{\eps}^{2}\mu_{\eps}(b_{\eps}x)^{2}\,\diff{x}=0$,
and hence $\int x^{2}\delta(x)^{2}\,\diff{x}\approx0$ is an infinitesimal.
The uncertainty principle Thm.~\ref{thm:uncertainty} implies that
it is an invertible infinitesimal. Considering the HFT $\mathbb{1}:=\mathcal{F}(\delta)$,
we have
\[
\int\omega^{2}\mathbb{1}(\omega)^{2}\,\diff{\omega}\ge\int_{-r}^{r}\omega^{2}\,\diff{\omega}=2\frac{r^{3}}{3}\quad\forall r\in\R_{>0}.
\]
Consequently, $\int\omega^{2}\mathbb{1}(\omega)^{2}\,\diff{\omega}$
is an infinite number.

\section{Conclusions}

In all the presented examples, the model describes some kind of singular
dynamical system including abrupt changes, impulsive stimuli, nonlinear
discontinuities, infinite barriers, etc. This kind of problems are
ubiquitous in applied mathematics, essentially because the real world
is made of different bodies, having boundaries and frequently interacting
in a non-smooth way. In constructing a model for these systems is
hence important to achieve mathematical simplicity but, at the same
time, a physical reasonably high faithfulness of description.

On the one hand, the use of infinitesimal and infinite numbers has
always been a method to simplify a given problem. Unfortunately, frequently
this technique remains only informal, using ``sufficiently small
quantities'' or ``taking the limit for $\eps\to0$'', and then
transforming approximated equalities into true ones. As motivational
thoughts, these remain wonderful methods. In our examples, we tried
to show that a corresponding simple and intuitively clear mathematical
theory of these infinitesimal and infinite quantities is possible.
Surprisingly, this theory allows one to arrive at very similar, but
clear and rigorous, thoughts. Therefore, the risk of doing mistakes
is quite lower, and its teaching is also way more clear.

On the other hand, physical systems with singularities are naturally
represented by non-smooth functions. We presented a theory that allows
one to deal with such functions in a way, as if they were smooth,
with a lot of properties that GSF share with ordinary smooth functions.
This is as generalized functions are still informally used in physics
and engineering, despite the fact that Schwartz theory of distributions
is quite old nowadays. Using GSF theory, we can therefore state that
the searched mathematical simplicity in models of singular systems,
possibly with a clear use of infinitesimal or infinite quantities,
is really achievable.

\subsubsection*{Declarations}

\subsubsection*{Author contributions}

A.~Bryzgalov developed Sec.~\ref{sec:Examples-of-applications}.

\noindent K.~Islami developed Sec.~\ref{sec:Formal-deductions}
and some examples in Sec.~\ref{sec:Examples-of-applications}.

\noindent P.~Giordano coordinated the development and formally checked
and fixed all the results.

\subsubsection*{Conflict of interest}

The authors declare that they have no relevant financial or non-financial
interests to disclose.

\subsubsection*{Ethical approval}

We certify that this manuscript is original and has not been published
and will not be submitted elsewhere for publication while being considered
by Nonlinear Dynamics. And the study is not split up into several
parts to increase the quantity of submissions and submitted to various
journals or to one journal over time. No data have been fabricated
or manipulated (including images) to support our conclusions. No data,
text, or theories by others are presented as if they were our own.
And authors whose names appear on the submission have contributed
sufficiently to the scientific work and therefore share collective
responsibility and accountability for the results. Finally, this article
does not contain any studies with human participants or animals performed
by any of the authors.

\end{document}